\RequirePackage[final]{graphicx}
\documentclass[UKenglish,cleveref, autoref]{lmcs}
\pdfoutput=1

\usepackage[utf8]{inputenc}

\usepackage{lastpage}
\lmcsdoi{17}{1}{20}
\lmcsheading{}{\pageref{LastPage}}{}{}%
{Dec.~16,~2019}{Mar.~04,~2021}{}

  \keywords{Session types, \and Concurrency, \and Subtyping, \and
    Algorithm.}
\usepackage{booktabs}
\usepackage{graphicx}
\usepackage{amsmath}
\usepackage{amssymb} 
\usepackage{color}
\usepackage{mathtools}
\usepackage{proof-dashed}
\usepackage{mathpartir}
\usepackage{extarrows}
\usepackage{cite}
\usepackage[draft,nomargin,inline]{fixme} \usepackage{xcolor} 
\usepackage{tikz}
\usetikzlibrary{arrows,shapes,automata,spy,positioning,
  calc,shadows,fit,arrows.meta,patterns,backgrounds,decorations.pathreplacing}
\usepackage{tikz-qtree}
\usepackage{tabularx}
\usepackage{cmll} \usepackage{thm-restate}

\fxsetup{theme=color,mode=multiuser}
\FXRegisterAuthor{MB}{aMB}{MB}
\FXRegisterAuthor{MC}{aMC}{MC}
\FXRegisterAuthor{JL}{aJL}{JL}
\FXRegisterAuthor{NY}{aNY}{NY}
\FXRegisterAuthor{GZ}{aGZ}{GZ}

\DeclareMathAlphabet{\mathpzc}{OT1}{pzc}{m}{it}

 \newcommand{\MC}{-\textsc{mc}}

\newcommand{\eqdef}{\stackrel{\text{def}}{=}}

\newcommand{\img}[1]{\mathit{img}(#1)}

\newcommand{\qst}{\, \mid \,}
\newcommand{\dst}{.\ }
\newcommand{\eqdefi}{\overset{{\text{def}}}{=}}

\newcommand{\farrow}{\xmapsto{\quad}}

\newcommand{\lempty}{\epsilon}

\newcommand{\naturals}{\mathbb{N}}

\newcommand{\dual}[1]{\overline{#1}}

\newcommand{\Alpha}{\mathbb{A}}
\newcommand{\Alphas}{\Alpha^{\ast}}
\newcommand{\word}{\omega}
\newcommand{\worda}{\word}
\newcommand{\wordb}{\omega'}
\newcommand{\Act}{\mathit{Act}}
\newcommand{\Acts}{\Act^{\ast}}
\newcommand{\act}{\ell}

\newcommand{\acts}{\psi}
\newcommand{\actsa}{\acts}
\newcommand{\actsb}{\varphi}
\newcommand{\Actsnd}{\Act_!}
\newcommand{\Actrcv}{\Act_?}

\newcommand{\sends}[1]{\mathsf{snd}(#1)}
\newcommand{\receives}[1]{\mathsf{rcv}(#1)}

\newcommand{\dir}[1]{\mathit{dir}(#1)}
\newcommand{\msg}[1]{\mathit{#1}}
\newcommand{\snd}[1]{!\msg{#1}}
\newcommand{\rcv}[1]{?\msg{#1}}
\newcommand{\trans}[1]{\xrightarrow{#1}}
\newcommand{\transC}{\xrightarrow{}^{\ast}}

\newcommand{\ntrans}{\nrightarrow}
\newcommand{\simtreetrans}[1]{\xhookrightarrow{#1}}

\newcommand{\States}{Q}

\newcommand{\state}{q}
\newcommand{\statea}{p}
\newcommand{\stateb}{\state}
\newcommand{\Trans}{\delta}

\newcommand{\sndL}[1]{! #1}
\newcommand{\rcvL}[1]{? #1}

\DeclareMathOperator{\nincsymbol}{\preccurlyeq}
\DeclareMathOperator{\nnincsymbol}{\not\preccurlyeq}

 \newcommand{\minc}[2]{ #1 \nincsymbol #2}

\newcommand{\var}[1]{\mathbf{#1}}

\newcommand{\loopio}[2]{\mathsf{cycle}(#1,#2)}
\newcommand{\acctree}[2]{\mathsf{accTree}(#1,#2)}

\newcommand{\closesthole}[3]{\mathsf{minAcc}(#1,#2,#3)}

\newcommand{\node}[1]{\mathtt{#1}}

\newcommand{\mypar}[1]{\noindent{\bf #1}}
\newcommand{\simtree}[2]{\mathsf{simtree}(#1,#2)}
\newcommand{\simpair}[2]{#1 \nincsymbol #2}
\newcommand{\simtreepair}[2]{#1 \nincsymbol #2}
\newcommand{\ctx}[1]{\mathcal{#1}}

\newcommand{\echoice}[2]{\langle #1 \rangle_{#2}}
\newcommand{\schoice}[2]{\langle #1 \rangle_{#2}}
\newcommand{\echset}[2]{\langle #1 \qst #2 \rangle}
\newcommand{\schset}[2]{\langle #1 \qst #2 \rangle}

\newcommand{\inTrans}[1]{\mathsf{in}(#1)}
\newcommand{\outTrans}[1]{\mathsf{out}(#1)}
\newcommand{\intree}[1]{\mathsf{inTree}(#1)}
\newcommand{\labfun}{\mathcal{L}}
\newcommand{\ctxtQ}{\mathcal{T}_Q}
\newcommand{\ctxtQtwo}{\mathcal{T}_{Q_2}}
\newcommand{\minHeight}[1]{\mathsf{minHeight}(#1)}
\newcommand{\height}[2]{\mathsf{height}_{#2}(#1)}
\newcommand{\extract}[2]{\mathsf{extract}(#1,#2)}
\newcommand{\DEF}{\mathbf{E}}
\newcommand{\anc}{\mathsf{anc}}
\newcommand{\ancestor}[1]{\anc(#1)}
\newcommand{\subtree}[3]{\mathsf{nodes}({#1,#2,#3})}

\newcommand{\exHQ}{\msg{nd}}
\newcommand{\exLQ}{\msg{pr}}
\newcommand{\exOK}{\msg{ok}}
\newcommand{\exKO}{\msg{ko}} 
\newcommand{\client}{M_{R}}
\newcommand{\clientsup}{M_C}
\newcommand{\server}{M_S}

\tikzset{
every state/.style={minimum size=1pt,inner sep=1.5pt, initial text={}},
  mycfsm/.style={
    font=\scriptsize,
    initial where=left,
    initial distance=0.25cm,
    ->,>=stealth,auto, node distance=0.8cm and 0.8cm,
    scale=1, every node/.style={transform shape},
    baseline=(current  bounding  box.center)
  },
  ogate/.style = {
    diamond, draw, fill=white,
    minimum size=4mm,
    inner sep=0pt,
    postaction={path picture={\draw[black]
        ([yshift=\gatedistancein]path picture bounding box.south) -- ([yshift=-\gatedistancein]path picture bounding box.north)
        ([xshift=-\gatedistancein]path picture bounding box.east) -- ([xshift=\gatedistancein]path picture bounding box.west)
        ;}}, drop shadow},
  agate/.style={draw,rectangle,
    minimum size=3mm,
    inner sep=0pt,
    fill=white,
    postaction={path picture={\draw[black]
        ([yshift=\gatedistanceinand]path picture bounding box.south) --
        ([yshift=-\gatedistanceinand]path picture bounding box.north) ;}}, drop shadow},
  source/.style={draw,circle,fill=white,
    minimum size=3mm,
    inner sep=0pt, drop shadow},
  sink/.style={draw,circle,double,fill=white,
    minimum size=3mm,
    inner sep=0pt, drop shadow},
  intera/.style = {rectangle, draw=black, align=center, fill=white, rounded corners=0.1cm,
    minimum height=12,
    inner sep=2pt, drop shadow},
  line/.style = {draw,->, rounded corners=0.07cm,>=latex},
  venn/.style={preaction={fill, #1},opacity=0.6},
  cnode/.style={rectangle,draw=black,inner sep=2pt},
  ancestor/.style={densely dashed,->, shorten >=1pt},
  hookedsilentedge/.style={>=latex,right hook->,shorten <=1pt,shorten >=1pt},
  lhookedsilentedge/.style={>=latex,left hook->,shorten <=1pt,shorten >=1pt},
  silentedge/.style={>=latex,->},
  nlabel/.style={fill=white,inner sep=0pt,font=\footnotesize},
  notexplo/.style={fill=gray!10},
  echnode/.style={rectangle,draw=black,inner sep=2pt},
  schnode/.style={diamond,draw=black,inner sep=0pt},
}

\newcommand{\varX}{\var{x}}
\newcommand{\varY}{\var{y}}
\DeclareMathOperator{\mysum}{\with}
\newcommand{\inchoicetop}{\oplus} \newcommand{\outchoicetop}{\mysum}

\newcommand{\rec}[1]{\mathtt{\mu}\, \var{#1} . }

\newcommand{\inference}[3]{\infer[\ifthenelse{\equal{#1}{}}{}{\inferrule{#1}}]{#3}{#2}}
\newcommand{\coinference}[3]{\infer=[\ifthenelse{\equal{#1}{}}{}{\inferrule{#1}}]{#3}{#2}}

\newcommand{\leftcolor}[1]{{\color{blue} #1}}
\newcommand{\rightcolor}[1]{{\color{red} #1}}

\newcommand{\pairconf}[2]{\leftcolor{#1} \nincsymbol \rightcolor{#2}}

\newcommand{\treeconf}[2]{
  \begin{tikzpicture}[node distance=0cm and -0.1cm]
    \node (right) {#2};
    \node[left=of right] (symbol) {$\nincsymbol$};
    \node[left=of symbol] (left)  {\leftcolor{#1}};
  \end{tikzpicture}
}

\newcommand{\aroundtree}[1]{
  \begin{tikzpicture}
    \tikzset{edge from parent/.style=
      {draw,
        edge from parent path={(\tikzparentnode.south)
          -| (\tikzchildnode)}},
      level 1/.style={level distance=13pt},
      level 2+/.style={level distance=13pt},
}
{#1}
  \end{tikzpicture}
}

\newcommand{\treeleaf}{\rightcolor{$q_2$}}

\newcommand{\treedone}{
  \aroundtree{
    \Tree 
    [
    .\,
    [.$\exKO$ \treeleaf{} ] 
    [.$\exOK$ \treeleaf{} ] 
    ]
  }
}

\newcommand{\treedtwo}{
  \aroundtree{
    \Tree 
    [
    .\,
    [.$\exKO$     
    [.$\exKO$ \treeleaf{} ] 
    [.$\exOK$ \treeleaf{} ]  
    ] 
    [.$\exOK$    
    [.$\exKO$ \treeleaf{} ] 
    [.$\exOK$ \treeleaf{} ]  
    ] 
    ]
  }
}

\newcommand{\treedthree}{
  \aroundtree{    
    \Tree 
    [
    .\,
    [.$\exKO$     
    [.$\exKO$ 
    [.$\exKO$ \treeleaf{} ] 
    [.$\exOK$ \treeleaf{} ] 
    ] 
    [.$\exOK$ 
    [.$\exKO$ \treeleaf{} ] 
    [.$\exOK$ \treeleaf{} ] 
    ]  
    ] 
    [.$\exOK$   
    [.$\exKO$ 
    [.$\exKO$ \treeleaf{} ] 
    [.$\exOK$ \treeleaf{} ] 
    ] 
    [.$\exOK$
    [.$\exKO$ \treeleaf{} ] 
    [.$\exOK$ \treeleaf{} ] 
    ]  
    ] 
    ]
  }
}

\newcommand{\btreedone}{
  \aroundtree{
    \Tree 
    [
    .\,
    [.$a$ \qoneleaf{} ] 
    [.$b$ \qthreeleaf{} ] 
    ]
  }
}

\newcommand{\btreetwo}{
  \aroundtree{
    \Tree 
    [
    .\,
    [.$a$
    [.$a$ \qoneleaf{} ] 
    [.$b$ \qthreeleaf{} ]  
    ] 
    [.$b$ \qthreeleaf{} ] 
    ]
  }
}

\newcommand{\btreedthree}{
  \aroundtree{
    \Tree 
    [
    .\,
    [.$a$
    [.$a$
    [.$a$ \qoneleaf{} ] 
    [.$b$ \qthreeleaf{} ]  
    ] 
    [.$b$ \qthreeleaf{} ]  
    ] 
    [.$b$ \qthreeleaf{} ] 
    ]
  }
}

\newcommand{\btreedfour}{
  \aroundtree{
    \Tree 
    [
    .\,
    [.$a$
    [.$a$
    [.$a$
    [.$a$ \qoneleaf{} ] 
    [.$b$ \qthreeleaf{} ]  
    ] 
    [.$b$ \qthreeleaf{} ]  
    ] 
    [.$b$ \qthreeleaf{} ]  
    ] 
    [.$b$ \qthreeleaf{} ] 
    ]
  }
}

\newcommand{\btreedfive}{
  \aroundtree{
    \Tree 
    [
    .\,
    [.$a$
    [.$a$
    [.$a$
    [.$a$
    [.$a$ \qoneleaf{} ] 
    [.$b$ \qthreeleaf{} ]  
    ] 
    [.$b$ \qthreeleaf{} ]  
    ] 
    [.$b$ \qthreeleaf{} ]  
    ] 
    [.$b$ \qthreeleaf{} ]  
    ] 
    [.$b$ \qthreeleaf{} ] 
    ]
  }
}

\newcommand{\qoneleaf}{\rightcolor{$q_1$}}
\newcommand{\qthreeleaf}{\rightcolor{$q_3$}}

\newcommand{\pvarroot}{X_{0}}
\newcommand{\pvar}[2]{X_{q_{#1},n_{#2}}}
\newcommand{\cvar}[1]{Y_{n_{#1}}}

\begin{document}
\title[Sound Algorithm for Asynchronous Session Subtyping]
      {A Sound Algorithm for Asynchronous Session Subtyping and its Implementation}

\author[M.~Bravetti]{Mario Bravetti\rsuper{a}}
\address{\lsuper{a}University of Bologna / INRIA FoCUS Team}
\email{\{mario.bravetti,gianluigi.zavattaro\}@unibo.it}{}{}

\author[M.~Carbone]{Marco Carbone\rsuper{b}}
\address{\lsuper{b}IT University of Copenhagen}
\email{carbonem@itu.dk}{}{}

\author[J.~Lange]{Julien Lange\rsuper{c}}
\address{\lsuper{c}Royal Holloway, University of London}
\email{julien.lange@rhul.ac.uk}{}{}

\author[N.~Yoshida]{Nobuko Yoshida\rsuper{d}}
\address{\lsuper{d}Imperial College London}
\email{n.yoshida@imperial.ac.uk}

\author[G.~Zavattaro]{Gianluigi Zavattaro\rsuper{a}}

\begin{abstract}
  Session types, types for structuring communication between endpoints
  in concurrent systems, are recently being integrated into mainstream
  programming languages.
In practice, a very important notion for dealing with such types is
  that of subtyping, since it allows for typing larger classes of
  systems, where a program has not precisely the expected behavior but
  a similar one. Unfortunately, recent work has shown that subtyping
  for session types in an asynchronous setting is undecidable. To cope
  with this negative result, the only approaches we are aware of
  either restrict the syntax of session types or limit communication
  (by considering forms of bounded asynchrony).
Both approaches are too restrictive in practice, hence we proceed
  differently by presenting an algorithm for checking subtyping which
  is sound, but not complete (in some cases it terminates without
  returning a decisive verdict).
The algorithm is based on a tree representation of the coinductive
  definition of asynchronous subtyping; this tree could be infinite,
  and the algorithm checks for the presence of finite witnesses of
  infinite successful subtrees.
Furthermore, we provide a tool that implements our algorithm.
We use this tool to test our algorithm on many examples that cannot
  be managed with the previous approaches, and to provide an empirical
  evaluation of the time and space cost of the algorithm.
\end{abstract}
\maketitle

\section{Introduction}\label{sec:introduction}
Session types are behavioural types that specify the structure of
communication between the endpoints of a system or the processes of a
concurrent program.
In recent years, session types have been integrated into several
mainstream programming languages (see,
e.g.,~\cite{HuY16,Padovani17,SY2016,LindleyM16,OrchardY16,AnconaEtAl16,NHYA2018})
where they specify the pattern of interactions that each endpoint must
follow, i.e., a communication protocol.
The notion of duality is at the core of theories based on session
types, where it guarantees that each send (resp.\ receive) action is
matched by a corresponding receive (resp.\ send) action, and thus
rules out deadlocks \cite{BoerBLZ18} and orphan messages.
A two-party communication protocol specified as a pair of session
types is ``correct'' (deadlock free, etc) when these types are dual of
each other.
Unfortunately, in practice, duality is a too strict prerequisite,
since it does not provide programmers with the flexibility necessary
to build practical implementations of a given protocol.
A natural solution for relaxing this rigid constraint is to adopt a
notion of (session) subtyping which lets programmers implement
refinements of the specification (given as a session type).  In
particular, an endpoint implemented as program $P_2$ with type $M_2$
can always be safely replaced by another program $P_1$ with type $M_1$
whenever $M_1$ is a subtype of $M_2$ (written $M_1 \nincsymbol M_2$ in
this paper).

The two main known notions of subtyping for session types differ in
the type of communication they support: either synchronous
(rendez-vous) or asynchronous (over unbounded FIFO channels).
\emph{Synchronous session subtyping} checks, by means of a so-called subtyping simulation game,
that the subtype implements fewer internal choices (sends), 
and more external choices (receives), than its supertype. 
Hence checking whether two types are
related can be done efficiently (quadratic time wrt.\ the size of the
types~\cite{LangeY16}).
Synchronous session subtyping is of limited interest in modern
programming languages such as Go and Rust, which provide
\emph{asynchronous} communication over channels. 
Indeed, in an asynchronous setting, the programmer needs to be able to
make the best of the flexibility given by non-blocking send actions.
This is precisely what the \emph{asynchronous session subtyping}
offers: it widens the synchronous subtyping relation by allowing the
subtype to anticipate send (output) actions, when this does not affect
its communication partner, i.e., it will notably execute all required
receive (input) actions later.
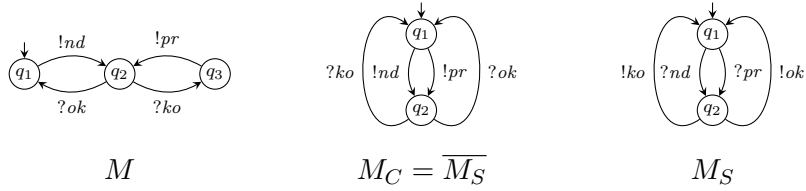
\begin{figure}[t]
  \centering
  \begin{tabular}{c@{\qquad\quad}c@{\qquad\quad}c}
    \begin{tikzpicture}[mycfsm, node distance = 1cm and 1cm
      ,scale=0.85, every node/.style={transform shape}]
      \node[state, initial, initial where=above] (s1) {$q_1$};
      \node[state, right =of s1] (s2) {$q_2$};
      \node[state, right=of s2] (s3) {$q_3$};
\path
      (s1) edge [bend left] node [above] {$\snd{\exHQ}$} (s2)
      (s2) edge [bend left] node {$\rcv{\exOK}$} (s1)
      (s2) edge [bend right] node [below] {$\rcv{\exKO}$} (s3)
      (s3) edge [bend right] node [above] {$\snd{\exLQ}$} (s2)
      ;
    \end{tikzpicture}
    &
      \begin{tikzpicture}[mycfsm, node distance = 0.7cm and 1cm
      ,scale=0.85, every node/.style={transform shape}]
        \node[state, initial, initial where=above] (s1) {$q_1$};
        \node[state, below =of s1] (s2) {$q_2$};
\path 
        (s1) edge [bend right=25] node [left] {$\snd{\exHQ}$} (s2)
        (s1) edge [bend left=25] node [right] {$\snd{\exLQ}$} (s2)
        (s2) edge [bend right=125,looseness=2] node [right] {$\rcv{\exOK}$} (s1)
        (s2) edge [bend left=125,looseness=2] node [left] {$\rcv{\exKO}$} (s1)
        ;
      \end{tikzpicture}
    &
      \begin{tikzpicture}[mycfsm, node distance = 0.7cm and 1cm
        ,scale=0.85, every node/.style={transform shape}]
        \node[state, initial, initial where=above] (s1) {$q_1$};
        \node[state, below =of s1] (s2) {$q_2$};
\path
        (s1) edge [bend right=25] node [left] {$\rcv{\exHQ}$} (s2)
        (s1) edge [bend left=25] node [right] {$\rcv{\exLQ}$} (s2)
        (s2) edge [bend right=125,looseness=2] node [right] {$\snd{\exOK}$} (s1)
        (s2) edge [bend left=125,looseness=2] node [left] {$\snd{\exKO}$} (s1)
        ;
      \end{tikzpicture}

    \\[-3pt]
    $M$  & $\clientsup = \dual{\server}$ & $\server $  
  \end{tabular}
  \caption{Hospital Service example. $M$ is the (refined) session
    type of the client, $\clientsup$ is a supertype of the client $M$, and
    $\server$ is the session type of the
    server.}\label{fig:runex-machines}
\end{figure}

\begin{figure}[t]
  \centering
  \begin{tabular}{c}
    \begin{tikzpicture}[mycfsm, node distance = 0.7cm and 1.3cm
      ,scale=0.85, every node/.style={transform shape}]
      \node[state, initial, initial where=above] (s1) {$q_1$};
      \node[state, below =of s1] (s2) {$q_2$};
      \node[state, right=of s2] (s3) {$q_3$};
\path 
      (s1) edge [bend left] node {$\snd{\exHQ}$} (s2)
      (s2) edge [bend left] node {$\rcv{\exOK}$} (s1)
      (s2) edge node [below] {$\rcv{\exKO}$} (s3)
      (s3) edge [bend right] node [above] {$\snd{\exLQ}$} (s1)
      ;
    \end{tikzpicture}
    \\[1pt]
    $\client$ \end{tabular}
  \caption{Refined Hospital Service client.  $\client$ is an
    asynchronous subtype of $\clientsup$, i.e., a refined session type of the
    Hospital Service client.}\label{fig:runex-machines2}
\end{figure}
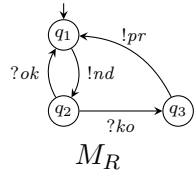

 We illustrate the salient points of the asynchronous session subtyping
with Figures~\ref{fig:runex-machines} and~\ref{fig:runex-machines2}, which
depict the hypothetical session types of the client and server
endpoints of a Hospital Service, represented as communicating machines
--- an equivalent formalism~\cite{cfsm83, DenielouY12}, see
Figure~\ref{fig:running-sessiontypes}.
Let us consider Figure~\ref{fig:runex-machines} first.
Machine $\server$ (right) is a server which can deal with two types of
requests: it can receive either a message $\exHQ$ (next patient data) or a
message $\exLQ$ (patient report).
After receiving a message of either type, the server replies with
$\exOK$ or $\exKO$, indicating whether the evaluation of received data
was successful or not, then it returns to its starting state.
Machine $\clientsup$ (middle) represents the type of the client. It is
the \emph{dual} of the server $\server$ (written $\dual{\server}$), as
required in standard two-party session types without subtyping.
A programmer may want to implement a slightly improved program which
behaves as Machine $M$ (left).
This version starts by sending $\exHQ$, then keeps sending patient
reports ($\exLQ$) until the previously sent data are deemed
satisfactory (it receives $\exOK$).
In fact, machine $M$ is a \emph{synchronous} subtype of machine
$\clientsup$, because of the covariance of outputs, i.e., $M$ is a
subtype of $\clientsup$, hence it can send fewer messages. Note that
$M$ can receive the same messages as $\clientsup$.
Machine $\client$ in Figure~\ref{fig:runex-machines2} is another
refinement of Machine $\clientsup$, but $\client$ is not a synchronous
subtype of $\clientsup$.
Instead, $\client$ is an \emph{asynchronous} subtype of
$\clientsup$.
Indeed, $\client$ is able to receive the same set of messages as
$\clientsup$, each of the sent messages are also allowed by
$\clientsup$, and the system consisting of the parallel composition of
machines $\client$ and $\server$ communicating via unbounded FIFO channels
is free from deadlocks and orphan messages.
We will use this example ($\client \nincsymbol \clientsup$) in the
rest of the paper to illustrate our theory.
Figure~\ref{fig:running-sessiontypes} gives the session types
corresponding to the machines in Figures~\ref{fig:runex-machines}
and~\ref{fig:runex-machines2}, where $\outchoicetop$ indicates an
external choice and $\inchoicetop$ indicates an internal choice.

Recently, we have proven that checking whether two types are in the
asynchronous subtyping relation is, unfortunately,
\emph{undecidable}~\cite{BravettiZ20,BCZ18,LangeY17,BravettiCZ17}.
In order to mitigate this negative result, some theoretical algorithms
have been proposed for restricted subclasses of session types.
These restrictions can be divided into two main categories:
syntactical restrictions, i.e., allowing only one type of non-unary
branching (internal or external choice), or adding bounds on the
number of pending messages in FIFO communication channels.
Both types of restrictions are problematic in practice.
Syntactic restrictions disallow protocols featuring both types of
internal/external choices, e.g., the machines $\clientsup$ and
$\server$ in Figure~\ref{fig:runex-machines} contain (non-unary)
external and internal choices.
On the other hand, applying a bound to the subtyping relation is
generally difficult because ($i$) it is generally undecidable whether
such a bound exists, ($ii$) the channel bounds used in the
implementation (if any) might not be known at compile time, and
($iii$) very simple systems, such as the one consisting of the
parallel composition of machines $\client$ and $\server$ discussed
above, require unbounded communication channels.

{\em 
The main contribution of this paper is to give a sound algorithm
for checking asynchronous session subtyping that does not impose 
syntactical restrictions nor bounds as done in previous works.}

\smallskip
\mypar{Overview of our approach.} Our approach will allow to
algorithmically check the subtyping between session types like
$\client$ and $\clientsup$.
In a nutshell, our algorithm proceeds as follows. We play the
classical subtyping simulation game with the subtype and supertype
candidates.
The game terminates when we encounter a {\em failure}, meaning that
the two types are not in the subtyping relation, or when we detect a
{\em repetitive} behaviour in the game.
In the latter case, we check whether this repetitive behaviour (which
can always be found) satisfies sufficient conditions that guarantee
that the subtyping simulation game will never encounter failures. If
the conditions are satisfied the algorithm concludes that the two
types are in the subtyping relation, otherwise no final verdict is
returned.

More precisely, session subtyping is defined following a coinductive
approach (Definition~\ref{def:subtyping}) that formalises a check on
the types that can be intuitively seen as a game.
At each step of the game, the candidate subtype proposes a challenge
(either an input or an output action to be executed) and the candidate 
supertype is expected to reply by performing a corresponding action.
The game ends in two possible ways: either both types terminate by
reaching their end state ({\em success}) or the candidate supertype is
unable to reply to the challenge ({\em failure}).
In case of failure, the two types are not in the subtyping relation,
otherwise they are.  This game is the so-called {\em subtyping simulation game}, and we
formally represent it as a {\em simulation tree}
(Definition~\ref{def:simtree}).
Hence two types are in the subtying relation if and only if their
simulation tree does not reach a failure (Theorem
\ref{thm:equivalence}).

Recall that asynchronous session subtyping allows the subtype to
anticipate output actions wrt.\ the supertype.
Hence, during the subtyping simulation game, a supertype can
reply to an output challenge by considering outputs that are not
immediately available, but are guarded by inputs.
These inputs cannot be forgotten during the game, because they could
be necessary to reply to subsequent input challenges. Thus, they are
recorded in so-called {\em input trees} (Definition
\ref{def:inputTree}).
Due to outputs inside loops, we can accumulate an unbounded
amount of inputs, thus generating input trees of unbounded depth.
For this reason, it is generally not possible to algorithmically
compute the entire simulation tree. To overcome this problem, we
propose a termination condition that intuitively says that the
computation of the simulation tree can be stopped when we reach a
point in the game that precisely corresponds to a previous point, or
differs simply because ``more'' inputs have been accumulated (Theorem
\ref{thm:termination}).

Using this termination condition, we compute a finite prefix of the
simulation tree. Given this finite tree, our algorithm proceeds as
follows: ($i$) it extracts special subtrees, called {\em candidate
  subtrees}, from the tree (Definition \ref{def:candidate}), and then
($ii$) checks whether all these subtrees satisfy certain properties
guaranteeing that, even if we have stopped the game, it would
certainly continue without reaching a failure. This is guaranteed if we have stopped the computation of the simulation tree
by reaching an already considered point, because subsequent
continuations of the game will continue repeating the exact same
steps.  In contrast, if we have stopped with ``more'' inputs, we must
have the guarantee that all possible continuations of the simulation
game cannot be negatively affected by these additional input
accumulations.
We formalise a sufficient condition on candidate subtrees (that
are named {\em witness trees} when they satisfy such a condition, see
Definition~\ref{def:witness}) that provides such a guarantee.

Concretely we use input tree equations (a sort of context-free tree
grammar, see Definition~\ref{def:equations}) to finitely represent
both the possible inputs of the candidate subtype and the inputs that
can be accumulated by the candidate supertype.
We then define a {\em compatibility} relation on input
tree equations, see Definition~\ref{def:compatibility}.
In a witness tree we impose that the input tree equations of the inputs
accumulated by the candidate supertype are compatible with those of the 
candidate subtype.
This implies that the candidate supertype will be always
ready to reply to all possible input challenges of the candidate 
subtype, simply by considering already
accumulated inputs (see our main Theorem \ref{thm:soundess}).
If all the candidate subtrees satisfy our sufficient conditions we can
conclude that the two initial session types are in the subtyping
relation, otherwise the algorithm replies with ``I don't know'' meaning
that it is not possible to conclude with a final verdict.

\begin{figure}
\[  
  \begin{array}{lcl}
    M  & = & \rec{x}
             \inchoicetop
                  \Big\{
                  \exHQ\!: \ \rec{y} \outchoicetop \big\{\
                             {ok}\!: \varX, \ \ {ko}\!:
                               \inchoicetop\{{\exLQ}\!: \varY \}
                               \ \big\}
                  \ \Big\}
    \\[2mm]
    \client & = & \rec{x}
             \inchoicetop
                  \Big\{
                  \exHQ\!: \outchoicetop \big\{\
                             {ok}\!: \varX, \ \ {ko}\!:
                               \inchoicetop\{{\exLQ}\!: \varX \}
                               \ \big\}
                  \ \Big\}
    \\[2mm]
    \clientsup & = & \rec{x}
             \inchoicetop
                  \big\{
                  \exHQ\!: \outchoicetop \{\
                     {ok}\!: \varX, \ \
                     {ko}\!: \varX
                               \ \},\quad 
                  \exLQ\!: \outchoicetop \{\
                     {ok}\!: \varX, \ \
                     {ko}\!: \varX
                               \ \}
                  \ \big\}
    \\[1mm]
    \server & = & \rec{x}
                  \outchoicetop
                  \big\{
                  {\exHQ}\!: \inchoicetop \{{ok}\!: \varX, \ {ko}\!: \varX \}, \quad
                  {\exLQ}\!: \inchoicetop \{{ok}\!: \varX, \ {ko}\!: \varX \}
                  \big\}
  \end{array}
\]
\caption{Session types corresponding to the machines in
  Figures~\ref{fig:runex-machines} and~\ref{fig:runex-machines2}.}
\label{fig:running-sessiontypes}
\end{figure}

\subsection{Structure of the paper}
The remainder of the paper is structured as follows.
\S~\ref{sec:preliminaries} reports some preliminary definitions,
namely the formalisation of session types as communicating machines
and the definition of asynchronous session subtyping.
Our approach for a sound algorithmic characterisation
of asynchronous session subtyping is presented in
\S~\ref{sec:algorithm}.
We also discuss in \S~\ref{sec:tool} a full implementation of our algorithm;
this has been used to 
test our approach on many examples that cannot be
managed with the previous approaches, and to provide an empirical
evaluation of the time and space cost of the algorithm.
Finally, the paper includes a discussion about related work in
\S~\ref{sec:related} and some concluding remarks in \S~\ref{sec:conclusions}.

This article is a full version of~\cite{BravettiCLYZ19}, with improved
presentation, refined definitions, detailed proofs and additional
examples.
Moreover, this version presents an empirical evaluation of
our algorithm: we tested the implementation of our algorithm on
automatically generated session types, see \S~\ref{sec:tool}. We have
also given an expanded discussion of related work and possible
extensions that can be addressed in the future, see
\S~\ref{sec:related} and \S~\ref{sec:conclusions}.

\section{Communicating Machines and Asynchronous Subtyping}\label{sec:preliminaries}
In this section,
we recall the definition of two-party communicating machines, that
communicate over unbounded FIFO channels (\S~\ref{sub:cfsms}), and
define asynchronous subtyping for session
types~\cite{LMCSasync,CDY2014}, which we adapt to communicating
machines, following~\cite{BCZ18} (\S~\ref{sub:asyncsub}).

\subsection{Communicating Machines}\label{sub:cfsms}
Let $\Alpha$ be a (finite) alphabet, ranged over by $a$, $b$, etc.
We let $\worda$, $\wordb$, etc.\ range over words in $\Alpha^{\ast}$.
The set of send (resp.\ receive) actions is
$\Actsnd = \{!\} \times \Alpha$, (resp.\
$\Actrcv = \{?\} \times \Alpha$).
The set of actions is $\Act = \Actsnd \cup \Actrcv$, ranged over by
$\act$, where a send action $\snd{a}$ puts message $a$ on an
(unbounded) buffer, while a receive action $\rcv{a}$ represents the
consumption of $\msg{a}$ from a buffer.
We define $\dir{\snd{a}} \eqdefi \, !$ and $\dir{\rcv{a}} \eqdefi \, ?$
and let $\actsa$ and $\actsb$ range over $\Acts$.
We write $\cdot$ for the concatenation operator on words and we write
$\epsilon$ for the empty word (overloaded for $\Alpha$ and
$\Alpha^{\ast}$).

In this work, we only consider communicating machines which correspond
to (two-party) session types. Hence, we focus on deterministic
(communicating) finite-state machines, without mixed states (i.e.,
states that can fire both send and receive actions) as
in~\cite{DenielouY12,DY13}.
\begin{defi}[Communicating Machine]\label{def:det-cfsm}
  A communicating machine $M$ is a tuple
  $(\States, \state_0, \Trans)$ where $\States$ is the (finite) set of
  states, $\state_0 \in \States$ is the initial state, and
  $\Trans \in \States \times \Act \times \States$ is a transition
  relation.
We further require that
  $\forall \state, \state' , \state'' \in \States \dst
  \forall \act, \act' \in \Act :$
\begin{enumerate}
  \item \label{it:directed}
$ (\state, \act, \state'),
    (\state, \act', \state'') \in \Trans$ implies $\dir{\act} = \dir{\act'}$,
and 
\item \label{it:deter}
    $
    (\state, \act, \state') , (\state, \act, \state'')
    \in \Trans$ implies $\state' = \state''
    $.
\end{enumerate}
We write $q \trans{\act} q'$ for $(q, \act, q') \in \delta$, omit
  unnecessary labels, and write $\transC$ for the reflexive transitive
  closure of $\trans{}$.
\end{defi}
Condition~\eqref{it:directed} requires all states to be directed, while
Condition~\eqref{it:deter} enforces determinism, i.e., all actions outgoing from a given state are pairwise distinct.

Given $M = (\States, \state_0, \Trans)$,
we say that $\state \in \States$ is \emph{final}, written $\state \ntrans$, iff
$\forall q' \in \States \dst \forall \act \in \Act \dst (q, \act, q') \notin \Trans$.
A state $\state \in \States$ is \emph{sending} (resp.\
\emph{receiving}) iff $\state$ is not final and $\forall q' \in
\States \dst \forall \act \in \Act \dst (q, \act, q') \in \Trans \dst
\dir{\act} = \, !$ (resp.\ $ \dir{\act} = \, ?$).
We use $\delta(\state,\act)$ to stand for $\state'$ such that
$(\state, \act, \state') \in \delta$.

We write $\state_0 \trans{\act_1 \cdots \act_k} \state_k$ iff there
are $\state_1, \ldots , \state_{k-1} \in \States$ such that
$\state_{i-1} \trans{\act_i} \state_i$ for $1 \leq i \leq k$.
Given a list of messages $\word = \msg{a}_1 \cdots \msg{a}_k$ ($k \geq
0$), we write $\rcvL{\word}$ for the list $\rcv{a}_1 \cdots \rcv{a}_k$
and $\sndL{\word}$ for $\snd{a}_1 \cdots \snd{a}_k$.

Given $\acts \in \Acts$ we define $\sends{\acts}$ and
$\receives{\acts}$:
\[
  \sends{\acts} =
  \begin{cases}
    \msg{a} \cdot \sends{\acts'}
    & \text{if } \acts = \snd{a} \cdot \acts'
    \\
    \sends{\acts'}
    & \text{if } \acts = \rcv{a} \cdot \acts'
    
    \\
    \epsilon & \text{if } \acts = \epsilon
  \end{cases}
  \qquad  \qquad
  \receives{\acts} =
  \begin{cases}
    \msg{a} \cdot \receives{\acts'}
    & \text{if } \acts = \rcv{a} \cdot \acts'
    \\
    \receives{\acts'}
    & \text{if } \acts = \snd{a} \cdot \acts'
    
    \\
    \epsilon & \text{if } \acts = \epsilon
  \end{cases}
\]
That is $\sends{\acts}$ (resp.\ $\receives{\acts}$) extracts the 
messages in send (resp.\ receive) actions from a sequence $\acts$.

\subsection{Asynchronous Session Subtyping}\label{sub:asyncsub}

\subsubsection{Input trees and contexts}
We define some structures and functions which we use to formalise the
subtyping relation.
In particular, we use syntactic constructs used to record the input
actions that have been anticipated by a candidate supertype, e.g.,
machine $M_2$ in Definition~\ref{def:inclusion}, as well as the local
states it may reach.
First, input trees (Definition~\ref{def:inputTree}) record input
actions in a standard tree structure.
\begin{defi}[Input Tree]\label{def:inputTree}
  An input tree is a term of the grammar:
  $$
  T\ \ ::=\ \ q \ \mid\ \echoice{a_i:T_i}{i\in I}
  $$
\end{defi}
In the sequel, we use $\ctxtQ$ to denote the input trees over states
$q \in Q$.
An input context is an input tree with ``holes'' in the place of
sub-terms.
\begin{defi}[Input Context]
  An input context is a term of $ {\mathcal A}\ \ ::=\ \ [\,]_j \ \mid\ \echoice{a_i:{\mathcal
      A}_i}{i\in I} $, where all indices $j$, denoted by
  $I(\mathcal A)$, are distinct and are associated to holes.

\end{defi}
For input trees and contexts of the form $\echoice{a_i:T_i}{i\in I}$
and $\echoice{a_i:{\mathcal A}_i}{i\in I}$, we assume that
$I \neq \emptyset$, $\forall i \neq j \in I \dst a_i \neq a_j$, and
that the order of the sub-terms is irrelevant.
When convenient, we use set-builder notation to construct input trees
or contexts, e.g., $\echset{a_i : T_i}{i \in I}$.

Given an input context ${\mathcal A}$ and an input context
${\mathcal A}_i$ for each $i$ in $I({\mathcal A})$,
we write ${\mathcal A}[{\mathcal A}_i]^{i \in I({\mathcal A})}$
for the input context obtained by replacing each hole $[\,]_i$ in
$\ctx A$ by the input context ${\mathcal A}_i$.
We write ${\mathcal A}[T_i]^{i \in I({\mathcal A})}$ for the input tree where
holes are replaced by input trees.

\subsubsection{Auxiliary functions}
In the rest of the paper we use the following auxiliary functions on
communicating machines.
Given a machine $M = (\States, \state_0, \Trans)$ and a state
$q \in \States$, we define:
\begin{itemize}
\item
  $\loopio{\star}{q} \iff \exists \word \in \Alphas, \word' \in
  \Alpha^{+}, q' \in Q \dst q \trans{\star {\word}} q' \trans{\star
    {\word'}} q'$ (with $\star \in \{!,?\}$),
\item
  $\inTrans{q}=\{a \mid \exists q'.q\trans{?a} q'\}$ and
  $\outTrans{q}=\{a \mid \exists q'.q\trans{!a} q'\}$,

\item let the \emph{partial} function $\intree{\cdot}$ be defined as:
\[
\intree{q} =
\begin{cases}
  \perp & \text{if }  \loopio{?}{q}
  \\
  q & \text{if } \inTrans{q} = \varnothing
  \\
  \echoice{a_i: \intree{\delta(q,?a_i)}}{i\in I} & \text{if }   \inTrans{q} = \{a_i \qst i \in I\} \neq  \varnothing
\end{cases}
\]
\end{itemize}
Predicate $\loopio{\star}{q}$ says that, from $q$, we
can reach a cycle with only sends (resp.\ receives), depending on whether
$\star=!$ or $\star=?$.
The function $\inTrans{q}$ (resp.\ $\outTrans{q}$) returns the
messages that can be received (resp.\ sent) from $q$.
When defined, $\intree{q}$ returns the tree containing all sequences
of messages which can be received from $q$ until a final or sending
state is reached.
Intuitively, $\intree{q}$ is undefined when $\loopio{?}{q}$ as it
would return an infinite tree.

\begin{exa}
  Given $\clientsup$ (Figure~\ref{fig:runex-machines}), we have the following:
\[
    \begin{array}{lllll}
      \inTrans{q_1} = \emptyset
      &
        \inTrans{q_2}=  \{ \exOK, \exKO \}
      \\
      \outTrans{q_1} = \{ \exHQ, \exLQ\}
      &
        \outTrans{q_2} = \emptyset
      \\
      \intree{q_1} = q_1
      &   \intree{q_2} = \echoice{\exOK : q_1 , \,
                          \exKO : q_1}{}
\end{array}
  \]
\end{exa}

\begin{exa}\label{example:newexample}
  Consider the following machine $M_1$: \\
  \begin{center}
    \begin{tikzpicture}[mycfsm, node distance = 0.4cm and 1cm
      ,scale=0.95, every node/.style={transform shape}]
      \node[state, initial, initial where=left] (s0) {$p_0$};
      \node[state, right=of s0] (s1) {$p_1$};
      \node[state, right=of s1] (s2) {$p_2$};
      \node[state, right=of s2] (s3) {$p_3$};
\path 
      (s0) edge node [above] {$\snd{b}$} (s1)
      (s1) edge node [above] {$\snd{c}$} (s2)
      (s2) edge node [above] {$\rcv{d}$} (s3)
      (s1) edge [loop below] node [below] {$\snd{a}$} (s1)
      ;
    \end{tikzpicture}
  \end{center}
  From state $p_0$ we can reach state $p_1$ with an output. The latter
  can loop into itself. Hence, we have both $\loopio{!}{p_0}$ and
  $\loopio{!}{p_1}$.
\end{exa}

\subsubsection{Asynchronous subtyping}
We present our definition of asynchronous subtyping (following the
orphan-message-free version from~\cite{CDY2014}). Our definition is a simple adaptation\footnote{In definitions for
  syntactical session types, e.g., \cite{MostrousIC15}, input contexts
  are used to accumulate inputs that precede anticipated outputs;
  here, having no specific syntax for inputs, we use input trees
  instead.}  of~\cite[Definition 2.4]{BCZ18} (given on syntactical
session types) to the setting of communicating machines.
\begin{defi}[Asynchronous Subtyping]\label{def:subtyping} \label{subtype} \label{def:inclusion}
Let $M_i = (\States_i, \state_{0_i}, \Trans_i)$ for $i \in \{1,2\}$.
  $\mathcal R$ is an asynchronous subtyping relation
  on $\States_1 \times \ctxtQtwo$ such that 
  $(p,T)\in\mathcal R$ implies:
  \begin{enumerate}
   \item \label{it:def-inc-final}
    if $\statea \ntrans$
    then $T=\stateb$ such that $\stateb \ntrans$;
\item \label{it:def-inc-tau}
    if  $\statea$ is a receiving state then 
    \begin{enumerate}
    \item  \label{it:def-inc-tau-sync}
      if $T=q$ then $q$ is a receiving state and

      \noindent
      $\forall q' \in Q_2 \ s.t.\ \ q \trans{?a} q'.\ \exists p' \in
      Q_1 \ s.t.\ p \trans{?a} p' \land (p',q')\in\mathcal R$;
    \item \label{it:def-inc-tau-async}
    if $T=\echoice{a_i:T_i}{i\in I}$
    then 
    $\forall i\in I.\
    \exists p' \in Q_1 \ s.t.\ p \trans{?a_i} p' \land  
    (p',T_i)\in\mathcal R$;
    \end{enumerate}
\item  \label{it:def-inc-alpha}
    if  $\statea$ is a sending state then 
    \begin{enumerate} 
    \item \label{it:def-inc-alpha-sync}
      if $T=q$ and $q$ is a sending state, then

      \noindent
      $\forall p' \in Q_1 \ s.t.\ \ p \trans{!a} p'.\ \exists q' \in
      Q_2 \ s.t.\ q \trans{!a} q' \land (p',q')\in\mathcal R$;
\item \label{it:def-inc-alpha-async}
    otherwise, if $T={\mathcal A}[q_i]^{i \in I}$ then
    $\neg\loopio{!}{\statea}$ and
    $\forall i\!\in\! I. \intree{q_i}\!=\!\ctx A_i[q_{i,h}]^{h\in H_i}\!$  and 
    $\forall p' \!\in\! Q_1 \ s.t.\ p \trans{!a} p'.$\\[.1cm]
\hspace*{1cm}
    $\forall i\!\in\! I. \forall h\! \in\! H_i.\ \exists q'_{i,h}\!\in\! Q_2\ s.t.\
    q_{i,h} \trans{!a} q'_{i,h} \land 
    (p',{\ctx A[\ctx A_i[q'_{i,h}]^{h\in H_i}]^{i\in I}}) 
    \in \mathcal R$.
\end{enumerate}   
  \end{enumerate}
  $M_1$ is an asynchronous subtype of $M_2$, written $\minc{M_1}{M_2}$, if there 
  is an asynchronous subtyping relation $\mathcal R$ such that $(\state_{0_1},\state_{0_2}) \in \mathcal R$.
\end{defi}

The relation $\minc{M_1}{M_2}$ checks that $M_1$ is a subtype of $M_2$
by executing $M_1$ and simulating its execution with $M_2$.
$M_1$ may fire send actions earlier than $M_2$, in which case $M_2$ is
allowed to fire these actions even if it needs to fire some receive actions first.
These receive actions are accumulated in an input context and
are expected to be subsequently matched by $M_1$.
Due to the presence of such an input context, the states reached by
$M_2$ during the computation are represented as input trees.
The definition first differentiates the type of state $\statea$:
\begin{description}\item [Final]
Case~\eqref{it:def-inc-final} says that if $M_1$ is in a final state, then
  $M_2$ is in a final state with an empty input context.
\item [Receiving]
Case~\eqref{it:def-inc-tau} says that if $M_1$ is in a receiving
  state, then either \eqref{it:def-inc-tau-sync} the input context is
  empty ($T = q$) and $M_1$ must be able to receive all messages that
  $M_2$ can receive;
or, \eqref{it:def-inc-tau-async} $M_1$ must be able to consume all
  the messages at the root of the input tree.
\item [Sending]
Case~\eqref{it:def-inc-alpha} applies when $M_1$ is in a sending
  state, there are two sub-cases.

  \noindent
  Case~\eqref{it:def-inc-alpha-sync} says that if the input context is
  empty ($T = q$) and $q$ is also a sending state, then $M_2$ must be 
  able to send all messages that $M_1$ can send.
If this sub-case above does not apply (i.e., the input context is not
  empty or $q$ is not a sending state), then
  the one below must hold.

  \noindent
  Case~\eqref{it:def-inc-alpha-async} enforces correct output
  anticipation, i.e., $M_2$ must be able to send every $\msg{a}$ that
  $M_1$ can send after some receive actions recorded in each
  $\ctx A_i[q_{i,h}]^{h\in H_i}$.
Note that whichever receiving path $M_2$ chooses, it must be able to send all possible 
output actions $\snd{a}$ of $M_1$, i.e., $\snd{a}$ should be available 
  at the end of each receiving path.
Moreover, given that there are accumulated inputs,
we require that $\loopio{!}{\statea}$ does \emph{not} hold,
  guaranteeing that subtyping preserves orphan-message freedom, i.e.,
  such accumulated receive actions will be eventually executed.
\end{description}
Observe that Case~\eqref{it:def-inc-tau} enforces a form of
contra-variance for receive actions, while
Case~\eqref{it:def-inc-alpha} enforces a form of covariance for send
actions.

\begin{exa}
  Consider $\clientsup$ and $\client$ from
  Figures~\ref{fig:runex-machines} and~\ref{fig:runex-machines2}, we
  have $\client \nincsymbol \clientsup$ (see \S~\ref{sec:algorithm}).
A fragment of the relation $\mathcal R$ from
  Definition~\ref{def:subtyping} is given in Figure~\ref{fig:simtree}.
Considering the identifier (bottom left) of each node in
  Figure~\ref{fig:simtree}, we have:
  \begin{itemize}
  \item Case~\eqref{it:def-inc-final} of
    Definition~\ref{def:subtyping} does not apply to any configuration
    in this example (there is no final node in these machines).
\item Case~\eqref{it:def-inc-tau-sync} applies to node $n_1$, i.e.,
    $q_2 \nincsymbol q_2$ (note that $q_2$ are receiving states in
    both machines).
\item Case~\eqref{it:def-inc-tau-async} applies to nodes $n_5$,
    $n_9$, and $n_{13}$; where $q_2$ of machine $\client$ is a
    receiving state and the input context is not empty.
\item Case~\eqref{it:def-inc-alpha-sync} applies to nodes $n_0$, $n_2$,
    and $n_3$, where the input context is empty and both states are
    sending states.
\item Case~\eqref{it:def-inc-alpha-async} applies to nodes $n_4$,
    $n_6$,
    $n_7$,
    $n_8$,
    $n_{10}$, $n_{11}$, $n_{12}$, $n_{14}$, $n_{15}$, and
    $n_{16}$.
Observe that this case does not require the input context to be
    non-empty (e.g., $n_4$), and 
that the condition $\neg\loopio{!}{p}$ holds for all states $p$ in
    $\client$ since there is no send-only cycle in this machine.
  \end{itemize}
\end{exa}

\begin{exa}
  For the two machines below, we have $M_1 \nnincsymbol M_2$ and $M_2 \nnincsymbol M_1$:
\begin{center}
  \begin{tabular}{l@{\quad}r@{\qquad\qquad}l@{\quad}r}
    $M_1$: & 
             \begin{tikzpicture}[mycfsm, node distance = 0.4cm and 1cm
               ,scale=0.95, every node/.style={transform shape}]
               \node[state, initial, initial where=left] (s1) {$p_1$};
               \node[state, right=of s1] (s2) {$p_2$};
               \node[state, right=of s2] (s3) {$p_3$};
\path 
               (s1) edge node [above] {$\snd{b}$} (s2)
               (s2) edge node [above] {$\rcv{c}$} (s3)
               (s1) edge [loop below] node [below] {$\snd{a}$} (s1)
               ;
             \end{tikzpicture}
    &
    $M_2$: &
             \begin{tikzpicture}[mycfsm, node distance = 0.4cm and 1cm
               ,scale=0.95, every node/.style={transform shape}]
               \node[state, initial, initial where=left] (s1) {$q_1$};
               \node[state, right=of s1] (s2) {$q_2$};
               \node[state, right=of s2] (s3) {$q_3$};
\path 
               (s1) edge node [above] {$\rcv{c}$} (s2)
               (s2) edge node [above] {$\snd{b}$} (s3)
               (s2) edge [loop below] node [below] {$\snd{a}$} (s2)
               ;
             \end{tikzpicture}
  \end{tabular}
\end{center}

For the $M_1 \nnincsymbol M_2$ case  consider the initial
configuration $(p_1, q_1)$. Since $p_1$ is a sending state, but $q_1$
is a receiving state, Case~\eqref{it:def-inc-alpha-async}
\emph{appears} to be the only applicable case of
Definition~\ref{def:subtyping}. However, we have $\loopio{!}{p_1}$
hence $(p_1, q_1) \notin \mathcal{R}$, for every asynchronous subtyping
relation $\mathcal{R}$.

For the $M_2 \nnincsymbol M_1$ case, consider the initial
configuration $(q_1, p_1)$. Since $q_1$ is a receiving state, only
Case~\ref{it:def-inc-tau} would be applicable. However, the input
context is empty and $p_1$ is a sending state, therefore neither
Case~\eqref{it:def-inc-tau-sync} nor Case~\eqref{it:def-inc-tau-async}
apply hence $(q_1, p_1) \notin \mathcal{R}$, for every asynchronous subtyping
relation $\mathcal{R}$.
\end{exa}

\section{A Sound Algorithm for Asynchronous Subtyping}\label{sec:algorithm}
Our subtyping algorithm takes two machines $M_1$ and $M_2$
then produces three possible outputs: \emph{true}, \emph{false},
or \emph{unknown}, which respectively indicate that ${M_1} \nincsymbol {M_2}$,
${M_1} \not\!\!\nincsymbol {M_2}$, or that the algorithm was unable to 
prove either of these two results.
The algorithm consists of three stages.
(1) It builds the \emph{simulation tree} of $M_1$ and $M_2$
(see Definition~\ref{def:simtree}) that represents sequences
of checks between $M_1$ and $M_2$, corresponding to the checks in the
definition of asynchronous subtyping.
Simulation trees may be infinite, but the construction terminates
whenever: either it reaches a node that cannot be expanded, it visits a node
whose label has been seen along the path from the root, or it expands
a node whose ancestors validate a termination condition that
we formalise in Theorem~\ref{thm:termination}.
The resulting tree satisfies one of the following conditions: 
(i) it contains a leaf that could not be expanded because the node
represents an unsuccessful check between $M_1$ and $M_2$ (in which case
the algorithm returns \emph{false}),
(ii) all leaves are successful final configurations, 
see Condition~\eqref{it:def-inc-final} of Definition~\ref{def:inclusion},
in which case the algorithm replies \emph{true},
or (iii) for each leaf $n$ it is possible to identify a corresponding
ancestor $\anc(n)$.  In this last case the tree and the identified
ancestors are passed onto the next stage.
(2) The algorithm divides the finite tree into several subtrees rooted at
those ancestors that do not have other ancestors above them
(see the strategy that we outline on page~\pageref{rem:divide-simtree}).
(3) The final stage analyses whether each subtree is of one of the two
following kinds.
(i) All the leaves in the subtree have the same label as their
ancestors: in this case all checks required to verify subtyping have
been performed.
(ii) The subtree is a \emph{witness subtree} (see
Definition~\ref{def:witness}), meaning that all the checks that may be
considered in any extension of the finite subtree are guaranteed to be
successful as well.
If all the identified subtrees are of one of these two kinds, the
algorithm replies \emph{true}. Otherwise, it replies \emph{unknown}.

\subsection{Generating Asynchronous Simulation Trees}\label{sub:simtree}
We first define labelled trees, of which our simulation trees are
instances; then, we give the operational rules for generating a
simulation tree from a pair of communicating machines.
\begin{defi}[Labelled Tree]
  A labelled tree is a tree\footnote{A tree is a connected directed
    graph without cycles: $\forall n \in N.\ n_0 \simtreetrans{}^* n \land \forall n,n' \in
    N.\ n \simtreetrans{}^+ n'.\ n\neq n'$.
}
$(N, n_0, \simtreetrans{}, \mathcal L, \Sigma, \Gamma)$, consisting
  of nodes $N$, root $n_0\in N$, edges
$\simtreetrans{}\ \subseteq N\times \Sigma \times N$, and node labelling
  function $\mathcal L : N \farrow \Gamma$.
\end{defi}
Hereafter, we write $n \simtreetrans{\sigma} n'$ when
$(n,\sigma, n') \in \simtreetrans{}$ and write
$n_1 \simtreetrans{\sigma_1 \cdots \sigma_{k}} n_{k+1}$ when there are
$n_1, \ldots, n_{k+1}$, such that
$n_i \simtreetrans{\sigma_i} n_{i+1}$ for all $1 \leq i \leq k$.
We write $n \simtreetrans{} n'$ when $n \simtreetrans{\sigma} n'$ for
some $\sigma$ and the label is not relevant. As usual, we write
$\simtreetrans{}^{\ast}$ for the reflexive and transitive closure of
$\simtreetrans{}$, and $\simtreetrans{}^{+}$ for its transitive
closure.
Moreover, we reason up-to tree isomorphism, i.e., two labelled trees
are equivalent if there exists a bijective node renaming that
preserves both node labelling and labelled transitions.

We can then define simulation trees, labelled trees representing all
possible configurations reachable by the simulation checked
by asynchronous session subtyping.

\begin{defi}[Simulation Tree]\label{def:simtree}
Let $M_1=(P,p_0,\delta_1)$ and $M_2=(Q,q_0,\delta_2)$ be two
communicating machines. 
The simulation tree of $M_1$ and $M_2$, written
  $\simtree{M_1}{M_2}$, is a labelled tree
  $(N, n_0, \simtreetrans{}, \labfun, \Act,
  P\times\ctxtQ)$. 
The labels $(p,T) \in (P \times \ctxtQ)$ are denoted also with $\simtreepair pT$. 
In order to define $\simtreetrans{}$
and $\labfun$, 
we first consider an $\Act$-labelled relation on  
$(P \times \ctxtQ)$,
with elements denoted with $\simtreepair pT\ \simtreetrans{\act}\ \simtreepair{p'}{T'}$,
defined as the minimal relation satisfying the following rules:
\[
    \begin{array}{c}
      \infer[\textsf{(In)}]
{\simtreepair pq\ \simtreetrans{?a}\ \simtreepair{p'}{q'}}
{p\trans{?a} p' & q\trans{?a} q' & \inTrans{p} \supseteq \inTrans{q}}      
\quad\quad
      \infer[\textsf{(Out)}]
{ \simtreepair pq\ \simtreetrans{!a}\ \simtreepair{p'}{q'} }
{p\trans{!a} p' & q\trans{!a} q' & \outTrans{p} \subseteq \outTrans{q}}
      \\[4mm]
\infer[\textsf{(InCtx)}]
{
\simtreepair{p}{ \echoice{a_i:T_i} }{i\in I} 
\ \simtreetrans{?a_k}\ 
\simtreepair{p'}{T_k}} 
{p \trans{?a_k} p' &  k\in I &
                                      \inTrans{p} \supseteq \{a_i \mid
                                      i\in I\, \}}
\\[4mm]
      \infer[\!\textsf{(OutAcc)}]
{
      \simtreepair p{\ctx A[q_j]^{j\in J}}
      \simtreetrans{!a}\simtreepair{p'}
      {\ctx A[\ctx A_j[q'_{j,h}]^{h\in H_j}]^{j\in J}}
      }
{
      \begin{array}{c@{}}
        \!
        p\trans{!a}p'
        \qquad\qquad
        \neg\loopio{!}{p}
        \\
        \forall j\in J. \big(
        \intree{q_j}\! =\! \ctx A_j[q_{j,h}]^{h\in H_j} \land
        \forall h \in H_j.
        ( \outTrans{p} \subseteq \outTrans{q_{j,h}} \land
        q_{j,h} \trans{!a} q'_{j,h}) \big)
      \end{array}
      }
    \end{array}
  \]
We now define $\simtreetrans{}$ and $\labfun$ as the transition
relation and the labelling function s.t. 
$\labfun(n_0)= \simtreepair {p_0}{q_0}$ and,
for each $n \in N$ with $\labfun(n)= \simtreepair pT$,
the following holds:
\begin{itemize}
\item
if $\simtreepair pT\ \simtreetrans{\act}\ \simtreepair{p'}{T'}$ then there 
exists a unique $n'$ s.t.
$n \simtreetrans{\act} n'$ with $\labfun(n')= \simtreepair {p'}{T'}$;
\item
if $n \simtreetrans{\act} n'$ with $\labfun(n')= \simtreepair {p'}{T'}$
then
$\simtreepair pT\ \simtreetrans{\act}\ \simtreepair{p'}{T'}$.

\end{itemize}
Notice that such a tree exists (it can be constructed inductively
starting from the root $n_0$) and it is unique (up-to tree isomorphism).
\end{defi}

Given machines $M_1$ and $M_2$, Definition~\ref{def:simtree} generates
a tree whose nodes are labelled by terms of the form
$\simtreepair p{\ctx A[q_i]^{i\in I}}$ where $p$ represents the state
of $M_1$, $\ctx A$ represents the receive actions accumulated by
$M_2$, and each $q_i$ represents the state of machine $M_2$ after each
path of accumulated receive actions from the root of $\ctx A$ to the
$i^{th}$ hole.
Note that we overload the symbol $\simtreepair{}{}$ used for
asynchronous subtyping (Definition~\ref{def:subtyping}), however the
actual meaning is always made clear by the context.
We comment each rule in detail below.

\noindent
\textbf{Rules \textsf{(In)} and \textsf{(Out)}} enforce contra-variance
of inputs and covariance of outputs, respectively, when no 
accumulated receive actions
are recorded, i.e., $\ctx A$ is a single hole.
Rule \textsf{(In)} corresponds to Case~\eqref{it:def-inc-tau-sync} of
Definition~\ref{def:inclusion}, while rule \textsf{(Out)} corresponds
to Case~\eqref{it:def-inc-alpha-sync}.

\noindent
\textbf{Rule \textsf{(InCtx)}} is applicable when the input tree $\ctx A$ is
non-empty and the state $p$ (of $M_1$) is able to perform a
receive action corresponding to any message located at the root of the
input tree (contra-variance of receive actions).
This rule corresponds to Case~\eqref{it:def-inc-tau-async} of
Definition~\ref{def:inclusion}.

\noindent
\textbf{Rule \textsf{(OutAcc)}} allows $M_2$ to execute some receive
actions before matching a send action executed by $M_1$.
This rule corresponds to Case~\eqref{it:def-inc-alpha-async} of
Definition~\ref{def:inclusion}.
Intuitively, each send action outgoing from state $p$ must also be
eventually executable from each of the states $q_j$ (in $M_2$) which
occur in the input tree ${\ctx A[q_j]^{j\in J}}$.
The possible combinations of receive actions executable from each $q_j$
before executing $\snd{a}$ is recorded in $\ctx A_j$, using
$\intree{q_j}$.
We assume that the premises of this rule only hold when all
invocations of $\intree{\cdot}$ are defined.
Each tree of accumulated receive actions is appended to its respective
branch of the input context $\ctx A$, using the notation
$ \ctx A[\ctx A_j[q'_{j,h}]^{h\in
  H_j}]^{j\in J}
$.
The premise
$\outTrans{p} \subseteq \outTrans{q_{j,h}} \land q_{j,h} \trans{!a}
q'_{j,h}$
guarantees that each $q_{j,h}$ can perform the send actions available
from $p$ (covariance of send actions).
The additional premise $\neg\loopio{!}{p}$ corresponds to that of
Case~\eqref{it:def-inc-alpha-async} of Definition~\ref{def:inclusion}.

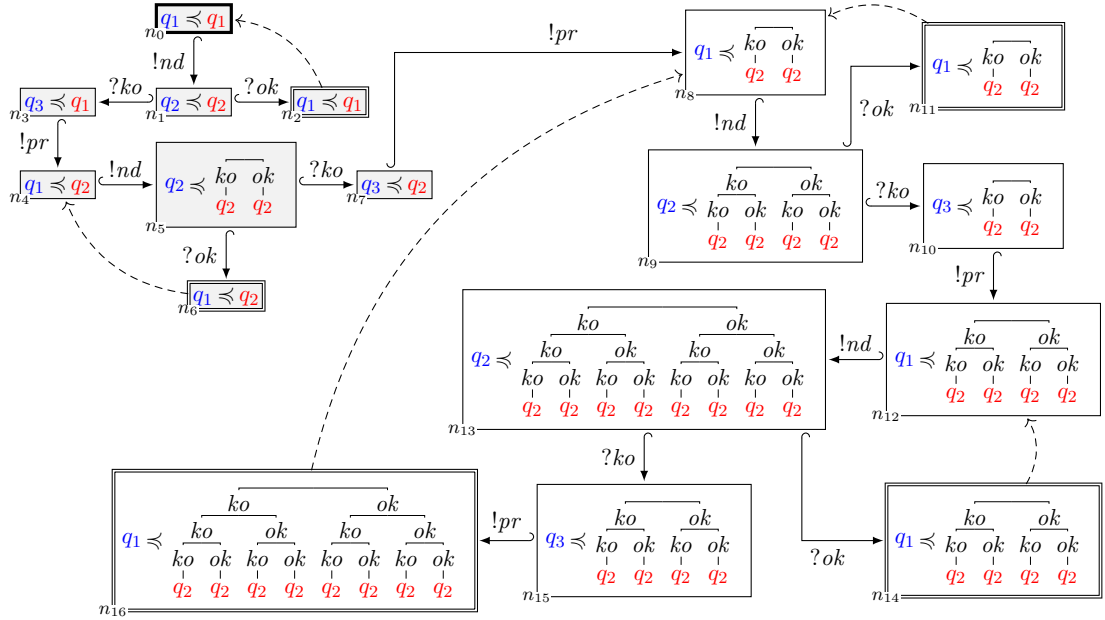
\begin{figure}[t]
\centering
  \begin{tikzpicture}
    [node distance=0.9cm and 1cm
    ,scale=0.8, every node/.style={transform shape}
    ]
  \node[cnode,notexplo,very thick] (n0) {$\pairconf{q_1}{q_1}$};
  \node[cnode,notexplo,below=of n0] (n1) {$\pairconf{q_2}{q_2}$};
  \node[cnode,notexplo,right=of n1,double] (n2) {$\pairconf{q_1}{q_1}$};
  \node[cnode,notexplo,left=of n1] (n3) {$\pairconf{q_3}{q_1}$};
  \node[cnode,notexplo,below=of n3] (n4) {$\pairconf{q_1}{q_2}$};  
  \node[cnode,notexplo,right=of n4] (n5) {\treeconf{$q_2$}{\treedone}};
  \node[cnode,notexplo,below=of n5,double] (n6) {$\pairconf{q_1}{q_2}$};
  \node[cnode,notexplo,right=of n5] (n7) {$\pairconf{q_3}{q_2}$};
  \node[cnode,right=of n7,yshift=2.2cm,xshift=3.2cm] (n8) {\treeconf{$q_1$}{\treedone}};
  \node[cnode,below=of n8] (n9) {\treeconf{$q_2$}{\treedtwo}};
  \node[cnode,right=of n9] (n10) {\treeconf{$q_3$}{\treedone}};
  \node[cnode,above=of n10,double] (n11) {\treeconf{$q_1$}{\treedone}};
  \node[cnode,below=of n10] (n12) {\treeconf{$q_1$}{\treedtwo}};
  \node[cnode,left=of n12] (n13) {\treeconf{$q_2$}{\treedthree}};
  \node[cnode,below right=of n13,double] (n14) {\treeconf{$q_1$}{\treedtwo}};
  \node[cnode,below=of n13] (n15) {\treeconf{$q_3$}{\treedtwo}};  
  \node[cnode,left=of n15,double] (n16) {\treeconf{$q_1$}{\treedthree}};
\draw[hookedsilentedge] (n0)  edge node [left] {$\snd{\exHQ}$} (n1);
  \draw[hookedsilentedge] (n1)  edge node [above] {$\rcv{\exOK}$} (n2);
  \draw[ancestor] (n2) edge [bend right] node {} (n0);
  \draw[lhookedsilentedge] (n1)  edge node [above] {$\rcv{\exKO}$} (n3);
  \draw[hookedsilentedge] (n3)  edge node [left] {$\snd{\exLQ}$} (n4);
\draw[hookedsilentedge] (n4)  edge node [above] {$\snd{\exHQ}$} (n5);
  \draw[hookedsilentedge] (n5)  edge node [left] {$\rcv{\exOK}$} (n6);
  \draw[ancestor] (n6) edge [bend left] node {} (n4);
  \draw[hookedsilentedge] (n5)  edge node [above] {$\rcv{\exKO}$} (n7);
  \draw[hookedsilentedge] (n7)  |- node [above right, near end] {$\snd{\exLQ}$} (n8);
  \draw[hookedsilentedge] (n8) edge node [left] {$\snd{\exHQ}$} (n9);
  \draw[hookedsilentedge] (n9) edge node [above] {$\rcv{\exKO}$} (n10);
  \draw[hookedsilentedge] ([xshift=-0.2cm]n9.north east) |- node [right, near start] {$\rcv{\exOK}$} (n11);
  \draw[ancestor] (n11) edge [bend right] node {} (n8);
  \draw[hookedsilentedge] (n10) edge node [left] {$\snd{\exLQ}$} (n12);
  \draw[lhookedsilentedge] (n12) edge node [above] {$\snd{\exHQ}$} (n13);
  \draw[hookedsilentedge] ([xshift=-0.4cm]n13.south east) |- node [below right] {$\rcv{\exOK}$} (n14);
  \draw[hookedsilentedge] (n13) edge node [left] {$\rcv{\exKO}$} (n15);
  \draw[ancestor] (n14) edge [bend right] node {} (n12);
  \draw[lhookedsilentedge] (n15) edge node [above] {$\snd{\exLQ}$} (n16);
  \draw[ancestor] (n16) edge [bend left] node {} (n8);
\node[nlabel] at (n0.south west) {$n_{0}$};
  \node[nlabel] at (n1.south west) {$n_{1}$};
  \node[nlabel] at (n2.south west) {$n_{2}$};
  \node[nlabel] at (n3.south west) {$n_{3}$};
  \node[nlabel] at (n4.south west) {$n_{4}$};
  \node[nlabel] at (n5.south west) {$n_{5}$};
  \node[nlabel] at (n6.south west) {$n_{6}$};
  \node[nlabel] at (n7.south west) {$n_{7}$};
  \node[nlabel] at (n8.south west) {$n_{8}$};
  \node[nlabel] at (n9.south west) {$n_{9}$};
  \node[nlabel] at (n10.south west) {$n_{10}$};
  \node[nlabel] at (n11.south west) {$n_{11}$};
  \node[nlabel] at (n12.south west) {$n_{12}$};
  \node[nlabel] at (n13.south west) {$n_{13}$};
  \node[nlabel] at (n14.south west) {$n_{14}$};
  \node[nlabel] at (n15.south west) {$n_{15}$};
  \node[nlabel] at (n16.south west) {$n_{16}$};
\end{tikzpicture}
\caption{Part of the simulation tree (solid edges only) and candidate tree for
  $\client \nincsymbol \clientsup$ (Figure~\ref{fig:runex-machines} and~\ref{fig:runex-machines2}).
The root is circled in thicker line.
The node identities are shown at the bottom left of each label.
}\label{fig:simtree}
\end{figure}

\begin{exa}
  Figure~\ref{fig:simtree} gives a graphical view of the initial part
  of the simulation tree $\simtree{\client}{\clientsup}$.
Consider the solid edges only for now, they correspond to the
  $\simtreetrans{}$-relation.
Observe that all branches of the simulation tree are infinite; some
  traverse nodes with infinitely many different labels, due to the
  unbounded growth of the input trees (e.g., the one repeatedly performing transitions $\snd{\exHQ} \cdot \rcv{\exKO} \cdot \snd{\exLQ}$);
  while others traverse nodes with \emph{finitely} many distinct
  labels (e.g., the one performing first transitions
  $\snd{\exHQ} \cdot \rcv{\exKO} \cdot \snd{\exLQ}$ and then
  repeatedly performing $\snd{\exHQ} \cdot\rcv{\exOK}$).
\end{exa}

We adapt the terminology of~\cite{JancarM99} and 
say that a node $n$ of $\simtree{M_1}{M_2}$ is a \emph{leaf} if it has
no successors.
A leaf $n$ is \emph{successful} iff $\labfun(n) = \simtreepair{p}{q}$,
with $p$ and $q$ final; all other leaves are unsuccessful.
A \emph{branch} (a full path through the tree) is \emph{successful}
iff it is infinite or finishes with a successful leaf; otherwise it is
unsuccessful.
Using this terminology, we relate asynchronous subtyping
(Definition~\ref{def:inclusion}) with simulation trees
(Definition~\ref{def:simtree}) in Theorem~\ref{thm:equivalence}.

\begin{restatable}{thm}{thmequivalence}\label{thm:equivalence}
  Let $M_1=(P,p_0,\delta_1)$ and $M_2=(Q,q_0,\delta_2)$ be two
  communicating machines.
All branches in $\simtree{M_1}{M_2}$ are successful if and only if
  $\minc{M_1}{M_2}$.
\end{restatable}
\begin{proof}
We start from the \emph{if part}.
Consider two communicating machines 
$M_1=(P,p_0,\delta_1)$ and $M_2=(Q,q_0,\delta_2)$ 
such that
$\minc{M_1}{M_2}$. 
By definition of $\minc{M_1}{M_2}$,
we have that there exists an asynchronous subtyping
$\mathcal R$
such that $(p_0,q_0) \in \mathcal R$.
Consider now $\simtree{M_1}{M_2}=
(N, n_0, \simtreetrans{}, \labfun, \Act,
  P\times\ctxtQ)$, having a root
labelled with
$\simpair {p_0} {q_0}$.
We have that also other nodes $n\in N$ are such that
$\mathcal L(n) = \simpair {p} {T}$ implies
$(p,T)\in \mathcal R$. This is easily proved by induction on
the length of the sequence of transitions
$n_0 \simtreetrans{}^+ n$,
observing that the rules for the
construction of the simulation tree
check on $p$ and $T$ 
the same properties checked by the definition of
asynchronous session subtyping, and
generate new transitions to nodes labelled with 
$\simpair {p'} {T'}$ corresponding to the pairs 
$(p',T')$ that are required to be in $\mathcal R$.
This guarantees that, for every $n$ in the simulation tree,
either $\mathcal L(n) = \simpair {p} {q}$
with $p\ntrans$ and $q\ntrans$ (i.e., $p$ and $q$ are 
final) implying that the branch to $n$ is successful,
or there exists $n'$ such that $n \simtreetrans{} n'$.
This guarantees that in $\simtree{M_1}{M_2}$ there exists 
no unsuccessful branch. 

We now move to the \emph{only if part}.
Consider two communicating machines 
$M_1=(P,p_0,\delta_1)$ and $M_2=(Q,q_0,\delta_2)$
and their simulation tree
$\simtree{M_1}{M_2}=
(N, n_0, \simtreetrans{}, \labfun, \Act,
  P\times\ctxtQ)$.
Consider now the relation $\mathcal R \subseteq P\times\ctxtQ$
such that $(p,T) \in \mathcal R$ if and only if there exists
$n \in N$ s.t. $\mathcal L(n) = \simpair {p} {T}$.
With similar arguments as in the above case, we prove
that $\mathcal R$ is an asynchronous subtyping
relation. Hence, given that 
$\mathcal L(n_0) = \simpair {p_0} {q_0}$,
we have $(p_0,q_0) \in \mathcal R$, hence also 
$\minc{M_1}{M_2}$. 
\end{proof}

\subsection{A Simulation Tree-Based Algorithm }\label{sub:algorithm}

A consequence of the undecidability of asynchronous session
subtyping~\cite{LangeY17,BCZ18,BravettiCZ17} is that checking whether
all branches in $\simtree{M_1}{M_2}$ are successful is undecidable.
The problem follows from the presence of infinite branches that cannot
be algorithmically identified.  Our approach is to characterise finite
subtrees (called \emph{witness subtrees}) such that all the branches
that traverse these finite subtrees are guaranteed to be infinite.

The presentation of our algorithm is in three parts. 
In Part (1), we give the definition of the kind of \emph{finite}
subtree (of a simulation tree) we are interested in
(called \emph{candidate} subtrees).
In Part (2), we give an algorithm to extract 
\emph{candidate} subtrees from a simulation tree $\simtree{M_1}{M_2}$.
In Part (3) we show how to check
whether a candidate subtree (which is finite)
is a \emph{witness} of infinite branches
(hence successful) in the simulation tree.

\subsubsection{Part 1. Characterising finite and candidate sub-trees}
We define the candidate subtrees of a simulation tree, which are
finite subtrees accompanied by an ancestor function mapping each
boundary node $n$ to a node located on the path from the root of the
tree to $n$.

\begin{defi}[Finite Subtree]
  A finite subtree $(r,B)$ of a labelled tree
  $S = (N, n_0, \simtreetrans{},\linebreak{} \mathcal L, \Sigma,
  \Gamma)$, with $r$ being the subtree root and $B \subseteq N$ the
  finite set of its leaves (boundary nodes), is the subgraph of $S$ such that:
  \begin{enumerate}
  \item \label{en:subtree-connected}
    $\forall n \!\in\! B.\ r \simtreetrans{}^* n$;
  \item \label{en:subtree-limit}
$\forall n \!\in\! B.\not\exists n' \!\in\! B.\ n
    \simtreetrans{}^+ n'$; and
  \item \label{en:subtree-onpath}
    $\forall n \in N.\ r \simtreetrans{}^* n \implies \exists n' \in B
    .\ n \simtreetrans{}^* n' \lor n' \simtreetrans{}^* n$.
  \end{enumerate}
We use
  $\subtree{S}{r}{B} = \{n \in N \mid \exists n'\in B.\ r
  \simtreetrans{}^* n \simtreetrans{}^*n'\}$ to denote the (finite)
  set of nodes of the finite subtree $(r,B)$. Notice that
  $r \in \subtree{S}{r}{B}$ and $B \subseteq \subtree{S}{r}{B}$.
\end{defi}
 
\noindent 
Condition~\eqref{en:subtree-connected} requires that each boundary
node can be reached from the root of the subtree.
Condition~\eqref{en:subtree-limit} guarantees that the boundary nodes
are not connected, i.e., they are on different paths from the root.
Condition~\eqref{en:subtree-onpath} enforces that each branch of the
tree passing through the root $r$ contains a boundary node.

\begin{defi}[Candidate Subtree]\label{def:candidate}
  Let $M_1=(P,p_0,\delta_1)$ and $M_2=(Q,q_0,\delta_2)$ be two
  communicating machines with $\simtree{M_1}{M_2} = (N,
  n_0, \simtreetrans{}, \labfun, \Act, P\times\ctxtQ)$.

  \noindent
  A \emph{candidate subtree} of $\simtree{M_1}{M_2}$ is a finite
  subtree $(r,B)$ paired with a function
$\anc: B \farrow \subtree{\simtree{M_1}{M_2}}{r}{B} \!\setminus\!B$
  such that, for all $n \in B$, we have:
  $\anc(n) \simtreetrans{}^+ n$ and
there are $p, \ctx A, \ctx A', I, J, \{q_j \mid j\in J\}$ and $\{q_i \mid i\in I\}$ such that
  \[
  \mathcal L(n)=\simpair p \ctx A[q_{i}]^{i\in I} 
\quad      \land  \quad
\mathcal L(\ancestor{n})=\simpair p \ctx A'[q_{j}]^{j\in
    J} 
\quad      \land  \quad
\{q_i \mid i\in I\} \subseteq \{q_j \mid j\in J\}
  \]
\end{defi}
A candidate subtree is a finite subtree accompanied by a total
function on its boundary nodes.
The purpose of function $\anc$ is to map 
each boundary node $n$ to a
``similar'' ancestor $n'$ such that: 
$n'$ is a node (different from $n$)
on the path from the root $r$
to $n$ (recall that we have $r \notin B$) such that the labels of $n'$ and $n$ share the same state $p$ of $M_1$, and the states of
$M_2$ (that populate the holes in the leaves of the input context of
the boundary node) are a subset of those considered for the ancestor.
Given a candidate subtree, we write $\img{\anc}$ for the set
$\{ n \qst \exists n' \in B \dst \anc(n') = n\}$, i.e., $\img{\anc}$
is the set of ancestors in the candidate subtree.

\begin{exa}
  Figure~\ref{fig:simtree} depicts a finite subtree of
  $\simtree{\client}{\clientsup}$.
We can distinguish several distinct candidate subtrees
  in Figure~\ref{fig:simtree}. For instance
one subtree is rooted at $n_0$, and its boundary nodes are
  $\{n_2,n_6,n_{11},n_{14},n_{16}\}$; another subtree is rooted at
  $n_8$ and its boundary nodes are $\{n_{11},n_{14},n_{16}\}$
  (boundary nodes are highlighted with a double border).
In each subtree, the $\anc$ function is represented by the dashed
  edges from its boundary nodes to their respective ancestors.
\end{exa}

 \subsubsection{Part 2. Identifying candidate subtrees} We now describe
how to generate a finite subtree of the simulation tree, from which we
extract candidate subtrees. Since simulation trees are potentially
infinite, we need to identify termination conditions (i.e., conditions
on nodes that become the boundary of the generated finite subtree).

We first need to define the auxiliary function $\extract{\ctx A}{\word}$,
which checks the presence of a sequence of messages $\word$ in an
input context $\ctx A$, and extracts the residual input context.
\[
\extract{\ctx A}{\word} = 
\begin{cases}
  \ctx A & \text{if }  \word = \lempty
  \\
  \extract{\ctx A_i}{\word'} 
  & \text{if } \word = a_i \cdot \word',
  \ctx A = \echoice{a_j:{\mathcal A}_j}{j\in J},
\text{ and } 
  i \in J
  \\
  \perp & \text{otherwise}
\end{cases}
\]

Our termination condition is formalised in
Theorem~\ref{thm:termination} below. This result follows from an
argument based on the finiteness of the states of $M_1$ and of the sets of states from
$M_2$ (which populate the holes of the input contexts in the labels of
the nodes in the simulation tree). We write $\minHeight{\ctx A}$ for the smallest $\height{\ctx A}{i}$,
with $i \in I(\ctx A)$, where $\height{\ctx A}{i}$ is the length of
the path from the root of the input context $\ctx A$ to the
$i$\textsuperscript{th} hole.
\begin{restatable}{thm}{thmtermination}\label{thm:termination}
  Let $M_1=(P,p_0,\delta_1)$ and $M_2=(Q,q_0,\delta_2)$ be two
  communicating machines with
  $\simtree{M_1}{M_2}=
  (N, n_0, \simtreetrans{}, \mathcal L, \Act, P\times\mathcal T_Q)$.

  \noindent
  For each infinite path 
$n_0 \simtreetrans{} n_1 \simtreetrans{} n_2 \cdots \simtreetrans{} n_l
  \simtreetrans{} \cdots$ 
there exist $i < j < k$,
  with 
\[
  \mathcal L(n_i)=\simpair p \ctx A_i[q_h]^{h \in H_i}
  \qquad
  \mathcal L(n_j)=\simpair p \ctx A_j[q'_h]^{h \in H_j}
  \qquad
  \mathcal L(n_k)=\simpair p \ctx A_k[q''_h]^{h \in H_k}
  \]
  such that
$\{q'_h \mid h \in H_j\} \subseteq \{q_h \mid h \in H_i\}$
    and
    $\{q''_h \mid h \in H_k\} \subseteq \{q_h \mid h \in H_i\}$;

    \noindent
    and,
for $n_i \simtreetrans{\acts} n_j$:
    \begin{enumerate}[label=($\roman*$)]
    \item
      $\receives{\acts}=\word_1 \cdot \word_2$ with $\word_1$ s.t.
      $\exists t,z.\
      \extract{\ctx A_i}{\word_1}=[\,]_t \land
      \extract{\ctx A_k}{\word_1}=[\,]_z$,  or
    \item
      $\minHeight{\extract{\ctx A_i}{\receives{\acts} }} \leq
      \minHeight{\extract{\ctx A_k}{\receives{\acts} }}$.
    \end{enumerate}
\end{restatable}
\begin{proof}
  Let $M_1=(P,p_0,\delta_1)$ and $M_2=(Q,q_0,\delta_2)$ be two
  communicating machines
with
  $\simtree{M_1}{M_2}=
  (N, n_0,
\simtreetrans{}, \mathcal L, \Act, P\times\mathcal T_Q)$, and
  let
  $n_0 \simtreetrans{} n_1 \simtreetrans{} n_2 \cdots \simtreetrans{} n_l
  \simtreetrans{} \cdots$
  be an infinite path in the simulation tree. 
For each $n_i$, let $\mathcal S_i$ be
the pair $(p_i,R_i)$, with $p_i \in P$ and $R_i \subseteq Q$, such that 
$\mathcal L(n_i)=\simpair {p_i} \ctx A_i[q_j]^{j \in J_i}$
and
$R_i = \{q_j \mid j \in J_i\}$.
Notice that there are at most $|P| \times 2^{|Q|}$
distinct pairs 
$(p_i,R_i)$, in which $p_i$ is an element taken from the finite set $P$,
and $R_i$ is a subset of the finite set $Q$.
This guarantees the existence of infinite pairs of nodes
$(n_{i_1},n_{i'_1}), (n_{i_2},n_{i'_2}), \dots ,(n_{i_j},n_{i'_j}),\dots$ taken from the above infinite path,
such that, for all $j$:
\begin{itemize}
\item
$\mathcal S_{i_j}=\mathcal S_{i'_j}$ and
\item
$i_j<i'_j<i_{j+1}$ and
\item
$i'_j-i_j \leq |P| \times 2^{|Q|}$.
\end{itemize}
The above follows from the possibility to repeatedly select, 
by following from left to right the infinite sequence
$n_0 \simtreetrans{} n_1 \simtreetrans{} n_2 \cdots \simtreetrans{} n_l
  \simtreetrans{} \cdots$, the first occurring pair 
$(n_k,n_l)$, with $k<l$, such that $\mathcal S_{k}=\mathcal S_{l}$. 
Being the first pair of this type that occurs,
we have that $l-k \leq |P| \times 2^{|Q|}$.

For the above infinite list of pairs $(n_{i_1},n_{i'_1}), (n_{i_2},n_{i'_2}), \dots ,(n_{i_j},n_{i'_j}),\dots$
let $\acts_j$ be such that $n_{i_j} \simtreetrans{\acts_j} n_{i'_j}$.
All these infinitely
many sequences of actions $\acts_j$
have bounded length (smaller than $|P| \times 2^{|Q|}$), hence
infinitely many of them will coincide
(this is because there are only boundedly many distinct actions
that are admitted).
Let $\alpha$ be such a sequence of actions that is considered
for infinitely many paths $n_{i_j} \simtreetrans{\alpha} n_{i'_j}$.
Moreover, being the possible distinct $(p_i,Q_i)$ finite,
there exists one pair $(p,Q)$ such that
infinitely many of these paths 
$n_{i_j} \simtreetrans{\alpha} n_{i'_j}$
will be such that $\mathcal S_{i_j}=\mathcal S_{i'_j}=(p,Q)$.

Summarising, we have proved the existence of
$(n_{v_1},n_{v'_1}), (n_{v_2},n_{v'_2}), \dots ,(n_{v_j},n_{v'_j}),\dots$,
with $v_j<v'_j<v_{j+1}$ for all $j$,
for which there exist $(p,Q)$ and $\alpha$ such that, 
for all $j$, $n_{v_j} \simtreetrans{\alpha} n_{v'_j}$ and
$\mathcal S_{v_j}=\mathcal S_{v'_j}=(p,Q)$.

We now consider $\word=\receives{\alpha}$.
We have that the input actions in $\word$,
executed in each path $n_{v_j} \simtreetrans{\alpha} n_{v'_j}$,
will be matched by the input context $\ctx A_j$
of $\mathcal L(n_{v_j}) = \simpair p \ctx A_j[q_j]^{h \in H}$.
There are two possibilities:
\begin{enumerate}
\item\label{case1}
either $\word$ is included
in a path root-hole of $\ctx A_j$ (hence 
$\extract{\ctx A_j}{\word}$ is defined),
\item\label{case2}
or there exists a path root-hole which corresponds
to a prefix of $\word$ (in this case we have that 
$\word=\word_1 \cdot \word_2$ with 
$\extract{\ctx A_j}{\word_1}=[\,]_{t_{j}}$).
\end{enumerate}
At least one of the two cases occurs infinitely often, i.e., there
exist infinitely many indices $j$, such that for all paths
$n_{v_j} \simtreetrans{\alpha} n_{v'_j}$ item \ref{case1} holds, or
there exist infinitely many indices $j$, such that for all paths
$n_{v_j} \simtreetrans{\alpha} n_{v'_j}$ item \ref{case2} holds.  In
the first case, we have that there exist at least two indices $j_1$
and $j_2$ such that
$\minHeight{\extract{\ctx A_{j_1}}{\word}} \leq
\minHeight{\extract{\ctx A_{j_2}}{\word}}$ (in fact, $\minHeight{}$
returns a non negative value, hence such values cannot infinitely
decrease).  In the second case, we have that there exist at least two
indices $j_1$ and $j_2$ such that
$\extract{\ctx A_{j_1}}{\word_1}=[\,]_{t_{j_1}}$ and
$\extract{\ctx A_{j_2}}{\word_1}=[\,]_{t_{j_2}}$ for the same
$\word_1$ prefix of $\word$ (in fact, $\word$ has only finitely many
prefixes).

We can conclude that the thesis holds
by considering $i={v_{j_1}}$, $j={v'_{j_1}}$ and 
$k={v_{j_2}}$.
\end{proof}

Intuitively, the theorem above says that for each infinite
branch in the simulation tree, we can find special nodes $n_i$, $n_j$
and $n_k$ such that the set of states in $\ctx A_j$ (resp.\
$\ctx A_k$) is included in that of $\ctx A_i$ and the receive actions
in the path from $n_i$ to $n_j$ are such that: either ($i$) only a
precise prefix of such actions will be taken from the receive actions
accumulated in $n_i$ and $n_k$ or ($ii$) all of them will be taken
from the receive actions in which case $n_k$ must have accumulated
more receive actions than $n_i$.
Case ($i$) deals with infinite branches with only finite labels (hence
finite accumulation) while case ($ii$) considers those cases in which
there is unbounded accumulation along the infinite branch.

As an example of this latter case, consider the simulation
tree depicted in Figure \ref{fig:simtree}. Let $n_i = n_8$,
$n_j = n_{12}$ and $n_k = n_{16}$. These nodes are along the
same path, moreover we have 
$\mathcal L(n_i)=\simpair {q_1} \ctx A_i[q_h]^{h \in H_i}$,
$\mathcal L(n_j)=\simpair {q_1} \ctx A_j[q'_h]^{h \in H_j}$,
$\mathcal L(n_k)=\simpair {q_1} \ctx A_k[q''_h]^{h \in H_k}$
with $\{q_h \mid h \in H_i\} = \{q'_h \mid h \in H_j\} = \{q''_h \mid h \in H_k\}
= \{q_2\}$ and 
$0 = \minHeight{\extract{\ctx A_i}{\rcv{\exKO}}} \leq
      \minHeight{\extract{\ctx A_k}{\rcv{\exKO}}} = 2$.
Notice that the path in the simulation tree from $n_8$ to $n_{16}$
can be infinitely repeated with the effect of increasing the height 
of the input context.

\smallbreak 

\noindent Based on Theorem~\ref{thm:termination},
the following {\em algorithm} generates a finite subtree of
$\simtree{M_1}{M_2}$:
\begin{quote}
  Starting from the root, compute the branches\footnote{The order
    nodes are generated is not important (our implementation uses a
    DFS approach, cf.~\S\ref{sec:tool}).} of $\simtree{M_1}{M_2}$
  stopping when one of the following types of node is
encountered: a leaf,
a node $n$ with a label already seen along the path from the root to
  $n$,
or a node $n_k$ (with the corresponding node $n_i$) as those
  described by the above Theorem \ref{thm:termination}.
\end{quote}

\begin{exa}
  Consider the finite subtree in Figure~\ref{fig:simtree}.  It is precisely the finite subtree
  identified as described above: we stop generating the simulation
  tree at nodes $n_2$, $n_6$, $n_{11}$, and $n_{14}$ (because their labels
  have been already seen at the corresponding ancestors $n_0$, $n_4$,
  $n_{8}$, and $n_{12}$) and $n_{16}$ (because of the ancestors $n_8$ and
  $n_{12}$ such that $n_8$, $n_{12}$ and $n_{16}$ correspond to the
  nodes $n_i$, $n_j$ and $n_k$ of Theorem \ref{thm:termination}).
\end{exa}

\noindent When the computed finite subtree contains an unsuccessful
leaf,
we can immediately conclude that the considered communicating machines
are not related. Otherwise, we extract smaller finite
subtrees (from the subtree) that are potential candidates to be subsequently checked.

\begin{quote}
We define the $\anc$ function
as follows: for boundary nodes $n$ with an ancestor $n'$ such that
$\labfun(n) = \labfun(n')$ we define $\anc(n) = n'$; for boundary
nodes $n_k$ with the corresponding node $n_i$ as those described by
Theorem \ref{thm:termination}, we define $\anc(n_k) = n_i$.
The extraction of the finite subtrees is done by characterising their
roots (and taking as boundary their reachable boundary nodes): let
$P = \{ n \in \img{\anc} \qst \exists n' \dst \anc(n') = n \land
\labfun(n)\neq\labfun(n')\}$, the set of such roots is
$R = \{ n \in P \qst \not\!\exists n'\in P.\ n' \simtreetrans{}^+
n\}$.
\end{quote}

\label{rem:divide-simtree}
\noindent Intuitively, to extract subtrees, we restrict our attention
to the set $P$ of ancestors with a label different from their
corresponding boundary node (corresponding to branches that can
generate unbounded accumulation).
We then consider the forest of subtrees
rooted in nodes in $P$ without an ancestor in $P$.
Notice that for successful leaves we do not define 
$\anc$; hence, only extracted subtrees without successful nodes
have a completely defined $\anc$ function.  These are candidate
subtrees that will be checked as described in the next step.

\begin{exa}
  Consider the finite subtree in
  Figure~\ref{fig:simtree}.
Following the strategy above we extract from it the candidate subtree rooted at $n_8$ (white
  nodes), with boundary $\{n_{11},n_{14},n_{16}\}$.  Note that each
  ancestor node above $n_8$ has a label identical to its boundary
  node.
\end{exa}

 \subsubsection{Part 3. Checking whether the candidate subtrees are
witnesses of infinite branches}
The final step of our algorithm consists in verifying a
property on the identified candidate subtrees which guarantees 
that all branches
traversing the root of the candidate subtree are infinite,
hence successful.
A candidate subtree satisfies this property when it 
is also a \emph{witness subtree}, which is the
key notion (Definition~\ref{def:witness})
presented in this third part.

In order for a subtree to be a witness, we require that any 
behaviour in the simulation tree going beyond the subtree 
is the infinite repetition of the behaviour already observed 
in the considered finite subtree.
This infinite repetition is only possible if whatever 
receive actions are 
accumulated in the input context $\ctx A$ (using Rule
\textsf{(OutAcc)}) are eventually executed
by the candidate subtype $M_1$ in Rule \textsf{(InCtx)}.
The compatibility check between the receive actions that can be
accumulated and the receive actions that are eventually executed is
done by first synthesising a finite representation of
the possible (repeated) accumulation of the
candidate supertype $M_2$ and the possible (repeated) receive actions
of the candidate subtype $M_1$.
We then check whether these representations of the input actions are \emph{compatible}, wrt.\ the
$\sqsubseteq$-relation, see Definition~\ref{def:compatibility}.
We define these representations of the input behaviours as a system of
(possibly) mutually recursive equations, which we call a \emph{system
  of input tree equations}.
  
Intuitively, a system of input tree equations represents 
a family of trees, that we use to represent the input behaviour
of types. We need to consider families of trees because types include also
output actions that, in case we are concerned with input actions only, 
can be seen as internal silent actions,
representing nondeterministic choices among alternative 
future inputs (i.e. alternative subtrees).
  
\begin{defi}[Input Tree Equations]\label{def:equations}
  Given a set of variables $\mathcal V$, ranged over by $X$, an input
  tree expression is a term of the grammar
  \[
    E\ \ ::=\ \ X \ \mid\ \echoice{a_i:E_i}{i\in I}\ \mid\  \schoice{E_i}{i \in I}
  \]
The free variables of an input tree expression $E$ are the variables
  which occur in $E$.
Let $T_\mathcal V$ be the set of input tree expressions whose free
  variables are in $\mathcal V$.
  
  A system of input tree equations is
  a tuple $\mathcal G = (\mathcal V,X_0,\DEF)$ consisting of a set of
  variables $\mathcal V$, an initial variable $X_0 \in \mathcal V$,
  and with $\DEF$ consisting of exactly one input tree expression
  $X \eqdef E$ for each $X \in \mathcal V$, with
  $E \in \mathcal T_\mathcal V$.
\end{defi}
Given an input tree expression of the form $\echoice{a_i:E_i}{i\in I}$
or $\schoice{E_i}{i \in I}$, we assume that $I \neq \emptyset$,
$\forall i \neq j \in I \dst a_i \neq a_j$, and that the order of the
sub-terms is irrelevant.
Whenever convenient, we use set-builder notation to construct an input
tree expression, e.g., $\schset{E_i}{i \in I}$.
In an input tree equation, the construct $\echoice{a_i:E_i}{i\in I}$
represents the capability of accumulating (or actually executing)
the receive actions on
each message $a_i$
then behaving as in $E_i$.
The construct $\schoice{E_i}{i \in I}$ represents a \emph{silent
  choice} between the different capabilities $E_i$.
  
  We now define the notion of compatibility between two systems
  of input tree equations.
  Intuitively, two systems of input tree equations are compatible 
  when all the trees of the former have \emph{less} 
  alternatives than the trees of the latter. More precisely, at
  each input choice, the alternative branchings of the
  former are included in those of the latter. 

\begin{defi} [Input Tree Compatibility]\label{def:compatibility}
  Given two systems of input tree equations
  $\mathcal G = (\mathcal V,X_0,\DEF)$ and
  $\mathcal G' = (\mathcal V',X_0',\DEF')$, such that
$\mathcal V \cap \mathcal V' = \emptyset$, we say that
$\mathcal G$ is \emph{compatible} with $\mathcal G'$, written
  $\mathcal G \sqsubseteq \mathcal G'$,
if there exists a compatibility relation 
  $\mathcal R \subseteq \mathcal T_\mathcal V \times \mathcal
  T_\mathcal V'$. That is a relation $\mathcal R$
s.t.  $(X_0,X_0') \in \mathcal R$ and:
  
\begin{enumerate}
  \item \label{en:unfold-l}
    if $(X,E) \in \mathcal R$ then $(E',E) \in \mathcal R$ with $X \eqdef E'$;
  \item \label{en:unfold-r}
    if $(E,X) \in \mathcal R$ then $(E,E') \in \mathcal R$ with $X \mathbin{\smash{\eqdef}} E'$;
  \item \label{en:schoice-l}
    if $( \schoice{E_i}{i \in I}, E) \in \mathcal R$ then 
    $\forall i \in I.\ (E_i,E)\in \mathcal R$;
  \item \label{en:schoice-r}
    if $(E, \schoice{E_i}{i \in I}) \in \mathcal R$ then 
    $\forall i \in I.\ (E,E_i)\in \mathcal R$; 
  \item \label{en:echoice}
    if $(\echoice{a_i:E_i}{i\in I},\echoice{a_j:E'_j}{j\in J}) \in \mathcal R$ then 
    $I \subseteq J$ and $\forall i \in I.\ (E_i,E'_i)\in \mathcal R$.
  \end{enumerate}
  We extend the use of $\sqsubseteq$, defined on input tree equations,
  to terms $E \in T_\mathcal V$ and $E' \in T_\mathcal V'$; namely, we 
  write $E \sqsubseteq E'$ if there exists a compatibility relation $\mathcal R$ 
  s.t. $(E,E') \in \mathcal R$.
\end{defi}
Notice that compatibility is formally
defined following a coinductive approach that performs
the following checks
on $\mathcal G$ and $\mathcal G'$, starting from the
initial pair $(X_0,X_0')$.
The first two items of Definition~\ref{def:compatibility} let
variables be replaced by their respective definitions.
The next two items explore all the successors of silent choices.
The last item guarantees that all the receive actions of the l.h.s.
can be actually matched by receive actions in the r.h.s.
The check of compatibility will be used in Definition \ref{def:witness},
in order to control that the candidate supertype always has
input branchings included in those of the candidate subtype. More precisely, 
we will check that the system of input tree equations, that represents
the possible inputs of the supertype, is compatible with that of the candidate
subtype.

\begin{exa}\label{ex:eqsystem}
  Consider the two systems of input tree equations in
  Figure~\ref{fig:eqsystems}.
We have $\mathcal G \sqsubseteq \mathcal G'$. We enumerate a few
  pairs which must be in the embedding relation:
  \[
      \begin{array}{ll}
        \text{Initial variables:} & (X_0, Y_{n_8})
        \\
        \text{Unfold $X_0$ (Case~\eqref{en:unfold-l} of Def.~\ref{def:compatibility})}
                                  & (\echoice{\msg{ok} : X_{q_2,n_8} , \msg{ko} : X_{q_2,n_8}}{} , Y_{n_8})
        \\
        \text{Unfold $Y_{n_8}$ (Case~\eqref{en:unfold-r} of Def.~\ref{def:compatibility})}
                                  & (\echoice{\msg{ok} : X_{q_2,n_8} , \msg{ko} : X_{q_2,n_8}}{} ,  \schoice{Y_{n_9}}{})
        \\
        \text{Silent-choice (Case~\eqref{en:schoice-r} of Def.~\ref{def:compatibility})}
                                  &(\echoice{\msg{ok} : X_{q_2,n_8} , \msg{ko} : X_{q_2,n_8}}{} ,  Y_{n_9})
        \\
        \text{Unfold $Y_{n_9}$ (Case~\eqref{en:unfold-r} of Def.~\ref{def:compatibility})}
                                  & (\echoice{\msg{ok} : X_{q_2,n_8} , \msg{ko} : X_{q_2,n_8}}{} ,  \echoice{\msg{ok}: Y_{n_8}, \ \msg{ko} :       
                                    Y_{n_{10}}}{})
        \\
        \text{Choice (Case~\eqref{en:unfold-r} of Def.~\ref{def:compatibility})}
                                  & (X_{q_2,n_8} ,  Y_{n_8}) \text{ and } (X_{q_2,n_8},   Y_{n_{10}})
                                    
      \end{array}
    \]
  \end{exa}

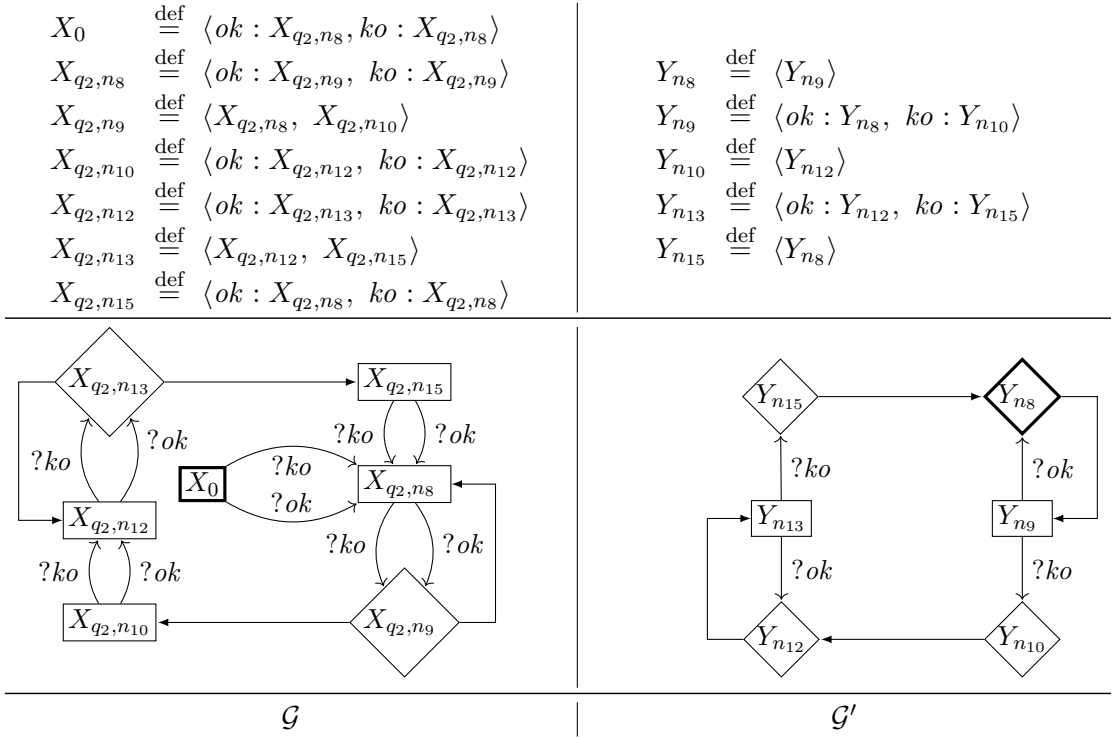
\begin{figure}[t]
  \centering
  \begin{tabular}{c|c}
        \begin{tabular}{l@{\; $\eqdef$ \;}l}
      $X_0$ & $\echoice{\msg{ok} : X_{q_2,n_8} , \msg{ko} : X_{q_2,n_8}}{}$
      \\
      $X_{q_2,n_8}$ & $\echoice{\exOK : X_{q_2,n_9} , \ \exKO :
                      X_{q_2,n_9}}{}$
      \\
      $X_{q_2,n_9}$ & $ \schoice{X_{q_2,n_8}, \ X_{q_2,n_{10}} }{}$
      \\
      $X_{q_2,n_{10}} $ & $ \echoice{\exOK : X_{q_2,n_{12}} , \ \exKO :
                          X_{q_2,n_{12}}}{}$
      \\
      $X_{q_2,n_{12}}  $ & $  \echoice{\exOK : X_{q_2,n_{13}} , \ \exKO :
                           X_{q_2,n_{13}}}{}$
      \\
      $X_{q_2,n_{13}}  $ & $ \schoice{X_{q_2,n_{12}}, \ X_{q_2,n_{15}}
                           }{}$
      \\
      $X_{q_2,n_{15}}  $ & $ \echoice{\exOK : X_{q_2,n_{8}} , \ \exKO :
                           X_{q_2,n_{8}}}{}$          
    \end{tabular}
            &
              \begin{tabular}{l@{\; $\eqdef$ \;}l}
                $Y_{n_8} $ & $  \schoice{Y_{n_9}}{}$
\\
                $Y_{n_{9}} $ & $  \echoice{\msg{ok}: Y_{n_8}, \ \msg{ko} :
                               Y_{n_{10}}}{}$ 
\\
                $Y_{n_{10}} $ & $   \schoice{Y_{n_{12}}}{}$
                \\
$Y_{n_{13}} $ & $   \echoice{\msg{ok}: Y_{n_{12}}, \ \msg{ko} :
                                Y_{n_{15}}}{}$
                \\
$Y_{n_{15}} $ & $  \schoice{Y_{n_{8}}}{}$ 
              \end{tabular}
    \\
    \midrule
    \begin{tikzpicture}
      [node distance=0.9cm and 2.7cm
      ,scale=0.95, every node/.style={transform shape}
      ,every edge/.append style={shorten >=0.5pt}
      ]
      \node[echnode] (x28) {$\pvar{2}{8\phantom{8}}$};
\node[schnode,below=of x28] (x29) {$\pvar{2}{9\phantom{8}}$};
      \node[echnode,left=of x29] (x211) {$\pvar{2}{10}$};
      \node[echnode,above=of x211] (x212) {$\pvar{2}{12}$};
      \node[schnode,above=of x212] (x213) {$\pvar{2}{13}$};
      \node[echnode,right=of x213] (x215) {$\pvar{2}{15}$};
\node[echnode,very thick,xshift=1.3cm] (x0)  at (x28-|x211)   {$\pvarroot$};
\draw[->] (x0)  edge [bend left] node [below] {$\rcv{\exKO}$} (x28.north west);
      \draw[->] (x0)  edge [bend right] node [above] {$\rcv{\exOK}$} (x28.south west);
\draw[->] (x28)  edge [bend left] node [right] {$\rcv{\exOK}$} (x29);
      \draw[->] (x28)  edge [bend right] node [left] {$\rcv{\exKO}$} (x29);
\draw[silentedge] (x29)  edge node [left] {} (x211);
      \draw[silentedge] (x29.east) -- ([xshift=0.5cm]x29.east) |- (x28.east);
\draw[->] (x211)  edge [bend left] node [left] {$\rcv{\exKO}$} (x212);
      \draw[->] (x211)  edge [bend right] node [right] {$\rcv{\exOK}$} (x212);
\draw[->] (x212)  edge [bend left] node [left] {$\rcv{\exKO}$} (x213);
      \draw[->] (x212)  edge [bend right] node [right, near end] {$\rcv{\exOK}$} (x213);
\draw[silentedge] (x213)  edge node [left] {} (x215);
      \draw[silentedge] (x213.west) -- ([xshift=-0.5cm]x213.west) |- (x212.west);
\draw[->] (x215)  edge [bend left] node [right] {$\rcv{\exOK}$} (x28);
      \draw[->] (x215)  edge [bend right] node [left] {$\rcv{\exKO}$} (x28);
    \end{tikzpicture}
    \qquad\qquad
    &
      \qquad\qquad
      \begin{tikzpicture}
        [node distance=0.9cm and 2.3cm
        ,scale=0.95, every node/.style={transform shape}
        ,every edge/.append style={shorten >=0.5pt}
        ]
        \node[schnode,very thick] (y8) {$\cvar{8{\phantom{8}}}$};
        \node[echnode,below=of y8] (y9) {$\cvar{9\phantom{8}}$};
        \node[schnode,below=of y9] (y11) {$\cvar{10}$};
        \node[schnode,left=of y11] (y12) {$\cvar{12}$}; 
        \node[echnode,above=of y12] (y13) {$\cvar{13}$}; 
        \node[schnode,left=of y8] (y15) {$\cvar{15}$}; 
\draw[silentedge] (y8)  -- ([xshift=0.5cm]y8.east) |- (y9);
        \draw[silentedge] (y11)  edge node [left] {} (y12);
        \draw[silentedge] (y12) -- ([xshift=-0.5cm]y12.west) |- (y13);
        \draw[silentedge] (y15)  edge node [left] {} (y8);
\draw[->] (y9)  edge node [right] {$\rcv{\exOK}$} (y8);
        \draw[->] (y9)  edge node [right] {$\rcv{\exKO}$} (y11);
\draw[->] (y13)  edge node [right] {$\rcv{\exOK}$} (y12);
        \draw[->] (y13)  edge node [right] {$\rcv{\exKO}$} (y15);
      \end{tikzpicture}
    \\\midrule
    $\mathcal{G}$ & $\mathcal{G'}$
  \end{tabular}

  \caption{Input tree equations
    for $\minc{\client}{\clientsup}$ (Figures~\ref{fig:runex-machines} and~\ref{fig:runex-machines2}) and their graphical representations.
The starting variables are $\pvarroot$ and $\cvar{8}$.
Silent choices are diamond-shaped nodes, other nodes are
    rectangles.}
  \label{fig:eqsystems}
\end{figure}

Before giving the definition of a witness subtree, we introduce a
few auxiliary functions on which it relies.
Given a machine $M= (\States, \state_0, \Trans)$, a state $q \in Q$,
and a word $\word \in \Alphas$, we define $\acctree{q}{\word}$ as
follows:
\[
\acctree{q}{\word} = 
\begin{cases}
q & \text{if } \word = \lempty
  \\
  \ctx A[\acctree{q'_i}{\word'}]^{i\in I}
  & \text{if }
  \word = a \cdot \word',
\ctx A[q_i]^{i \in I}\!=\!\mathsf{inTree}(q),
\forall i \!\in\! I.\ q_i \trans{\snd{a}} q'_i
  \\
  \perp & \text{otherwise}
\end{cases}
\]
Function $\acctree{q}{\word}$ is a key ingredient of the
witness subtree definition as it allows for the construction of the accumulation of
receive actions
(represented as an input tree) that is
generated from a state $q$ mimicking the sequence of send
actions sending the messages in $\word$.
We illustrate the usage of $\acctree{q}{\word}$ in
Example~\ref{ex:minacc} below.

We use the auxiliary function $\closesthole{k}{Q'}{\acts}$ below to
ensure that the effect of performing the transitions from an ancestor
to a boundary node is that of increasing (possibly
non-strictly) the accumulated receive actions.
Here, $k$ represents a known lower bound for the length of the sequences
of receive actions accumulated in an input context $\ctx A$,
i.e., a lower bound for $\minHeight{\ctx A}$.
Assuming that the holes in $\ctx A$ contain the states populating
the set of states $Q'$,
the function returns a lower bound for the length of the
sequences of accumulated receive actions
after the transitions in $\acts$ have been executed.
Formally, given a natural number $k$ ($k \geq 0$), a sequence of action
$\acts \in \Acts$, and a set of states
${\{q_j \mid j\in J\}} \subseteq Q$, we define this function as
follows:

\[
\hspace*{-2mm}
\begin{array}{l@{}}
\closesthole{k}{\{q_j \mid j\in J\}}{\acts} = 
\vspace*{+1mm}\\
\begin{cases}
  k
  & 
  \begin{array}{l@{}}
  \text{if } \acts = \lempty \end{array}
  \vspace*{+1mm}\\
  \closesthole{k-1}{\{q_j \mid j\in J\}}{{\acts}'}
  & 
  \begin{array}{l@{}}
  \text{if } \acts =\ ?a \cdot {\acts}' \land k>0 
  \end{array}
    \vspace*{+1mm}\\
\closesthole{k{+}
    \text{min}_{j \in J} \minHeight{\ctx A_j}
}{\{ q_{i,j}\!\! \mid\! j \in J, i \in I_j \}}{{\acts}'} \!\!\!\!\!\!\!
&
  \begin{array}{l@{}}
  \text{if }
  \acts =\ !a \cdot {\acts}' \ \land \! \vspace*{-1mm}\\ 
  \forall j\! \in\! J \ldotp
\acctree{q_j}{a}\!\!=\!\!\ctx A_j[q_{i,j}]^{i \in I_j}
  \end{array}
\vspace*{+1mm}\\
  \perp & 
  \begin{array}{l@{}}
  \text{otherwise}
  \end{array}
\end{cases}
\end{array}
\]

\begin{exa}\label{ex:minacc}
  Consider the transitions from node $n_7$ to $n_9$ in
  Figure~\ref{fig:simtree}.
There are two send actions $\snd{\exLQ}$ and $\snd{\exHQ}$ that
  cannot be directly fired from state $q_2$ which is a receiving state;
  the effect is to accumulate receive actions. 
Such an accumulation is
  computed by
  $\acctree{q_2}{\exLQ \cdot \exHQ} = \langle \exKO: \langle \exKO:
  q_2 , \exOK: q_2 \rangle, \exOK: \langle \exKO: q_2 , \exOK: q_2
  \rangle \rangle$.
For this sequence of transitions, the effect on the (minimal) length
  of the accumulated receive actions can be computed by
  $\closesthole{0}{\{q_2\}}{\snd{\exLQ} \cdot \snd{\exHQ}}=2$; meaning
  that before executing the sequence of transitions
  $\snd{\exLQ} \cdot \snd{\exHQ}$ state $q_2$ has not accumulated
  receive actions in front, while at the end an input context with
  minimal depth 2 is generated as accumulation.
\end{exa}

We now prove a couple of properties of $\closesthole{n}{Q'}{\acts}$.

\begin{prop}\label{lab:lemmaMarioAndGigio}
If $\closesthole{k}{Q'}{\psi}$ is defined, then the following statements hold:
\begin{enumerate}
\item\label{lemmaMG1}
for each $k'\geq k$
we have that
$k-\closesthole{k}{Q'}{\psi}=k'-\closesthole{k'}{Q'}{\psi}$;
\item\label{lemmaMG2}
if $\psi=\psi' \cdot \psi''$ then
$\closesthole{k}{Q'}{\psi' \cdot \psi''}=\closesthole{\closesthole{k}{Q'}{\psi'}}{Q''}{\psi''}$
with\\ $Q''=\bigcup_{q \in Q'}\{q_h \mid h \in H\ \text{s.t.}\ 
      \acctree{q}{\sends{\acts'}} = \ctx A''[q_{h}]^{h\in H}\}$;
\item\label{lemmaMG3}
if $\closesthole{k}{Q''}{\psi}$ is defined for a set of states $Q''$
s.t. $Q' \subseteq Q''$ then 
$\closesthole{k}{Q''}{\psi} \leq \closesthole{k}{Q'}{\psi}$.
\end{enumerate}
A direct consequence of $(1)$ is that: $\closesthole{k'}{Q'}{\psi}-\closesthole{k}{Q'}{\psi}=k'-k \geq 0$.

\end{prop}
\begin{proof}
Item \ref{lemmaMG1} is proved by induction on the length of $\psi$.
If the length of $\psi$ is $0$ then $\psi=\lempty$ and $k-\closesthole{k}{Q'}{\lempty}=
k'-\closesthole{k'}{Q'}{\lempty}=0$. In the inductive case we have two distinct cases:
if $\psi=?a\cdot\psi'$ we have $k > 0$ and
$k-\closesthole{k}{Q'}{?a\cdot\psi'}=k-\closesthole{k-1}{Q'}{\psi'}$ and
$k' > 0$ and
$k'-\closesthole{k'}{Q'}{?a\cdot\psi'}=k'-\closesthole{k'-1}{Q'}{\psi'}$,
and by inductive hypothesis $k-1-\closesthole{k-1}{Q'}{\psi'}=k'-1-\closesthole{k'-1}{Q'}{\psi'}$ hence also  $k-\closesthole{k-1}{Q'}{\psi'}=k'-\closesthole{k'-1}{Q'}{\psi'}$;
if $\psi=!a\cdot\psi'$ we have 
$k-\closesthole{k}{Q'}{!a\cdot\psi'}=k-\closesthole{k+w}{Q''}{\psi'}$ and
$k'-\closesthole{k'}{Q'}{!a\cdot\psi'}=k'-\closesthole{k'+w}{Q''}{\psi'}$ 
for a set of states $Q''$ and a value $w$,
and by inductive hypothesis 
$k+w-\closesthole{k+w}{Q''}{\psi'}=k'+w-\closesthole{k'+w}{Q''}{\psi'}$
hence also
$k-\closesthole{k+w}{Q''}{\psi'}=k'-\closesthole{k'+w}{Q''}{\psi'}$.

Item \ref{lemmaMG2} is proved by induction on the length of $\psi'$.
In the base case, the thesis directly follows from $\closesthole{k}{Q'}{\lempty}=k$
and $\acctree{q}{\lempty} = q$ (hence we have $Q''=Q'$).
In the inductive case we have two distinct cases:
If $\psi'=?a\cdot\psi'''$ we have  ($k >0$ because $\closesthole{k}{Q'}{?a\cdot\psi'''\cdot\psi''}$ is defined) 
$\closesthole{\closesthole{k}{Q'}{?a\cdot\psi'''}}{Q''}{\psi''}=
\closesthole{\closesthole{k-1}{Q'}{\psi'''}}{Q''}{\psi''}$,
but the latter, by applying the inductive hypothesis, coincides with
$\closesthole{k-1}{Q'}{\psi'''\cdot\psi''}=\closesthole{k}{Q'}{?a\cdot\psi'''\cdot\psi''}$.
If $\psi'=!a\cdot\psi'''$ we have
$\closesthole{\closesthole{k}{Q'}{!a\cdot\psi'''}}{Q''}{\psi''}=
\closesthole{\closesthole{k+w}{Q'''}{\psi'''}}{Q''}{\psi''}$,
for the set of states $Q'''$ obtained by allowing the states
in $Q'$ to anticipate $!a$ and a value $w$ that depends on $Q'$ and $a$ (see definition of the $\mathsf{minAcc}$ function);
the latter, by applying the inductive hypothesis, coincides with
$\closesthole{k+w}{Q'''}{\psi'''\cdot\psi''}$
(because $Q''$ is obtained from $Q'''$ by allowing the states in $Q'''$ 
to anticipate the send actions in $\psi'''$), which is equal to $\closesthole{k}{Q'}{!a\cdot\psi'''\cdot\psi''}$
as $Q'''$ and $w$ depend on $Q'$ and $a$ as described above.

Item \ref{lemmaMG3} is proved by induction on the length of $\psi$.
The base case is trivial because 
$\closesthole{k}{Q''}{\lempty} = \closesthole{k}{Q'}{\lempty}=k$.
In the inductive case we have two distinct cases:
If $\psi=?a\cdot\psi'$ we have 
$\closesthole{k}{Q''}{?a\cdot\psi'}=\closesthole{k-1}{Q''}{\psi'}$
with the latter, by inductive hypothesis, that is smaller than or equal to 
$\closesthole{k-1}{Q'}{\psi'}=\closesthole{k}{Q'}{?a\cdot\psi'}$.
If $\psi=!a\cdot\psi'$ we have
$\closesthole{k}{Q''}{!a\cdot\psi'}=\closesthole{k+w}{Q'''}{\psi'}$
for a set of states $Q'''$ obtained from $Q''$ by allowing its states
to anticipate $!a$ and $w$ the corresponding minimal $\mathsf{height}$.
Now, if we allow the states in the smaller (or equal) set $Q'$ to anticipate
$!a$, we obtain a smaller (or equal) set $Q''''$ and a value $w'$ that cannot
be smaller than $w$, hence we can apply the inductive hypothesis 
to obtain the greater or equal value 
$\closesthole{k+w}{Q''''}{\psi'}$,
which is smaller or equal, for item \ref{lemmaMG1} of this Proposition,
than $\closesthole{k+w'}{Q''''}{\psi'}=\closesthole{k}{Q'}{!a\cdot\psi'}$.
\end{proof}

We finally give the definition of witness subtree.

\begin{defi}[Witness Subtree]\label{def:witness}
  Let $M_1=(P,p_0,\delta_1)$ and $M_2=(Q,q_0,\delta_2)$ be two
  communicating machines with
  $\simtree{M_1}{M_2}=(N, n_0, \simtreetrans{},
  \mathcal L, \Act, P\times\mathcal T_Q)$.
A candidate subtree of $\simtree{M_1}{M_2}$ with root
  $r$, boundary $B$, and ancestor function $\anc$,
  is a \emph{witness} if the following holds:
  \begin{enumerate}
\item \label{wt:no-out-loop}
For all $n \in B$, given $\acts$ such that $\ancestor{n} \simtreetrans{\acts} n$,
    we have $| \receives{\acts}| > 0$.

\item For all $n \in \img{\anc}$ and $n' \in \img{\anc} \cup B$
    such that
$n \simtreetrans{\acts} n'$,
$
    \mathcal L(n)=\simpair p \ctx A[q_{i}]^{i\in I}$,
    and
$\mathcal L(n')= \simpair {p'} \ctx A'[q_{j}]^{j\in J}$,
we have that : \begin{enumerate}
    \item \label{wt:set-leaves}
$\forall i \in I \ .\ \{q_h \mid h \in H\ \text{s.t.}\ 
      \acctree{q_i}{\sends{\acts}} = \ctx A''[q_{h}]^{h\in H}\} 
      \subseteq  \{q_j \mid j \in J\}$;

      \item \label{wt:min-height}
if $n' \in B$ then $\closesthole{\minHeight{\ctx A}}{\{q_i\mid i\in I\}}{\acts}
      \geq \minHeight{\ctx A}$.
    \end{enumerate}
  \item  \label{wt:embedding}
    $\mathcal G \sqsubseteq \mathcal G'$ where
    \begin{enumerate}
    \item \label{wt:producer}
      $\mathcal G = (\{X_0\} \cup \{X_{q,n} \mid q \in Q, n \in \subtree{S}{r}{B} \!\setminus\!B \},X_0,\DEF)$ with 
      $\DEF$ defined as follows:
\begin{enumerate}
      \item \label{wt:producer-root}
        $X_0 \eqdef T \{{X_{q,r}}/q \mid q \in  Q\}$, with
        $\mathcal L(r)=\simpair p  T$
      \item \label{wt:producer-node}
        $X_{q,n} \eqdef$\\
        $
        \left \{ \!\!\!
          \begin{array}{ll@{}}
            \schset{ X_{q,tr(n')}}{\exists a.n \simtreetrans{?a}n'} & \text{if}\
                                                                      \exists a.n \simtreetrans{\! ?a \!}
            \\
            \schset{ \ctx A[X_{q'_i,tr(n')}]^{i\in I} \!}
            {\!\exists a.n \simtreetrans{!a}n' \land  
            \intree{q}\!=\!\ctx A[q_i]^{i\in I}\! \land \forall i \!\in\! I.
            q_i \trans{!a} q'_i} \!\!\!\!
                                                                    & \text{otherwise}
          \end{array}
        \right .
        $
\end{enumerate} 
      
    \item \label{wt:consumer}
      $\mathcal G' = (\{Y_{n} \mid n \in \subtree{S}{r}{B} \!\setminus\!B\},Y_r,\DEF')$ with 
      $\DEF'$ defined as follows:
      
      \begin{tabular}{p{8cm}p{4cm}}
        $
        Y_n \eqdef
        \begin{cases}
          \schset{ Y_{tr(n')}}{ n \simtreetrans{!a}n'}
& \text{if } \exists n'. n\simtreetrans{!a} n'
\\
\echset{a:Y_{tr(n')}}{
            n\simtreetrans{?a} n'}
& \text{if }   \exists n'. n\simtreetrans{?a} n'  
        \end{cases}
            $
\end{tabular}
    \end{enumerate}
where $tr(n)=n$, if $n \not\in B$; $tr(n)=\ancestor{n}$, otherwise.
  \end{enumerate}
\end{defi}

\noindent
Condition~\eqref{wt:no-out-loop} requires the existence of a receive
transition between an ancestor and a boundary node. This implies that
if the behaviour beyond the witness subtree is the repetition
of behaviour already observed in the subtree, then 
there cannot be send-only cycles.

\noindent
Condition~\eqref{wt:set-leaves} requires that the transitions from
ancestors to boundary nodes (or to other ancestors) are such that they
include those behaviours that can be computed by the $\mathsf{accTree}$
function. 
We assume that this condition does not hold if
$\acctree{q_i}{\sends{\acts}} = \bot$ for any $i \in I$; hence the
states $q_i$ of $M_2$ in an ancestor are able to mimic all the send
actions performed by $M_1$ along the sequences of transitions in the
witness subtree starting from the considered ancestor.

\noindent
Condition~\eqref{wt:min-height} ensures that 
by repeating transitions from ancestors to 
boundary nodes, the accumulation
of receive actions is, overall, non-decreasing. In other words, the rate at
which accumulation is taking place is higher than the rate at which
the context is reduced by Rule \textsf{(InCtx)}.

\noindent
Condition~\eqref{wt:embedding} checks that
the receive actions that can be accumulated by $M_2$(represented by $\mathcal G$) and 
those that are expected to be actually executed by $M_1$
(represented by $\mathcal G'$) are compatible.
In $\mathcal G$, there is an equation for the root node and for each
pair consisting of a local state in $M_2$ and a node $n$ in the
witness subtree.
The equation for the root node is given in~\eqref{wt:producer-root},
where we simply transform an input context into an input
tree expression.
The other equations are given in~\eqref{wt:producer-node}, where we
use the partial function $\intree{q}$. Each equation
represents what can be accumulated by starting from node
$n$ (focusing on local state $q$).
In $\mathcal G'$, there is an equation for each node $n$ in the
witness subtree, as defined in~\eqref{wt:consumer}
There are two types of equations depending on the type of transitions
outgoing from node $n$. A send transition leads to silent choices,
while receive transitions generate corresponding receive choices.

\begin{exa}
  We have that the candidate subtree rooted at $n_8$ in Figure~\ref{fig:simtree}
  satisfies Definition~\ref{fig:simtree}.
\eqref{wt:no-out-loop} Each path from an ancestor to a boundary node
  includes at least one receive action.
\eqref{wt:set-leaves} For each sequence of transitions from an
  ancestor to a boundary node (or another ancestor) the behaviour of
  the states of $M_2$, as computed by the $\mathsf{accTree}$ function,
  has already been observed.
\eqref{wt:min-height} For each sequence of transitions from an
  ancestor to a boundary node, the rate at which receive actions
  are accumulated is higher than or equal to the rate at which
  they are removed from the accumulation.
\eqref{wt:embedding} The systems of input tree equations
  $\mathcal G$ \eqref{wt:producer} and $\mathcal G'$
  \eqref{wt:consumer} are given in Figure~\ref{fig:eqsystems}, and are
  compatible, see Example~\ref{ex:eqsystem}.

  We now describe how $\mathcal G$ and $\mathcal G'$ 
  (Figure~\ref{fig:eqsystems}) are
  constructed from the witness tree rooted at $n_8$ in
  Figure~\ref{fig:simtree}.
For $\mathcal G$ we have the following equations:
  \begin{itemize}
  \item
    $X_0 \eqdef \echoice{\msg{ok} : X_{q_2,n_8} , \msg{ko} :
      X_{q_2,n_8}}{}$ since the root of the witness tree is $n_8$ and
    its label is
    $q_1 \nincsymbol \echoice{\msg{ok} : q_2, \ \msg{ko} : q_2}{} $.
In Figure~\ref{fig:eqsystems}, we depict this equation as a pair
    of transitions from the node labelled by $X_0$ to the node labelled
    by $X_{q_2,n_8}$
\item
    $X_{q_2,n_8} \eqdef \echoice{\exOK : X_{q_2,n_9} , \ \exKO :
      X_{q_2,n_9}}{}$ since $n_8$ has a unique outgoing \emph{send}
    transition to $n_9$, i.e., $n'$ in Case~\eqref{wt:producer-node}
    of Definition~\ref{def:witness}, and
    $\intree{q_2} = \echoice{\exOK : q_1 , \ \exKO : q_1}{}$ with
    $q_1 \trans{\snd{\exHQ}} q_2$ and $q_1 \trans{\snd{\exLQ}} q_2$ in $\clientsup$
\item
    $X_{q_2,n_9} \eqdef \schoice{X_{q_2,n_8}, \ X_{q_2,n_{10}} }{}$
    since $n_9$ has two \emph{receive} transitions: one to $n_{11}$ (a
    boundary node whose ancestor is $n_8$, i.e., $tr(n_{11})=n_8$) and
    one to $n_{10}$ (which is not boundary node, i.e.,
    $tr(n_{10})=n_{10}$)
\item
    $X_{q_2,n_{10}} \eqdef \echoice{\exOK : X_{q_2,n_{12}} , \ \exKO :
      X_{q_2,n_{12}}}{}$ since $n_{10}$ has a unique outgoing
    \emph{send} transition to $n_{12}$, and
    $\intree{q_2} = \echoice{\exOK : q_1 , \ \exKO : q_1}{}$
\item
    $X_{q_2,n_{12}} \eqdef \echoice{\exOK : X_{q_2,n_{13}} , \ \exKO :
      X_{q_2,n_{13}}}{}$ since $n_{12}$ has a unique outgoing
    \emph{send} transition to $n_{13}$, and
    $\intree{q_2} = \echoice{\exOK : q_1 , \ \exKO : q_1}{}$
\item
    $X_{q_2,n_{13}} \eqdef \schoice{X_{q_2,n_{12}}, \ X_{q_2,n_{15}}
    }{}$ since $n_9$ has two \emph{receive} transitions: one to
    $n_{14}$ (a boundary node whose ancestor is $n_{12}$) and one to
    $n_{15}$ (which is not boundary node)
\item
    $X_{q_2,n_{15}} \eqdef \echoice{\exOK : X_{q_2,n_{8}} , \ \exKO :
      X_{q_2,n_{8}}}{}$ since
    $\intree{q_2} = \echoice{\exOK : q_1 , \ \exKO : q_1}{}$ and
    $n_{15}$ has a unique outgoing \emph{send} transition to $n_{16}$,
    which is a boundary node ($n_{16} \in B$) whose ancestor is $n_8$.
\end{itemize}
  We omit the other equations, e.g., $X_{q_1, n_8}$ as they are
  not reachable from $X_0$.

  For $\mathcal G'$ we have the following equations:
  \begin{itemize}
  \item $Y_{n_8} \eqdef \schoice{Y_{n_9}}{}$ since $n_8$ has a unique
    \emph{send} transition to $n_9$, i.e., $n'$ in Case~\eqref{wt:consumer} of
    Definition~\ref{def:witness}, and $n_9$ is not a boundary node
\item
    $Y_{n_{9}} \eqdef \echoice{\msg{ok}: Y_{n_8}, \ \msg{ko} :
      Y_{n_{10}}}{}$ since $n_{9}$ has two \emph{receive} transitions:
    one to $n_{10}$ which is not a boundary node, and one to $n_{11}$
    which is a boundary node whose ancestor is $n_8$
\item $Y_{n_{10}} \eqdef  \schoice{Y_{n_{12}}}{}$ since $n_{10}$ has a unique
    \emph{send} transition to $n_{12}$ which is not a boundary node
  \item $Y_{n_{12}} \eqdef \schoice{Y_{n_{13}}}{}$ since $n_{12}$ has a unique
    \emph{send} transition to $n_{13}$ which is not a boundary node
\item $Y_{n_{13}} \eqdef  \echoice{\msg{ok}: Y_{n_{12}}, \ \msg{ko} :
      Y_{n_{15}}}{}$ since $n_{13}$ has two \emph{receive} transitions: one to
    $n_{15}$ which is not a boundary node, and one to $n_{14}$ which is a
    boundary node whose ancestor is $n_{12}$
\item $Y_{n_{15}} \eqdef \schoice{Y_{n_{8}}}{}$ since $n_{15}$ has a
    unique \emph{send} transition to $n_{16}$, and $n_{16}$ is a
    boundary node whose ancestor is $n_8$.
  \end{itemize}
\end{exa}

We now prove the main property of the $\closesthole{k}{Q'}{\psi}$
function, i.e., given information $k$ and $Q'$ extracted from an 
ancestor $n$ in a witness subtree, such a function correctly computes 
a lower bound of the length of the input accumulation in a 
node $n'$ reachable from $n$ by executing the 
sequence of actions $\psi$.

\begin{prop}\label{lab:lemmaLowerBound}
Consider a witness subtree with ancestor function $\anc$;
given two nodes of the tree, $n \in \img{\anc}$ and 
$n'$ s.t. $n \simtreetrans{\acts} n'$, with
$\mathcal L(n)=\simpair p \ctx A[q_{i}]^{i\in I}$ and
$\mathcal L(n')=\simpair {p'} \ctx A'[q_{j}]^{j\in J}$,
we have that
$\closesthole{\minHeight{\ctx A}}{\{q_i \mid i \in I\}}{\psi} \leq 
\minHeight{\ctx A'}$.
\end{prop}
\begin{proof}
We prove a more general result proceeding by induction on the length of 
$\acts$, i.e., that
$\closesthole{\minHeight{\ctx A}}{\{q_i \mid i \in I\}}{\psi} \leq 
\minHeight{\ctx A'}$ and $\{q_j \mid j \in J\} \subseteq 
\bigcup_{i \in I}\{q_h \mid 
h \in H\ \text{s.t.}\ \acctree{q_i}{\sends{\acts}}=\ctx{A'''}[q_{h}]^{h\in H}\}$. 

The base case is trivial because, by definition,
$\closesthole{\minHeight{\ctx A}}{\{q_i \mid i \in I\}}{\lempty} =
\minHeight{\ctx A}$ and having $n=n'$ then $\minHeight{\ctx A}=
\minHeight{\ctx A'}$. Moreover,
$\bigcup_{i \in I}\{q_h \mid 
h \in H\ \text{s.t.}\ \acctree{q_i}{\lempty}=\ctx{A'''}[q_{h}]^{h\in H}\}
= \{q_i \mid i \in I\}$ and having $n=n'$ then $\{q_i \mid i \in I\} = 
\{q_j \mid j \in J\}$.

In the inductive case we have either $\psi=\psi' \cdot ?a$ or
$\psi=\psi' \cdot !a$. In both cases we observe that,
by definition of witness subtree,
$\closesthole{\minHeight{\ctx A}}{\{q_i \mid i \in I\}}{\psi}$
is defined as it is defined for a longer sequence of
transitions from $n$ to a boundary node (traversing $n'$).

We first consider $\psi=\psi' \cdot ?a$. 
Let $n''$ be the node
reached after the sequence of transitions $\psi'$, and let
$\mathcal L(n'')=\simpair {p''} \ctx A''[q_{w}]^{w\in W}$.
By inductive hypothesis
we have that $\closesthole{\minHeight{\ctx A}}{\{q_i \mid i \in I\}}{\psi'} \leq
\minHeight{\ctx A''}$ and,
letting $Q''=\bigcup_{i \in I}\{q_h \mid 
h \in H\ \text{s.t.}\ \acctree{q_i}{\sends{\psi'}}=\ctx{A'''}[q_{h}]^{h\in H}\}$,
we also have $\{q_w \mid w \in W\} \subseteq Q''$.
By Proposition \ref{lab:lemmaMarioAndGigio}, item \ref{lemmaMG2},
we have that the following holds:
$\closesthole{\minHeight{\ctx A}}{\{q_i \mid i \in I\}}{\psi'\cdot ?a} =
\closesthole{\closesthole{\minHeight{\ctx A}}{\{q_i \mid i \in I\}}{\psi'}}{Q''}{?a}$.
By definition of $\mathsf{minAcc}$, we also have that 
$\closesthole{\closesthole{\minHeight{\ctx A}}{\{q_i \mid i \in I\}}{\psi'}}{Q''}{?a}=\closesthole{\minHeight{\ctx A}}{\{q_i \mid i \in I\}}{\psi'}-1$.
As a direct consequence of the inductive hypothesis we have 
$\closesthole{\minHeight{\ctx A}}{\{q_i \mid i \in I\}}{\psi'}-1 \leq
\minHeight{\ctx A''}-1$, but we have 
that $\minHeight{\ctx A''}-1 \leq \minHeight{\ctx A'}$,
because the effect of an input transition on the input context is simply that of 
consuming one initial input branching. 
We conclude this case by observing that $\{q_j \mid j \in J\} \subseteq 
\{q_w \mid w \in W\}$ because, as observed above, 
the effect of an input transition on the input context is simply that of 
consuming one initial input branching, without changing the states populating
the leaves of the input tree. On the other hand, the set of
states obtained from the states $q_i$ by anticipating the outputs
in $\psi'\cdot?a$ coincides with the above set $Q''$ because 
only send actions are considered, and $\sends{\psi'}=\sends{\psi'\cdot?a}$.
By inductive hypothesis, we have $\{q_w \mid w \in W\} \subseteq Q''$.

We now consider $\psi=\psi' \cdot !a$. Let $n''$ be the node
reached after the sequence of transitions $\psi'$, and let
$\mathcal L(n'')=\simpair {p''} \ctx A''[q_{w}]^{w\in W}$.
By inductive hypothesis
we have that
\[
  \closesthole{\minHeight{\ctx A}}{\{q_i \mid i \in I\}}{\psi'} \leq
  \minHeight{\ctx A''}
\]
and,
letting
\[Q''=\bigcup_{i \in I}\{q_h \mid 
  h \in H\ \text{s.t.}\ \acctree{q_i}{\psi'}=\ctx{A'''}[q_{h}]^{h\in H}\}
\]
we also have $\{q_w \mid w \in W\} \subseteq Q''$. 
By Proposition \ref{lab:lemmaMarioAndGigio}, item \ref{lemmaMG2},
\[
  \begin{array}{l}
    \closesthole{\minHeight{\ctx A}}{\{q_i \mid i \in I\}}{\psi'\cdot !a}
    \\
    = \closesthole{\closesthole{\minHeight{\ctx A}}{\{q_i \mid i \in I\}}{\psi'}}{Q''}{!a}
  \end{array}
\]
moreover, by definition of $\mathsf{minAcc}$,
we have 
\[
  \begin{array}{l}
    \closesthole{\closesthole{\minHeight{\ctx A}}{\{q_i \mid i \in I\}}{\psi'}}{Q''}{!a}
    \\ =
    \closesthole{\minHeight{\ctx A}}{\{q_i \mid i \in I\}}{\psi'}+z
  \end{array}
\]
with $z$ the minimal depth of the holes in the input tree that are
accumulated by the states in $Q''$ when they anticipate $!a$.
Having $n'' \simtreetrans{!a} n'$, we have that the minimal depth
of the input tree in $n'$, i.e.\ $\minHeight{\ctx A'}$, will increase that of $n''$, i.e.\ $\minHeight{\ctx A''}$, depending 
on new accumulation generated by the anticipation of $!a$, hence
the increase, i.e.\  $\minHeight{\ctx A'}-\minHeight{\ctx A''}$, will be greater than or equal to the minimal depth of the 
holes in the input tree that are accumulated by the states in 
$\{q_w \mid w \in W\}$ when they anticipate $!a$.
Being $\{q_w \mid w \in W\} \subseteq Q''$, we have that 
such an increase will be also greater than or equal to $z$, i.e. $\minHeight{\ctx A'}-\minHeight{\ctx A''} \geq z$.
As a direct consequence of the inductive hypothesis we have 
\[
  \closesthole{\minHeight{\ctx A}}{\{q_i \mid i \in I\}}{\psi'}+z \leq
  \minHeight{\ctx A''}+z \leq
  \minHeight{\ctx A'}
\]
We now consider $Q'''=\bigcup_{i \in I}\{q_h \mid 
h \in H\ \text{s.t.}\ \acctree{q_i}{\psi'\cdot !a}=\ctx{A'''}[q_{h}]^{h\in H}\}$;
we have that the states in $Q'''$ are generated by the states
in $Q''$ when they anticipate the output $!a$. 
The same holds also for $\{q_j \mid j \in J\}$, i.e.,  
the states $q_j$ are generated by the states
in $\{q_w \mid w \in W\}$ when they anticipate the output $!a$.
Having $\{q_w \mid w \in W\} \subseteq Q''$, we also have
$\{q_j \mid j \in J\}\subseteq Q'''$.
\end{proof}

We conclude by proving our main result; given a simulation
tree with a witness
subtree with root $r$, all the branches  in the simulation
tree traversing $r$ are infinite (hence successful).

\begin{restatable}{thm}{thmsoundess}\label{thm:soundess}
  Let $M_1=(P,p_0,\delta_1)$ and $M_2=(Q,q_0,\delta_2)$ be two
  communicating machines
with $\simtree{M_1}{M_2}=(N, n_0, \simtreetrans{}, 
  \mathcal L, \Act, P\times\mathcal T_Q)$.
  If $\simtree{M_1}{M_2}$ has a witness 
  subtree with root $r$ then
  for every node $n \in N$ such that $r \simtreetrans{}^* n$
  there exists $n'$ such that $n \simtreetrans{} n'$.
\end{restatable}
\begin{proof}
Let $B$ be the leaves of the witness subtree rooted in
$r$ (i.e. the witness subtree is $\subtree{S}{r}{B}$). 
If there exists $l \in B$ such that $n \simtreetrans{}^+ l$
the thesis trivially holds. For all other nodes $n$ such that
$r \simtreetrans{}^* n$, there exists $l \in B$ such that 
$l \simtreetrans{}^* n$. 

We now prove by induction on the length of $l \simtreetrans{}^* n$,
with $\mathcal L(n)=\simpair p \ctx A[q_{i}]^{i\in I}$,
that there exist $m,m' \in \subtree{S}{r}{B}\setminus B$,
s.t. 
$m \in \img{\anc}$,
$m \simtreetrans{\acts} m'$,
$\mathcal L(m)=\simpair {p'} \ctx A'[q_{j}]^{j\in J}$,
$\mathcal L(m')=\simpair p \ctx A''[q_{k}]^{k\in K}$
such that:
\begin{itemize}
\item
$\{q_i \mid i\in I\} \subseteq 
\bigcup_{j \in J}\{q_h \mid 
h \in H\ \text{s.t.}\ \acctree{q_j}{\sends{\acts}}=\ctx{A'''}[q_{h}]^{h\in H}\}
$;
\item
$\ctx A[X_{q_{i},{m'}}]^{i\in I} \sqsubseteq Y_{m'}$;
\item
$\minHeight{\ctx A} \geq 
\closesthole{\minHeight{\ctx A'}}{\{q_j\mid j \in J\}}{\acts}$.
\end{itemize} 
\bigskip

The base case is when $n \in B$. In this case, let $m,m' = 
\ancestor{n}$.

The first item follows from the definition
of candidate subtree according to which
$\{q_i \mid i\in I\} \subseteq \{q_j \mid j\in J\}$.

The second item follows from the following
reasoning: we consider $X_0 \sqsubseteq Y_r$
and apply on such pair the following transformations
of the l.h.s.\ $X_0$ and r.h.s.\ $Y_r$.
We consider
the sequence of transitions from 
$r$ to $n$ and proceed as follows.
For each receive transition $o \simtreetrans{?a} o'$ 
we modify the r.h.s.\ by considering $Y_{tr(o')}$ and
the l.h.s.\ by consuming the initial
message $a$ and by replacing each 
variable $X_{q,o}$ (for any $q$) with the variable $X_{q,tr(o')}$ that, inside their corresponding definitions, is present because of the transition $o \simtreetrans{?a} o'$.
For each send transition $o \simtreetrans{!a} o'$
we modify the r.h.s.\ by considering $Y_{tr(o')}$ and
the l.h.s.\ by replacing
each variable with the term that, inside their corresponding definitions, is present because of the transition $o \simtreetrans{!a} o'$.
Since $tr(n)=m'$,
we 
obtain $\ctx A[X_{q_{i},{m'}}]^{i\in I} \sqsubseteq Y_{m'}$.
Notice that the relation $\sqsubseteq$ actually holds because
in the modification of the initial terms $X_0$ and $Y_r$
s.t. $X_0 \sqsubseteq Y_r$ we follow the simulation game formalized in 
the Definition \ref{def:compatibility} of input tree compatibility:
in the case of input transitions $o \simtreetrans{?a} o'$
we consume an initial $a$ in both terms and resolve some silent 
choice in the l.h.s; 
in the case of output transitions $o \simtreetrans{!a} o'$
we resolve the initial silent choice in the r.h.s. while
in the l.h.s. we replace variables with their definition
and resolve the initial silent choice in such definitions.

The third item coincides with proving that 
$\minHeight{\ctx A} \geq \minHeight{\ctx A'}$
because, having $m=m'$, the sequence $\acts$ is empty in the expression
$\closesthole{\minHeight{\ctx A'}}{\{q_j\mid j \in J\}}{\acts}$.
By definition of 
witness subtree, we have that 
$\closesthole{\minHeight{\ctx A'}}{\{q_j\mid j \in J\}}{\acts'}
      \geq \minHeight{\ctx A'}$
for every sequence of transitions $\acts'$ from $m$ to a
boundary node, hence also to $n$. By Proposition \ref{lab:lemmaLowerBound},
if we consider the sequence of transitions $\acts'$ from $\anc(n)$ to $n$,
we have that      
$\minHeight{\ctx A} \geq 
\closesthole{\minHeight{\ctx A'}}{\{q_j\mid j \in J\}}{\acts'}$, 
from which we conclude $\minHeight{\ctx A} \geq \minHeight{\ctx A'}$. 
\bigskip

We now move to the inductive case.
Suppose, by inductive hypothesis, that the above three properties
hold for $n$ s.t. $l \simtreetrans{}^+ n$, and consider $n \simtreetrans{} n'$.
We separate the analysis in two parts, the case in which 
an output action $n \simtreetrans{!a} n'$ is executed, 
and the opposite case in which $n \simtreetrans{?a} n'$.
We have to show that in both cases there exist two 
nodes $m_1,m_2 \in \subtree{S}{r}{B}\setminus B$ 
such that the three properties, defined for $n,m,m'$,
hold also for $n',m_1,m_2$, respectively.
\medskip

We now consider $n \simtreetrans{!a} n'$.
In this case we have that $p \trans{!a} p'$,
hence also $m' \simtreetrans{!a} m''$.

We first consider the case in which $m'' \not \in B$:
in this case we take $m_1=m$ and $m_2=m''$.

The first item holds because; 
by inductive hypothesis we have 
$\{q_i \mid i\in I\} \subseteq \bigcup_{j \in J}\{q_h \mid 
h \in H\ \text{s.t.}\ \acctree{q_j}{\sends{\acts}}=\ctx{A'''}[q_{h}]^{h\in H}\}$;
by Definition \ref{def:witness} of witness subtree, item
\ref{wt:set-leaves}, we have that all the above states $q_h$
can anticipate the output action $!a$ because $\acctree{q_j}{\sends{\acts'}}$
is defined for a sequence of actions $\acts'$, from $m$ to a boundary 
node, that contains $\sends{\acts}\cdot !a$ as a prefix; and the states in 
$\{q_i \mid i\in I\}$ are
modified by the transition $!a$ in the same way as the same
states that are present also in the superset
$\bigcup_{j \in J}\{q_h \mid 
h \in H\ \text{s.t.}\ \acctree{q_j}{\sends{\acts}}=\ctx{A'''}[q_{h}]^{h\in H}\}$ change considering the longer 
sequence $\sends{\acts}\cdot !a$
instead of $\sends{\acts}$ only.

The second item holds because;
by inductive hypothesis we have $\ctx A[X_{q_{i},{m'}}]^{i\in I} \sqsubseteq Y_{m'}$;
the accumulated input tree in $n'$
is obtained by replacing each of the variables in $\ctx A[X_{q_{i},{m'}}]^{i\in I}$ with the term that, inside their corresponding definitions, is present because $m' \simtreetrans{!a} m''$ and because, as observed above, each state $q_{i}$ can anticipate 
the output action $!a$;
the l.h.s. term obtained in this way (by simply replacing variables
with their definition and resolving initial silent choices) continue to 
be in $\sqsubseteq$ relation with $Y_{m'}$ hence also with $Y_{m''}$
which is present in the definition of $Y_{m'}$ because $m' \simtreetrans{!a} m''$.
 
The third item holds because; 
if we take 
$k=\closesthole{\minHeight{\ctx A'}}{\{q_j\mid j \in J\}}{\acts}$,
by inductive hypothesis we have 
$\minHeight{\ctx A} \geq k$;
by Proposition \ref{lab:lemmaMarioAndGigio}, item \ref{lemmaMG2},
we have that
$\closesthole{\minHeight{\ctx A'}}{\{q_j\mid j \in J\}}{\acts\cdot !a}=
\closesthole{k}{Q}{!a}$
with
$Q=\bigcup_{j \in J}\{q_h \mid 
h \in H\ \text{s.t.}\ \acctree{q_j}{\sends{\acts}}=\ctx{A'''}[q_{h}]^{h\in H}\}$
where $J$ is the set of indices of the holes in the input context in the 
label of node $m$;
by inductive hypothesis (first item) we have that $\{q_i \mid i\in I\} \subseteq Q$
where $I$ is the set of indices of the holes in the input context in the 
label of node $n$;
by definition of the $\mathsf{minAcc}$ function the increment
$\closesthole{k}{Q}{!a} - k$
cannot be strictly greater than the increment of 
$\mathsf{minHeight}$ when the transition $!a$ is executed from $n$ to $n'$, 
because $\closesthole{k}{Q}{!a}$ considers the minimal 
accumulation generated by the states $Q$ when anticipating $!a$
and, having $\{q_i \mid i\in I\} \subseteq Q$, such a minimal
accumulation cannot be greater than the accumulation generated
by the states $q_i$ present in the leaves of the input tree 
of $n$. From the inductive hypothesis $\minHeight{\ctx A} \geq k$ we, thus, have that
$\mathsf{minHeight}$ in $n'$ is greater or equal to $\closesthole{k}{Q}{!a}$.

We now consider the case in which $m'' \in B$.
We have two distinct cases:
\begin{enumerate}
\item
$\ancestor{m''} \simtreetrans{}^* m$\\
In this case we take $m_1=m_2=\ancestor{m''}$.
The first item holds because of the same 
arguments considered in the corresponding case for $m''\not\in B$
plus the observation that
$\bigcup_{j \in J}\{q_h \mid 
h \in H\ \text{s.t.}\ \acctree{q_j}{\sends{\acts}\cdot !a}=\ctx{A'''}[q_{h}]^{h\in H}\}$ is a subset of the states in the holes 
of the input context in $m''$ (definition of witness
subtree), which is a subset of the states in the holes
of the input context in $\ancestor{m''}$ (definition
of candidate subtree).
The second item holds for the same argument considered in the case
$m'' \not\in B$
(simply replacing $Y_{m''}$ with $Y_{\ancestor{m''}}$).
The third item holds 
for the following reasons.
By applying the same 
arguments considered in the corresponding case for $m''\not\in B$
we obtain that the new $\mathsf{minHeight}$ in $n'$ is greater
or equal than
$\closesthole{\minHeight{\ctx A'}}{\{q_j\mid j \in J\}}{\acts\cdot !a}$,
where $J$ is the set of indices of the holes in the input context in the 
label of node $m$; hence proving the third item reduces to prove 
that
$\closesthole{\minHeight{\ctx A'}}{\{q_j\mid j \in J\}}{\acts\cdot !a}
\geq
\closesthole{\minHeight{\ctx A_1}}{Q_w}{\acts'}$, with
$\mathcal L(m_1)=\simpair {p_1} \ctx A_1[q_{w}]^{w\in W}$,
$Q_w=\{q_w \mid w \in W\}$ and $\acts'$ corresponding to the sequence
of transitions from $m_1=\anc(m'')$ to $m''$ that traverses $m$, hence
$\acts'=\acts''\cdot \acts\, \cdot \, !a$ (for some $\acts''$). This
is because
$\closesthole{\minHeight{\ctx A_1}}{Q_w}{\acts'} \geq \minHeight{\ctx
  A_1}$ by definition of witness subtree.  By Proposition
\ref{lab:lemmaMarioAndGigio}, item \ref{lemmaMG2}, we have
$\closesthole{\minHeight{\ctx A_1}}{Q_w}{\acts''\cdot \acts \cdot
  !a}=\closesthole{\closesthole{\minHeight{\ctx
      A_1}}{Q_w}{\acts''}}{Q''}{\acts \cdot !a}$ with
$Q''=\bigcup_{q \in Q_w}\{q_h \mid h \in H\ \text{s.t.}\
\acctree{q}{\sends{\acts''}} = \ctx A_1[q_{h}]^{h\in H}\}$; given that
the states $\{q_j\mid j \in J\}$ are generated starting from the
states in $Q_w$ by anticipation of the send actions in the sequence
$\acts''$ we have that $\{q_j\mid j \in J\} \subseteq Q''$;
by Proposition \ref{lab:lemmaMarioAndGigio}, item \ref{lemmaMG3},
we have that 
$\closesthole{\closesthole{\minHeight{\ctx A_1}}{Q_w}{\acts''}}{Q''}{\acts \cdot !a}
\leq$\\
 $\closesthole{\closesthole{\minHeight{\ctx A_1}}{Q_w}{\acts''}}{\{q_j\mid j \in J\}}{\acts \cdot !a}$;
by Proposition \ref{lab:lemmaLowerBound} we have that 
$\closesthole{\minHeight{\ctx A_1}}{Q_w}{\acts''} \leq \minHeight{\ctx A'}$
and as a consequence
of Proposition \ref{lab:lemmaMarioAndGigio}, item \ref{lemmaMG1},
we have that
$\closesthole{\closesthole{\minHeight{\ctx A_1}}{Q_w}{\acts''}}{Q''}{\acts \cdot !a}
\leq
\closesthole{\minHeight{\ctx A'}}{Q''}{\acts \cdot !a}$; finally by Proposition \ref{lab:lemmaMarioAndGigio}, item \ref{lemmaMG3}, we have \\
$\closesthole{\minHeight{\ctx A'}}{Q''}{\acts \cdot !a} \leq
\closesthole{\minHeight{\ctx A'}}{\{q_j\mid j \in J\}}{\acts\cdot !a}$.

\item
$m \simtreetrans{}^+ \ancestor{m''}$\\
In this case we take $m_1=m$ and $m_2=\ancestor{m''}$.
The first item holds because of the same 
arguments considered in the corresponding case 
for $m''\not\in B$
plus the observation (as done in the previous case) that
$\bigcup_{j \in J}\{q_h \mid 
h \in H\ \text{s.t.}\ \acctree{q_j}{\sends{\acts}\cdot !a}=\ctx{A'''}[q_{h}]^{h\in H}\}$ is a subset of the states in the holes 
of the input context in $m''$ (definition of witness
subtree); which is a subset of the states in the holes
of the input context in $\ancestor{m''}$ (definition
of candidate subtree);
which is subset of 
$\bigcup_{j \in J}\{q_h \mid 
h \in H\ \text{s.t.}\ \acctree{q_j}{\sends{\acts''}}=\ctx{A'''}[q_{h}]^{h\in H}\}$ 
with $\acts''$ s.t. $m_1=m \simtreetrans{\acts''} \ancestor{m''}=m_2$.
Notice that the latter subset inclusion holds because the
states in the holes
of the input context in $\ancestor{m''}=m_2$
are generated starting from the
states in $Q_j$ by anticipation of the send actions in the
sequence $\acts''$.
The second item holds for the same arguments considered in the case
$m'' \not\in B$
(simply replacing $Y_{m''}$ with $Y_{\ancestor{m''}}$).
We proceed by 
contraposition to show that the third item also holds. Given
$\mathcal L(m)=\simpair {p'} \ctx A'[q_{j}]^{j\in J}$
and $m \simtreetrans{\acts''} \anc(m'')$, we assume by contraposition
that $\mathsf{minHeight}$ applied
to $n'$ is strictly smaller than 
$\closesthole{\minHeight{\ctx A'}}{\{q_j \mid j\in J\}}{\acts''}$.
In the following we let 
$x=\closesthole{\minHeight{\ctx A'}}{\{q_j \mid j\in J\}}{\acts''}$.
By application of the same arguments as above (case
$m'' \not\in B$, third item), we have that $\mathsf{minHeight}$ applied
to $n'$ should be greater than or equal to 
$\closesthole{\minHeight{\ctx A'}}{\{q_j \mid j\in J\}}{\acts\cdot !a}$,
hence also $x > \closesthole{\minHeight{\ctx A'}}{\{q_j \mid j\in J\}}{\acts\cdot !a}$.
But being $m$ 
above $\anc(m'')$, we have that $\acts''$ is a 
prefix of $\acts$; 
then, by Proposition \ref{lab:lemmaMarioAndGigio}, item \ref{lemmaMG2},
we have 
$\closesthole{\minHeight{\ctx A'}}{\{q_j \mid j\in J\}}{\acts\cdot !a} = 
\closesthole{x}{Q''}{\acts''' \cdot !a}$
with $\acts = \acts'' \cdot \acts'''$ and 
$Q''=\bigcup_{q \in Q_j}\{q_h \mid h \in H\ \text{s.t.}\ 
      \acctree{q}{\sends{\acts'''}} = \ctx A''[q_{h}]^{h\in H}\}$
      where $Q_j = \{q_j \mid j\in J\}$.
So far, we have proved that $x-\closesthole{x}{Q''}{\acts''' \cdot !a}>0$.
We now observe that, by Proposition \ref{lab:lemmaLowerBound},
$x$ is smaller than or equal to $\mathsf{minHeight}$ applied to 
$\anc(m'')$, i.e. assuming $\mathcal L(\anc(m''))=\simpair {p_w} \ctx A_2[q_{w}]^{w\in W}$
and $x'=\minHeight{\ctx A_2}$, we have $x \leq x'$; by Proposition \ref{lab:lemmaMarioAndGigio},
item \ref{lemmaMG2}, we have that also
$x'-\closesthole{x'}{Q''}{\acts''' \cdot !a}>0$,
hence $x' > \closesthole{x'}{Q''}{\acts''' \cdot !a}$.
By definition of witness subtree, given that $m\in \img{\anc}$
and $m \simtreetrans{\acts''} \anc(m'')$, $Q''$ is a (non-strict) subset of
the states $\{q_w \mid w\in W\}$, hence by Proposition
\ref{lab:lemmaMarioAndGigio}, item \ref{lemmaMG3}, we obtain
$\closesthole{x'}{Q''}{\acts''' \cdot !a} \geq
\closesthole{x'}{\{q_w \mid w\in W\}}{\acts''' \cdot !a}$.
By combination of the last two inequations we obtain
$\minHeight{\ctx A_2} > 
\closesthole{\minHeight{\ctx A_2}}{\{q_w \mid w\in W\}}{\acts''' \cdot !a}$
that contradicts the
definition of witness subtree (item 2b).
\end{enumerate}
\medskip

We now consider $n \simtreetrans{?a} n'$.
In this case we have that $p \trans{?a} p'$.
We have that $\ctx A$ cannot be a single hole,
otherwise $\minHeight{\ctx A}=0$, that implies
$\closesthole{\minHeight{\ctx A'}}{\{q_j\mid j \in J\}}{\acts}=0$,
that implies that 
there exists a sequence of transitions
$\acts'$, extending $\acts$ and leading to a boundary node, 
such that $\closesthole{\minHeight{\ctx A'}}{\{q_j\mid j \in J\}}{\acts'}$ is 
undefined, contrary to what definition of witness subtree says.
Hence $\ctx A[X_{q_{i},{m'}}]^{i\in I}$ contains
initially an $a$, that must be mimicked in the
simulation game by $Y_{m'}$. This implies that 
also $m' \simtreetrans{?a} m''$.
We first consider the case in which $m'' \not \in B$:
in this case we take $m_1=m$ and $m_2=m''$.
The first item trivially holds because the set on the left
cannot grow while the set on the right remains unchanged.
The second item trivially holds because 
by inductive hypothesis we have 
$\ctx A[X_{q_{i},{m'}}]^{i\in I} \sqsubseteq Y_{m'}$;
we modify the l.h.s. by consuming the initial inputs,
taking the continuation of $a$, and replacing the
remaining variables $X_{q_{i},{m'}}$
with $X_{q_{i},{m''}}$; as r.h.s. we take $Y_{m''}$.
The relation $\sqsubseteq$ continue to hold as we
follow on step $a$ of the simulation game formalized in 
the Definition \ref{def:compatibility} of input tree compatibility,
and we resolve some input choices in the l.h.s.
The last item holds because the r.h.s. of the inequality reduce by one,
while the l.h.s. cannot reduce by more than one.
We now consider the case in which $m'' \in B$.
There are two distinct cases:
$\ancestor{m''} \simtreetrans{}^* m$
or 
$m \simtreetrans{}^+ \ancestor{m''}$.
These two cases are treated as already done above
for the case $n \simtreetrans{!a} n'$, 
subcase in which $m' \simtreetrans{!a} m''$
and $m'' \in B$.
\bigskip

We can finally prove the thesis
considering $\mathcal L(n)=\simpair p \ctx A[q_{i}]^{i\in I}$.

If $p$ is sending, then $m'$ can perform all 
send actions that $p$ can do. Given any of such send actions $!a$,
by definition of witness subtree we have that
$\acctree{q_j}{\sends{\acts'}}$
is defined for a sequence of actions $\acts'$, from $m$ to a boundary 
node, that contains $\acts\cdot !a$ as a prefix;
hence we have that all the states
$
\bigcup_{j \in J}\{q_h \mid 
h \in H\ \text{s.t.}\ \acctree{q_j}{\sends{\acts}}=\ctx{A'''}[q_{h}]^{h\in H}\}
$
can anticipate $!a$.
Given that
$\{q_i \mid i\in I\} \subseteq 
\bigcup_{j \in J}\{q_h \mid 
h \in H\ \text{s.t.}\ \acctree{q_j}{\sends{\acts}}=\ctx{A'''}[q_{h}]^{h\in H}\}$
we also have that all the states $q_i$ can anticipate $!a$.
The possibility to perform the transition 
$n \simtreetrans{!a} n'_a$
also requires that
$p$ has no infinite loop of send actions, i.e., $\neg\loopio{!}{p}$.
Assume by contraposition that $p$ has such an infinite loop
of send actions.
This means that there exists an infinite sequence of output
transitions in the
witness subtree that starts from the node $m'$ (which
is such that 
$\mathcal L(m')=\simpair p \ctx A''[q_{k}]^{k\in K}$)
reaches a boundary node, and then continues from the
ancestor of such boundary node to another boundary node,
and so on.
Eventually, an ancestor of a reached boundary node 
will be in between the last traversed ancestor and 
such boundary node
(otherwise, we infinitely move strictly upward
in the finite witness subtree,
going from boundary nodes
to ancestors that are always strictly above the last traversed
ancestor).
This contradicts the definition of witness subtree stating
that in all paths from an ancestor $\ancestor{o}$, 
to a corresponding boundary node $o$, there is at least one
receiving transition.

If $p$ is receiving, then $\ctx A$ cannot be a single hole
(see the reasoning above for the case $n \simtreetrans{?a} n'$).
Let $\ctx A=\echoice{a_i:{\mathcal A}_i}{i\in I}$.
Having $\ctx A[X_{q_{i},{m'}}]^{i\in I} \sqsubseteq Y_{m'}$,
we have that (by definition of $Y_{m'}$ and $\sqsubseteq$), 
for every $i \in I$, there exists a transition 
$m' \simtreetrans{?a_i} m'_i$ hence also $p \trans{?a_i} p_i$.
So we can conclude that we have also 
$n \simtreetrans{?a_i} n_i$, for every $i \in I$.
\end{proof}

Hence, we can conclude that if the candidate subtrees of
$\simtree{M_1}{M_2}$ identified following the strategy explained in
Part (2) are also witness subtrees, then we have $\minc{M_1}{M_2}$.

\begin{rem}
  When our algorithm finds a successful leaf, a previously seen label,
  or a witness subtree in each branch then the machines are in the
  subtyping relation.
If an unsuccessful leaf is found (while generating the initial
  finite subtree as described in Part~(2)),
then the machines are \emph{not} in the
  subtyping relation.
In all other cases, the algorithm is unable to give a decisive
  verdict (i.e., the result is \emph{unknown}). 
There are two possible causes for an unknown result: either ($i$) it
  is impossible to extract a forest of candidate subtrees (i.e., there
  are successful leaves below some ancestor) or ($ii$) at least one
  candidate subtree is not a witness (see
  Example~\ref{ex:notcaputred}).
\end{rem}

\begin{figure}[t]
\centering
  \begin{tikzpicture}
    [node distance=0.9cm and 1cm
    ,scale=0.8, every node/.style={transform shape}
    ]
    \node[cnode,notexplo,very thick] (n0) {$\pairconf{q_1}{q_1}$};
\node[cnode,notexplo,below left=of n0,yshift=0.3cm] (n1) {$\pairconf{q_2}{q_2}$};
    \node[cnode,notexplo,left=of n1] (na2) {$\pairconf{q_3}{q_1}$};
    \node[cnode,below=of na2] (na3) {\treeconf{$q_4$}{\btreedone}};
    \node[cnode,right=of na3] (na4) {\treeconf{$q_1$}{\btreetwo}};
\node[cnode,below=of na4] (na7) {\treeconf{$q_2$}{\btreedone}};
    \node[cnode,left=of na7] (na8) {\treeconf{$q_3$}{\btreetwo}};
\node[cnode,below=of na8] (na9) {\treeconf{$q_4$}{\btreedthree}};
\node[cnode,right=of na4] (na5) {$\pairconf{q_5}{q_3}$};
    \node[cnode,right=of na5,double] (na6) {$\pairconf{q_5}{q_3}$};
\node[cnode,right=of na9] (na10) {\treeconf{$q_1$}{\btreedfour}};
    \node[cnode,right=of na10] (na11) {$\pairconf{q_5}{q_3}$};
    \node[cnode,right=of na11,double] (na12) {$\pairconf{q_5}{q_3}$};
    \node[cnode,below=of na10] (na13) {\treeconf{$q_2$}{\btreedthree}};
    \node[cnode, left=of na13] (na14) {\treeconf{$q_3$}{\btreedfour}};
    \node[cnode, below=of na14,double] (na15) {\treeconf{$q_4$}{\btreedfive}};
\node[cnode,notexplo,below right=of n0,yshift=0.3cm] (n2) {$\pairconf{q_5}{q_3}$};
    \node[cnode,notexplo,right=of n2,double] (nb2) {$\pairconf{q_5}{q_3}$};
\draw[lhookedsilentedge] (n0)  -| node [above, near start] {$\rcv{a}$} (n1);
    \draw[hookedsilentedge] (n0) -| node [above, near start] {$\rcv{b}$} (n2);
    \draw[hookedsilentedge] (n2)  edge node [above] {$\snd{x}$} (nb2);
    \draw[lhookedsilentedge] (n1)  edge node [above] {$\snd{x}$} (na2);
    \draw[hookedsilentedge] (na2)  edge node [left] {$\snd{x}$} (na3);
    \draw[hookedsilentedge] (na3)  edge node [above] {$\snd{x}$} (na4);
    \draw[hookedsilentedge] (na4)  edge node [above] {$\rcv{b}$} (na5);
    \draw[hookedsilentedge] (na5)  edge node [above] {$\snd{x}$} (na6);        
    \draw[hookedsilentedge] (na4)  edge node [left] {$\rcv{a}$} (na7);
    \draw[lhookedsilentedge] (na7) edge node [above] {$\snd{x}$} (na8);
    \draw[hookedsilentedge] (na8)  edge node [left] {$\snd{x}$} (na9);
    \draw[hookedsilentedge] (na9)  edge node [above] {$\snd{x}$} (na10);
    \draw[hookedsilentedge] (na10)  edge node [above] {$\rcv{b}$} (na11);
    \draw[hookedsilentedge] (na11)  edge node [above] {$\snd{x}$} (na12);
    \draw[hookedsilentedge] (na10)  edge node [left] {$\rcv{a}$} (na13);
    \draw[lhookedsilentedge] (na13) edge node [above] {$\snd{x}$} (na14);
    \draw[hookedsilentedge] (na14)  edge node [left] {$\snd{x}$} (na15);    
    \begin{scope}[every edge/.append style={shorten <=3pt,shorten >=4pt}]
\path[ancestor] (nb2.south) edge [bend left] node {} (n2.south);
    \path[ancestor] (na6.south) edge [bend left] node {} (na5.south);
    \path[ancestor] (na12.south) edge [bend left] node {} (na11.south);
    \end{scope}
    \path[ancestor] (na15.west) edge [bend left=60,shorten <=3pt] node {} (na3.west);
\node[nlabel] at (n0.south west) {$n_{0}$};
  \node[nlabel] at (n1.south west) {$n_{1}$};
  \node[nlabel] at (n2.south west) {$n_{2}$};
  \node[nlabel] at (nb2.south west) {$n_{3}$};
  \node[nlabel] at (na2.south west) {$n_{4}$};
  \node[nlabel] at (na3.south west) {$n_{5}$};
  \node[nlabel] at (na4.south west) {$n_{6}$};
  \node[nlabel] at (na5.south west) {$n_{7}$};
  \node[nlabel] at (na6.south west) {$n_{8}$};
  \node[nlabel] at (na7.south west) {$n_{9}$};
  \node[nlabel] at (na8.south west) {$n_{10}$};
  \node[nlabel] at (na9.south west) {$n_{11}$};
  \node[nlabel] at (na10.south west) {$n_{12}$};
  \node[nlabel] at (na11.south west) {$n_{13}$};
  \node[nlabel] at (na12.south west) {$n_{14}$};
  \node[nlabel] at (na13.south west) {$n_{15}$};
  \node[nlabel] at (na14.south west) {$n_{16}$};
  \node[nlabel] at (na15.south west) {$n_{17}$};
\end{tikzpicture}
\caption{Simulation tree of Example~\ref{ex:notcaputred}.
}\label{fig:bad-simtree}
\end{figure}
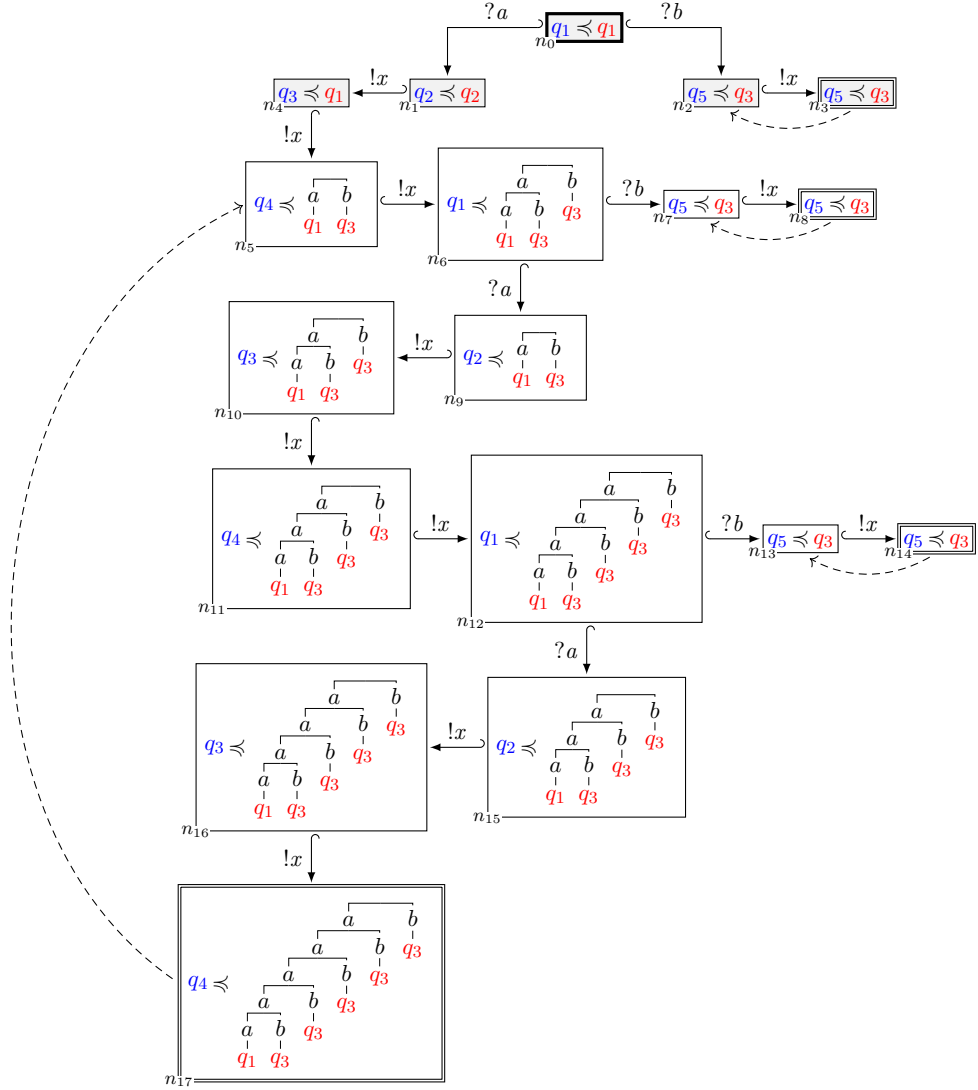

\begin{exa}\label{ex:notcaputred}
  Consider the machines $M_1$ and $M_2$ below:
  \begin{center}
  \begin{tabular}{l@{\quad}r@{\qquad\qquad}l@{\quad}r}
    $M_1$: & 
             \begin{tikzpicture}[mycfsm, node distance = 0.4cm and 1cm
               ,scale=0.95, every node/.style={transform shape}]
               \node[state, initial, initial where=above] (s1) {$q_1$};
               \node[state, left=of s1] (s2) {$q_2$};
               \node[state, right=of s1] (s5) {$q_5$};
               \node[state, below=of s1] (s4) {$q_4$};
               \node[state, below=of s2] (s3) {$q_3$};
\path 
               (s1) edge node [above] {$\rcv{a}$} (s2)
               (s2) edge [bend right] node [left] {$\snd{x}$} (s3)
               (s1) edge node {$\rcv{b}$} (s5)
               (s3) edge node [above] {$\snd{x}$} (s4)
               (s4) edge [bend right] node [right] {$\snd{x}$} (s1)
               (s5) edge [loop right] node {$\snd{x}$} (s5)
               ;
             \end{tikzpicture}
    &
    $M_2$: &
             \begin{tikzpicture}[mycfsm, node distance = 0.8cm and 1cm
               ,scale=0.95, every node/.style={transform shape}]
               \node[state, initial, initial where=above] (s1) {$q_1$};
               \node[state, left=of s1] (s2) {$q_2$};
               \node[state, right=of s1] (s3) {$q_3$};
\path 
               (s1) edge [bend right] node [above] {$\rcv{a}$} (s2)
               (s2) edge [bend right] node [below] {$\snd{x}$} (s1)
               (s1) edge node {$\rcv{b}$} (s3)
               (s3) edge [loop right] node {$\snd{x}$} (s3)
               ;
             \end{tikzpicture}
  \end{tabular}
\end{center}
 The simulation tree $\simtree{M_1}{M_2}$, whose initial part is
  given in Figure~\ref{fig:bad-simtree}, contains infinitely many
  nodes with labels of the form:
$ \simpair{q_1} {\echoice{a: \echoice{ a: \echoice{a:
          \echoice{\cdots}{} , b: q_3 }{}, b: q_3}{} {} , b: q_3 }{}}
  $ (e.g., $n_6$ and $n_{12}$ in Figure~\ref{fig:bad-simtree}).
Each of these nodes has two successors, one where $?a$ is fired (the
  machines stay in their larger loops), and one where $?b$ is fired (the
  machines move to their self loops).
The machines can always enter this send-only cycle, e.g., between
  $n_2$ and $n_3$ or between $n_{13}$ and $n_{14}$.
Because of these \emph{send only} paths between ancestors (e.g.,
  $n_2$) and leaves (e.g., $n_3$), Condition~\eqref{wt:no-out-loop} of
  Definition~\ref{def:witness} never applies on the infinite branches
  of $\simtree{M_1}{M_2}$, hence no witness subtrees can be found.
Note however that our approach successfully identifies a candidate
  subtree, i.e., the white nodes in Figure~\ref{fig:bad-simtree}.
\end{exa}

\section{Implementation and evaluation}\label{sec:tool}
To evaluate the applicability and cost of our algorithm, we have
produced a faithful implementation of it, which is freely available on
GitHub~\cite{asyncsubtool}.

\subsubsection*{Implementation}
The tool is implemented in Haskell and it mostly follows the structure
of \S~\ref{sec:algorithm}.
(1) It takes two machines $M_1$ and $M_2$ as input for which it builds
a simulation tree following Definition~\ref{def:simtree} in a
depth-first search manner, while recording the nodes visited in
different branches to avoid re-computing several times the same
subtrees.
The function terminates whenever it expands a node whose label has
been seen along the path from the root; or whenever it expands a node
which has two ancestors that validate the termination condition from
Theorem~\ref{thm:termination}.
The resulting tree is then passed onto the next function.
(2) The next function divides the finite tree into several (finite)
subtrees following the strategy outlined on
page~\pageref{rem:divide-simtree}.
(3) A third function analyses each subtree to verify that they
validate conditions~\eqref{wt:no-out-loop}-\eqref{wt:min-height} of
Definition~\ref{def:witness}.
(4) Finally, for those subtrees that validate the property checked in
(3), the tool builds their systems of input tree equations and checks
whether they validate the compatibility condition from
Definition~\ref{def:compatibility}.

In function (1), if the tool finds a node for which none of the rules of
Definition~\ref{def:simtree} apply, then it says that the two types
are \emph{not} related.
If each subtree identified in (2) corresponds to branches that loop or
that lead to a witness tree, then the tool says that the input types
are in the subtyping relation.
In all other cases, the result is still \emph{unknown}, hence
the tool checks for $\dual{M_2} \nincsymbol \dual{M_1}$ (relying on a
previous result showing that
${M_1} \nincsymbol {M_2} \iff \dual{M_2} \nincsymbol \dual{M_1}$
\cite{LangeY17, BravettiCZ17}).
Once this pass terminates, the tool returns \emph{true} or
\emph{false}, accordingly, otherwise the result is \emph{unknown}.

For debugging and illustration purposes, the tool can optionally
generate graphical representations of the simulation and candidate
trees, as well as the systems of input tree equations.

\subsubsection*{Evaluation}
We have run our tool on 174 tests which were either taken from the
literature on asynchronous subtyping~\cite{LangeY17, LMCSasync}, or
handcrafted to test the limits of our approach.
All of these tests terminate under a second.
Out of these tests, 92 are \emph{negative} (the types are not in the
subtyping relation) and our tool gives the expected result (``false'')
for all of them.
The other 82 tests are \emph{positive} (the types are in the subtyping
relation) and our tool gives the expected result (``true'') for all
but 8 tests, for which it returns ``unknown''.
All of these 8 examples feature complex accumulation patterns, that our
theory cannot recognise.
Example~\ref{ex:notcaputred} gives a pair of machines for which our
tool returns ``unknown'' for both ${M_1} \nincsymbol {M_2}$ and
$\dual{M_2} \nincsymbol \dual{M_1}$.

To assess the cost of our approach in terms of computation time and
memory consumption, we have automatically generated a series of pairs
of communicating machines that are successfully identified by our
algorithm to be in the asynchronous subtyping relation.
Our benchmarks consists in applying our algorithm to check that
$M_1 \nincsymbol M_2$ holds, with $M_1$ and $M_2$ as specified below,
where $n, m \in \naturals_{>0}$ are the parameters of our experiments.
\[
  \begin{array}{ccccc}
    M_1:
    \begin{tikzpicture}[mycfsm, node distance = 0.5cm and 1.3cm
      ,scale=0.99, every node/.style={transform shape}]
      \node[state, initial, initial where=left] (s0) {};
      \node[state, right =of s0] (s1) {};
      \node[state, right =of s1] (s2) {};
      \node[state, right =of s2] (s3) {};
\draw[>=,densely dotted]   (s1) -- (s2) ;
\path 
      (s0) edge node [above] {$\snd{a_0}$} (s1)
      (s2) edge node [above] {$\snd{a_n}$} (s3)
      (s3) edge [bend left=40] node [below] {$\{ \rcv{b_0}, \ldots, \rcv{b_m} \}$} (s0)
      ;           
    \end{tikzpicture}
    & \qquad\qquad  & M_2:
                      \begin{tikzpicture}[mycfsm, node distance = 0.5cm and 1.3cm
                        ,scale=0.99, every node/.style={transform shape}]
                        \node[state, initial, initial where=left] (s0) {};
                        \node[state, right =of s0] (s1) {};
\path 
                        (s0) edge [bend left=70] node [above] {$\{ \snd{a_0}, \ldots,  \snd{a_n} \}$} (s1)
                        (s1) edge [bend left=70] node [below] {$\{ \rcv{b_0}, \ldots, \rcv{b_m}\} $} (s0)
                        ;
                      \end{tikzpicture}
  \end{array}
\]
Machine $M_1$ {sends} a sequence of $n$ message $a_i$, after which it
expects to {receive} a message from the alphabet
$\{b_0, \ldots, b_m\}$, then returns to its initial state.
Machine $M_2$ can choose to send any message in
$\{a_0, \ldots, a_n\}$, then waits for a message in
$\{b_0, \ldots, b_m\}$ before returning to its initial state.
Observe that for any $n$ and $m$ we have that $M_1 \nincsymbol M_2$
holds.
The shape of these machines allows us to assess how our approach fares
in two interesting cases: when the sequence of message accumulation
grows (i.e., $n$ grows) and when the number of possible branches grows
(i.e., $m$ grows).
Accordingly, we ran two series of benchmarks. 
The plots in Figure~\ref{fig:plots} gives the time taken for our tool
to terminate and the maximum amount of memory used during its
execution (left and right $y$ axis, respectively) with respect to the
parameter $n$ (left-hand side plot) or $m$ (right-hand side plot).
The top plots use linear scales for both axes, while the bottom plots
show the same data but using a logarithm scale for the $y$ axis.

Observe that the left-hand side plot depicts a much steeper curve for
computational time than the one of the right.
Indeed, the depth of the finite subtree that needs to be computed and
analysed increases with $n$ (the depth of the finite subtree is $2n+5$
when $m=1$).
Accordingly, the depth of the input contexts that need to be recorded
increases similarly ($2n+1$). Each input context node has two children
in this case, i.e.,
$\echoice{b_0: \echoice{\ldots}{} , b_1: \echoice{\ldots}{} }{}$.

In contrast, when $m$ increases the depth of the simulation tree is
bounded at 11. Consequently, the sizes of the finite subtrees are
stable (depth of $7$ when $n=1$) but the number of (identical)
candidate subtrees that need to be analysed increases, i.e., the tool
produces $m{+}1$ trees when $n{=}1$.
In this case the maximum depth of input contexts is also stable (the
maximum depth is $3$) but their widths increase with $m$, i.e., we
have input context of the form:
$\echoice{b_0: \echoice{\ldots}{} , \ldots ,  b_m: \echoice{\ldots}{} }{}$.
These observations suggest that our algorithm is better suited to deal
with session types that feature few anticipation steps (smaller $n$),
but performs relatively well with types that contain many branches
(larger $m$).

The left-hand side plots show that the memory consumption follows a
similar exponential growth to the computational time,
unsurprisingly. For instance, our tool needs ~2GB to check a pair
machines where $n=10$ and $m=1$, and ~8 GB when $n=11$ and $m=1$.
The right-hand side plots show a much smaller memory footprint when
$m$ increases, this is explained by the fact that the depth of the
simulation tree is bounded, only the input context of its nodes are
growing in width. The memory in this case is more reasonable, e.g.,
our tool needs less than 11MB to check a pair of machines where $n=1$
and $m=19$.
We suspect the several jumps in the memory usage curve are due to the GHC
runtime requesting new arenas of memory from the operating system.

All the benchmarks in this paper were run on an 8-core Intel i7-7700
machine with 16GB RAM running a 64-bit Linux.
The time was measured by taking the difference between the system
clock before and after our tool was invoked.
The memory usage refers to the \emph{maximum resident set size} as
reported by the \texttt{/usr/bin/time -v} command.
Each test was ran 5 times, the plots report the average time (resp.\
memory) measurements.
All our test data and infrastructure are available on our GitHub
repository~\cite{asyncsubtool}.

\begin{figure*}[t]
  \centering
\begin{tikzpicture}
    [node distance = 0cm and 0cm]
\node (plota) {\includegraphics[width=0.49\textwidth]{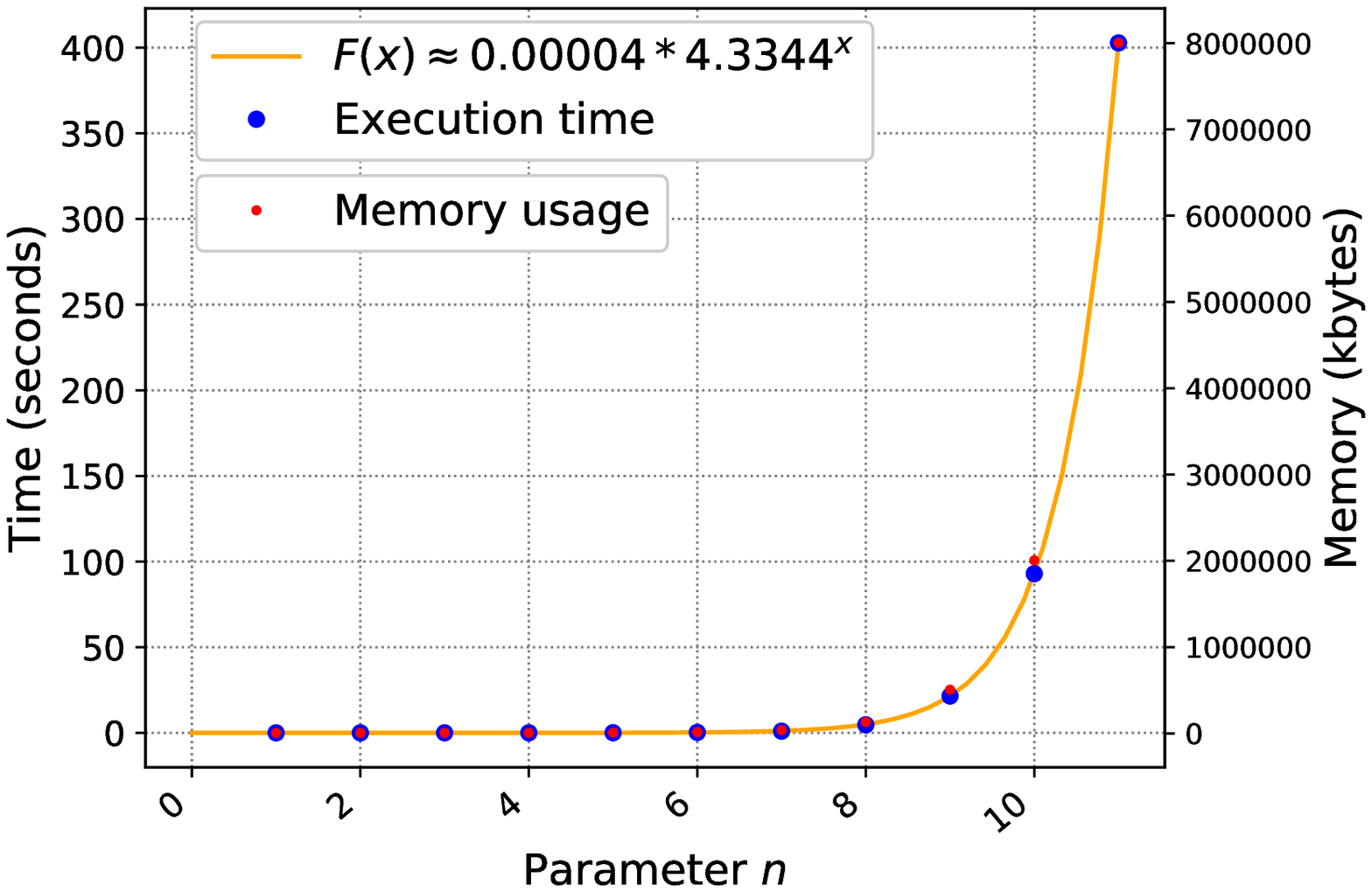}};
\node[right=of plota] (plotb) {\includegraphics[width=0.49\textwidth]{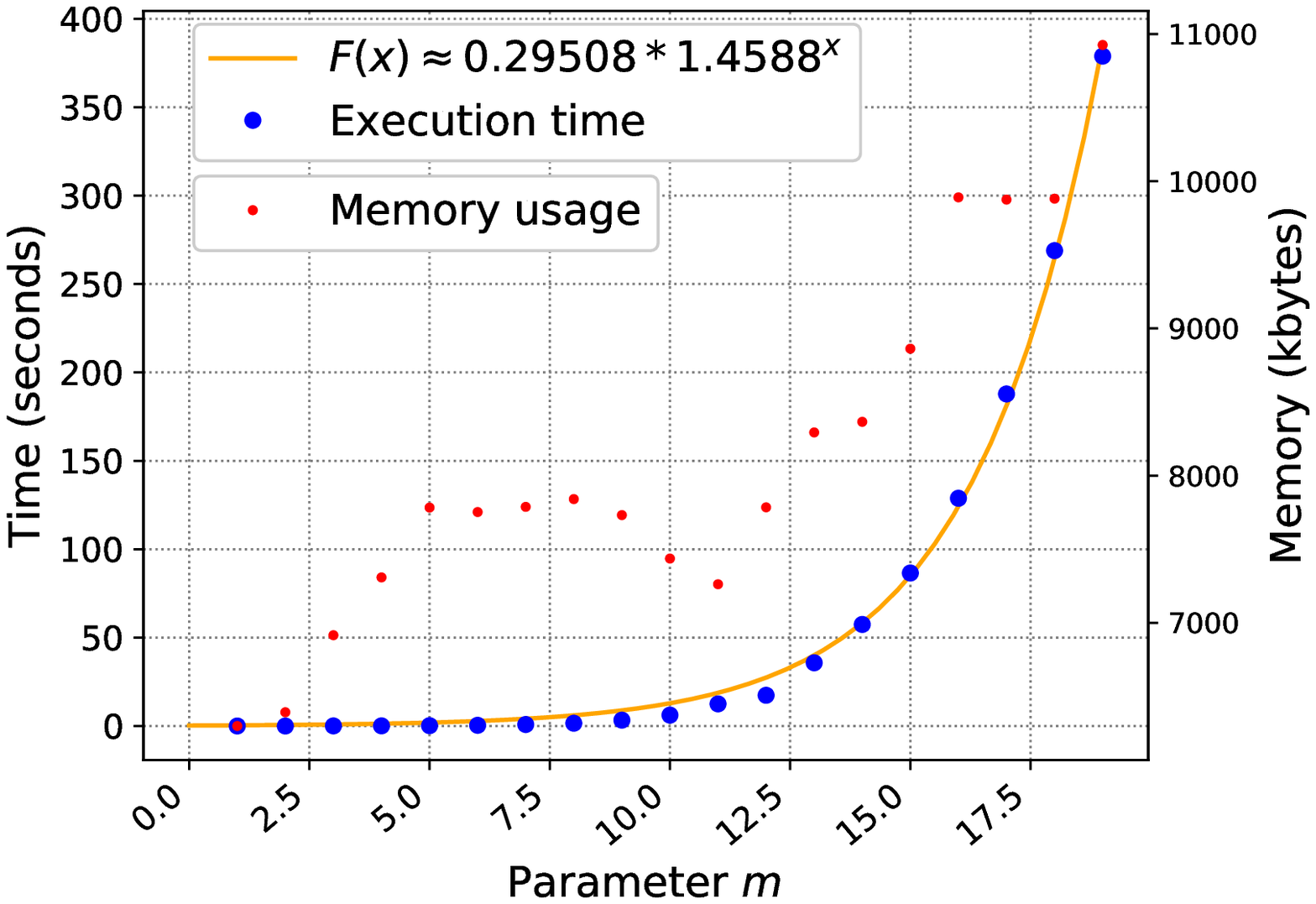}};
\node[below=of plota] (plotalog) {\includegraphics[width=0.49\textwidth]{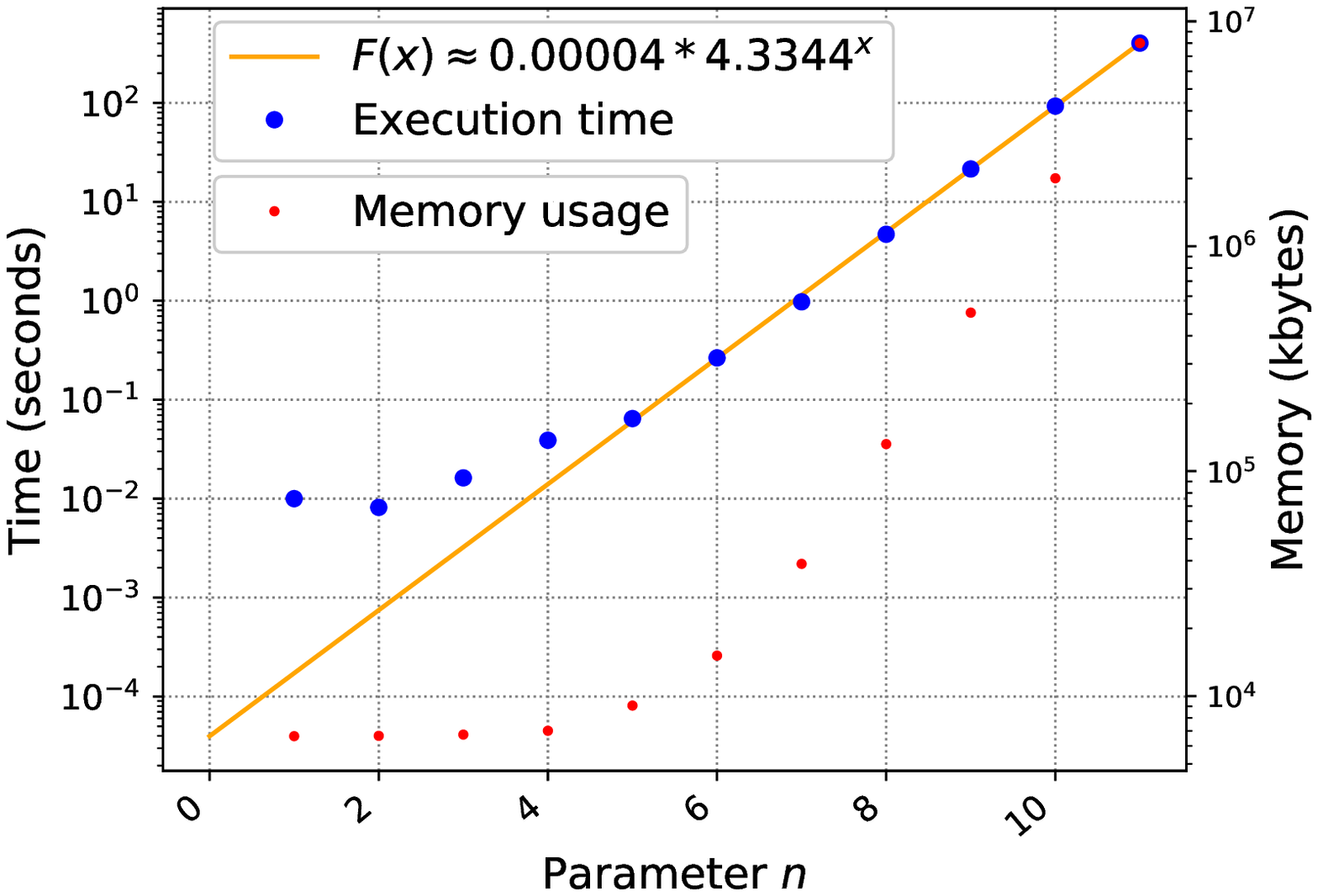}};
    \node[right=of plotalog] (plotblog) {\includegraphics[width=0.49\textwidth]{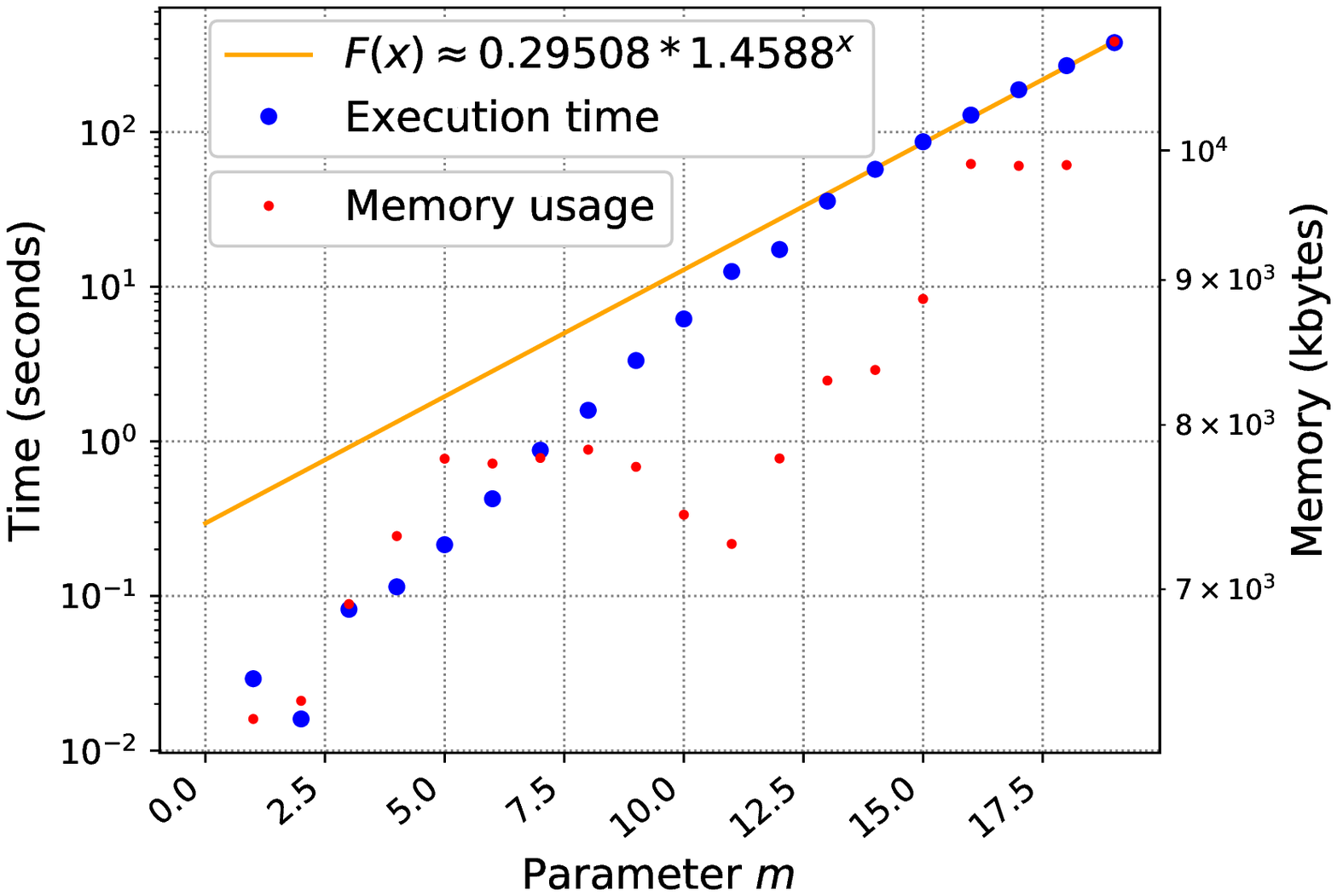}};
    
  \end{tikzpicture}
\caption{Benchmarks: $m{=}1$ and increasing $n$ (left) and $n{=}1$ and increasing $m$ (right).
Top and bottom plot show the same data, but 
    the top plots use linear scales for all axes, the bottom plots use
    logarithmic scales for the vertical axes.
  }
  \label{fig:plots}
\end{figure*}

\section{Related Work}
\label{sec:related}
Gay and Hole~\cite{GH99,GH05} were the first to introduce subtyping
for session types. Their definition, called {\em synchronous
  subtyping}, focuses on the possibility for a subtype to have different
  sets of labels in selections and branchings. In that paper, input selection
  is covariant (the subtype can have less inputs) while output branching
  is contravariant (the subtype can have more outputs). In our formulation
  of subtyping we have the opposite (branchings are covariant and
  selections are contravariant)
  because we follow a process-oriented
  interpretation of session types, while Gay and Hole~\cite{GH99,GH05}
  followed a channel-oriented interpretation.

Later,
Mostrous et al.~\cite{MostrousESOP09} extended such notion to {\em
  asynchronous subtyping}, by allowing for delayed inputs.
Chen et al.~\cite{CDY2014,LMCSasync} subsequently provided an
alternative definition which prohibits {\em orphan messages} and which
is the definition we adopted in this work. Recently, asynchronous
subtyping was shown to be undecidable by reducing it to an equivalent
problem for Turing machines~\cite{LangeY17} and queue
machines~\cite{BravettiCZ17}.

Our previous work~\cite{LangeY17,BravettiCZ17,BCZ18} investigated
different restrictions to achieve decidability: in all of our previous
approaches, these restrictions are either ($i$) setting bounds on the
number of pending messages in the FIFO channels, or ($ii$) restricting
the syntax of communicating machines and session types.
Lange and Yoshida~\cite[\S~4]{LangeY17} identified two subclasses of
(two-party) communicating machines for which the asynchronous
subtyping relation is decidable via syntactical restrictions:
\emph{alternating machines} and \emph{non-branching machines}.
Alternating machines were introduced by Gouda et al.~\cite{GoudaMY84}
and require that each sending transition is followed by a receiving
transition.
A consequence of this restriction is that each FIFO queue may contain
at most one pending message, i.e., it enforces a form of $1$-bounded
asynchrony.
Non-branching machines enforce a syntactical restriction such that
each state has at most one outgoing transition, i.e., given
$M = (Q, q_0, \delta)$, \ $\forall q\in Q \ : \ |\delta(q)| \leq 1$.
Bravetti et al.~\cite{BravettiCZ17,BCZ18} investigate other decidable
fragments of asynchronous subtyping.
In contrast with the present work and those by Lange and Yoshida, they
take a direct syntactical approach, i.e., they work directly on the
syntax of (binary) session types rather than communicating machines.
Chronologically, their first article~\cite{BravettiCZ17}
proves the undecidability of asynchronous session subtyping
in a restricted setting in which subtypes are non-branching (see
definition above) in all output selections and supertypes are
non-branching in all their input branchings. Then, a decidability 
result is proved for a fragment in which they additionally impose 
that the subtype is also non-branching in input branchings 
(or that the supertype is also non-branching in output selections).
Later,
in~\cite{BCZ18}, the same authors consider more fragments, namely
$k$-bounded asynchronous subtyping (bound on the size of
input-anticipations), and two syntactical restrictions that
imposes non-branching only on outputs (resp. inputs).
More formally, following the automata notation, they restrict
to machines $M$ s.t.\
$M = (Q, q_0, \delta)$, \ $\forall q\in Q' \ : \ |\delta(q)| \leq 1$
with $Q'$ coinciding with the set of sending (resp. receiving)
states of $Q$.
All such fragments
are shown to be decidable.

In~\cite{BLZ21}, Bravetti et al.\ propose a \emph{fair} variant of
asynchronous session subtyping.
This fair subtyping handles candidate subtypes that may
simultaneously send each other a finite but unspecified amount of
messages before removing them from their respective buffers.
Such types are not supported by the relation studied here, notably due
the finiteness of input contexts $\mathcal{A}$ and the
$\neg\loopio{!}{\statea}$ condition in Definition~\ref{def:inclusion}
(\ref{it:def-inc-alpha-async}).
This fair subtyping is shown to be undecidable, but a sound
algorithm and its implementation are given in~\cite{BLZ21}.

The relationship between communicating machines and {\em binary} 
asynchronous session types
has been studied in~\cite{BravettiZ21}, where a correspondence result
between asynchronous session subtyping and asynchronous
machine refinement is established.
On the other hand, the relationship between communicating machines and 
{\em multiparty} asynchronous session types has been studied 
in~\cite{DenielouY12,DY13}.
Communicating machines are Turing-powerful, hence their
properties are generally undecidable~\cite{cfsm83}. Many variations have been
introduced in order to recover decidability, e.g., using (existential
or universal) bounds~\cite{GenestKM07}, restricting to different types
of topologies~\cite{TorreMP08,PengP92}, or using bag or lossy channels
instead of FIFO queues~\cite{ClementeHS14, CeceFI96,AbdullaJ93,
  AbdullaBJ98}.
In this context, existentially bounded communicating
machines~\cite{GenestKM07} are one of the most interesting sub-classes
because they include infinite state systems.
However, deciding whether communicating machines are existentially
bounded is generally undecidable.
Lange and Yoshida~\cite{LY2019} proposed a (decidable) property that
soundly characterises existential boundedness on communicating
machines corresponding to session types.
This property, called \emph{$k$-multiparty compatibility} ($k$\MC),
also guarantees that the machines validate the \emph{safety} property
of session types~\cite{DY13,LTY15}, i.e., all messages that are sent
are eventually received and no machine can get permanently stuck
waiting for a message.
This notion of safety is closely related to asynchronous session
subtyping for two-party communicating machines, i.e., we have that
$M_1 \nincsymbol \dual{M_2}$ implies that the system $M_1 \mid M_2$ is
safe~\cite{LangeY17,LMCSasync}.
Because the present work is restricted to two-party systems, our
algorithm cannot be used to verify the safety of multiparty protocols,
e.g., the protocol modelling the double-buffering
algorithm~\cite{MostrousESOP09} is $2$-multiparty compatible but
cannot be verified with our subtyping algorithm because it involves
three parties.
This algorithm is used in multicore systems~\cite{Sancho08} and can be
type-checked up-to asynchronous subtyping~\cite{MostrousESOP09}.
An extension of our work to support multiparty protocols is being
considered, see \S~\ref{sec:conclusions}.
We note that because the $k$\MC\ property of~\cite{LY2019} is based on
a bounded analysis, it cannot guarantee the safety of systems that
exhibit an intrinsically unbounded behaviour, like machines $\client$
and $\server$ in Figures~\ref{fig:runex-machines}
and~\ref{fig:runex-machines2}.

\section{Conclusions and Future Work}
\label{sec:conclusions}
We have proposed a sound algorithm for checking asynchronous session
subtyping, 
showing that it is still possible to decide whether
two types are related for many nontrivial examples.
Our algorithm is
based on a (potentially infinite) tree representation of the
coinductive definition of asynchronous subtyping; it
checks for the
presence of finite witnesses of infinite successful subtrees.
We have provided an implementation and applied it to
examples that cannot be recognised by previous approaches.

Although the (worst-case) complexity of our algorithm is rather high
(the termination condition expects to encounter a set of states
already encountered, of which there may be exponentially many), our
implementation shows that it actually terminates under a second for
machines of size comparable to typical communication protocols used in
real programs, e.g., Go programs feature between three and four
communication primitives per channel and whose branching construct
feature two branches, on average~\cite{dilley19}.

As future work, we plan to enrich our algorithm to recognise subtypes
featuring more complex accumulation patterns, e.g.,
Example~\ref{ex:notcaputred}.
Moreover, due to the tight correspondence
with safety of communicating machines~\cite{LangeY17}, we plan to
investigate the possibility of using our approach to characterise a
novel decidable subclass of communicating machines.
It is an interesting open question to extend our algorithm   
to multiparty communications, 
as multiparty session types 
allow more permutations of actions inside a single CFSM and 
can type more practical use cases which involve several participants. 
Recently \emph{precise} multiparty asynchronous 
subtyping (in the sense of \cite{CDY2014,LMCSasync,GJPSY2019}) for the 
asynchronous multiparty session $\pi$-calculus \cite{HYC08,HondaYC16}
was proposed in \cite{GPPSY21}. In another direction of future work we will consider
an algorithm for checking subtyping which
is sound, but not complete with respect to~\cite{GPPSY21}.  Finally, a significant further extension could be to also encompass pre-emption mechanisms, see e.g.~\cite{BravettiZ09,Bravetti21}, which are often used in communication protocols.

\section*{Acknowledgement}
The work is partially funded by 
H2020-MSCA-RISE Project 778233 
(BEHAPI); and 
EPSRC EP/K034413/1, EP/K011715/1,
   EP/L00058X/1, EP/N027833/1, EP/N028201/1, EP/T014709/1,
   EP/V000462/1, EP/T006544/1 and NCSC VeTSS.

\bibliographystyle{alpha} \bibliography{async}

\end{document}